%% file: main.tex
\documentclass[11pt, letterpaper]{article}
\usepackage[utf8]{inputenc}
\usepackage{fullpage,times}

\usepackage{amsmath,amsfonts,amsthm,amssymb,multirow}
\usepackage{thm-restate}
\usepackage{times}
\usepackage{graphicx}
\usepackage{floatpag}
\usepackage{algorithm}
\usepackage[noend]{algpseudocode}
\usepackage{bbold}
\usepackage{enumitem}
\usepackage{subcaption} 
\usepackage{calc}
\usepackage{tikz}
\usetikzlibrary{decorations.markings,decorations.pathreplacing,arrows, automata, positioning, quotes, shapes,backgrounds, arrows.meta}
\tikzstyle{vertex}=[circle, draw, inner sep=0pt, minimum size=4pt, fill = black]

\usepackage{graphicx}
\usepackage{tabularx}
\usepackage[colorlinks,citecolor=blue,linkcolor=blue,urlcolor=red,pagebackref]{hyperref}
\usepackage{comment}
\usepackage[T1]{fontenc}

\makeatletter
\newcommand{\multiline}[1]{%
  \begin{tabularx}{\dimexpr\linewidth-\ALG@thistlm}[t]{@{}X@{}}
    #1
  \end{tabularx}
}
\makeatother
\usepackage{mathtools}
\usepackage{diagbox}

\makeatletter
\def\BState{\State\hskip-\ALG@thistlm}
\makeatother

\usepackage[compact]{titlesec}
\titlespacing{\section}{0pt}{3ex}{2ex}
\titlespacing{\subsection}{0pt}{2ex}{1ex}
\titlespacing{\subsubsection}{0pt}{0.5ex}{0ex}

\newtheorem{theorem}{Theorem}[section]

\newtheorem{remark}[theorem]{Remark}

\newtheorem{corollary}[theorem]{Corollary}

\newtheorem{lemma}[theorem]{Lemma}
\newtheorem{claim}[theorem]{Claim}
\newtheorem{hypothesis}{Hypothesis}

\newtheorem{problem}[theorem]{Problem}

\makeatletter
\let\c@fconjecture\c@conjecture
\makeatother

\makeatletter
\let\c@fconj\c@conj
\makeatother

\newtheorem{innercustomhypo}{Hypothesis}
\newenvironment{customhypo}[1]
{\renewcommand\theinnercustomhypo{#1}\innercustomhypo}
{\endinnercustomhypo}

\def \eps {\varepsilon}

\def \poly {\mathop{\rm poly}} %

\def\tO{\widetilde{O}}

\newcommand{\Var}{\mathrm{Var}}

\newcommand{\negbigskip}{\vspace{-\bigskipamount}}
\newcommand{\OO}{\widetilde{O}}
\newcommand{\OOmega}{\widetilde{\Omega}}
\newcommand{\Ex}{\mathbb{E}}

\newcommand{\OV}{\mbox{\sf OV}}
\newcommand{\OVCount}{\mbox{\sf \#OV}}
\newcommand{\RealAPSP}{\mbox{\sf Real-APSP}}
\newcommand{\IntAPSP}{\mbox{\sf Int-APSP}}
\newcommand{\APSP}{\mbox{\sf APSP}}
\newcommand{\APSPCount}{\mbox{\sf \#APSP}}
\newcommand{\APSPCountMod}[1]{\mbox{\sf \#$_{\scriptsize {\rm mod}\, #1}$APSP}}

\newcommand{\MinPlus}{\mbox{\sf Min-Plus-Product}}
\newcommand{\MinPlusCount}{\mbox{\sf \#Min-Plus-Product}}
\newcommand{\Equality}{\mbox{\sf Equality-Product}}
\newcommand{\ExistEquality}{\mbox{\sf $\exists$Equality-Product}}
\newcommand{\Dominance}{\mbox{\sf Dominance-Product}}
\newcommand{\ExactTri}{\mbox{\sf Exact-Tri}}
\newcommand{\AEExactTri}{\mbox{\sf AE-Exact-Tri}}
\newcommand{\ExactTriCount}{\mbox{\sf \#Exact-Tri}}
\newcommand{\AEExactTriCount}{\mbox{\sf \#AE-Exact-Tri}}
\newcommand{\AEExactTriCountReal}{\mbox{\sf \#AE-Real-Exact-Tri}}
\newcommand{\NegTri}{\mbox{\sf Neg-Tri}}
\newcommand{\AENegTri}{\mbox{\sf AE-Neg-Tri}}
\newcommand{\AENegTriCount}{\mbox{\sf \#AE-Neg-Tri}}
\newcommand{\AENegTriCountReal}{\mbox{\sf \#AE-Real-Neg-Tri}}
\newcommand{\NegTriCount}{\mbox{\sf \#Neg-Tri}}

\newcommand{\AEMonoTri}{\mbox{\sf AE-Mono-Tri}}
\newcommand{\ExactKClique}{\mbox{\sf Exact-$k$-Clique}}
\newcommand{\ExactKCliqueCount}{\mbox{\sf \#Exact-$k$-Clique}}
\newcommand{\MinKClique}{\mbox{\sf Min-$k$-Clique}}
\newcommand{\MinKCliqueCount}{\mbox{\sf \#Min-$k$-Clique}}
\newcommand{\MinPlusConv}{\mbox{\sf Min-Plus-Convolution}}
\newcommand{\MinPlusConvCount}{\mbox{\sf \#Min-Plus-Convolution}}
\newcommand{\MinEqualityConv}{\mbox{\sf Min-Equal-Convolution}}

\newcommand{\MinWitness}{\mbox{\sf Min-Witness-Prod}}%
\newcommand{\MinWitnessEq}{\mbox{\sf Min-Witness-Eq-Prod}}
\newcommand{\MinEqualityProd}{\mbox{\sf Min-Equal-Prod}}
\newcommand{\MaxMin}{\mbox{\sf Max-Min-Prod}}

\newcommand{\uAPSP}{\mbox{\sf u-dir-APSP}}
\newcommand{\APSLP}{\mbox{\sf undir-APSLP$_{1, 2}$}}
\newcommand{\APSLPIntro}{\mbox{\sf undir-APSLP}}
\newcommand{\APLSP}{\mbox{\sf undir-APLSP$_{1, 2}$}}
\newcommand{\APLSPIntro}{\mbox{\sf undir-APLSP}}

\newcommand{\BatchMode}{\mbox{\sf Batch-Range-Mode}}

\newcommand{\ThreeSUM}{\mbox{\sf 3SUM}}
\newcommand{\ThreeSUMCount}{\mbox{\sf \#3SUM}}
\newcommand{\AllThreeSUM}{\mbox{\sf All-Nums-3SUM}}
\newcommand{\AllThreeSUMCount}{\mbox{\sf \#All-Nums-3SUM}}
\newcommand{\AllThreeSUMCountReal}{\mbox{\sf \#All-Nums-Real-3SUM}}
\newcommand{\ThreeSUMConv}{\mbox{\sf 3SUM-Convolution}}
\newcommand{\ThreeSUMConvCount}{\mbox{\sf \#3SUM-Convolution}}
\newcommand{\AllThreeSUMConv}{\mbox{\sf All-Nums-3SUM-Convolution}}
\newcommand{\AllThreeSUMConvCount}{\mbox{\sf \#All-Nums-3SUM-Convolution}}

\newcommand{\MonoConv}{\mbox{\sf Mono-Convolution}}
\newcommand{\MonoConvCount}{\mbox{\sf \#Mono-Convolution}}
\newcommand{\BMM}{\mbox{\sf BMM}}

\newcommand{\StrongAPSP}{Strong APSP Hypothesis}
\newcommand{\StrongAPSPPDF}{Strong APSP Hypothesis}
\newcommand{\StrongAPSPShort}{Strong APSP}
\newcommand{\StrongConv}{Strong Min-Plus Convolution Hypothesis}
\newcommand{\StrongConvShort}{Strong Min-Plus Convolution}
\newcommand{\uAPSPH}{u-dir-APSP Hypothesis}

\newcommand{\TunwtdirAPSP}{\textsc{{u}-dir-apsp}}%
\newcommand{\Tminwit}{\textsc{min-witness-prod}}

\newcommand{\TundirAPSLP}{\textsc{undir-apslp}_{1,2}}
\newcommand{\Trangemode}{\textsc{batched-range-mode}}%
\newcommand{\Tminwiteq}{\textsc{min-witness-eq-prod}}
\newcommand{\Tgeneq}{\textsc{gen-eq-prod}}
\newcommand{\Teq}{\textsc{eq-prod}}

\newcommand{\up}[1]{\left\lceil#1\right\rceil}

\newcommand{\lam}{\lambda}

\DeclareMathOperator*{\argmax}{arg\,max}
\DeclareMathOperator*{\argmin}{arg\,min}

\setlist[itemize]{leftmargin=*}
\setlist[enumerate]{leftmargin=*}

\title{Fredman's Trick Meets Dominance Product: Fine-Grained Complexity of Unweighted APSP, 3SUM Counting, and More}

\author{Timothy M. Chan\thanks{Supported by NSF Grant CCF-2224271.}\\UIUC\\tmc@illinois.edu \and Virginia {Vassilevska Williams}\thanks{Supported by an NSF CAREER Award, NSF Grant CCF-2129139 and BSF Grant BSF:2012338, a Google Research Fellowship and a Sloan Research Fellowship.}\\MIT\\virgi@mit.edu \and Yinzhan Xu\thanks{Partially supported by NSF Grant CCF-2129139.}\\MIT\\xyzhan@mit.edu}
\date{}

\begin{document}
\pagenumbering{gobble} 
\maketitle

\begin{abstract}
In this paper we carefully combine Fredman's trick [SICOMP'76] and Matou\v{s}ek's approach for dominance product [IPL'91] to obtain powerful results in fine-grained complexity:
\begin{itemize}
    \item 
    Under the hypothesis that APSP for undirected graphs with edge weights in $\{1, 2, \ldots, n\}$ requires $n^{3-o(1)}$ time (when $\omega=2$), we show a variety of conditional lower bounds, including an $n^{7/3-o(1)}$ lower bound for unweighted directed APSP and an $n^{2.2-o(1)}$ lower bound for computing the Minimum Witness Product between two $n \times n$ Boolean matrices, even if $\omega=2$, improving upon their trivial $n^2$ lower bounds. Our techniques can also be used to reduce the unweighted directed APSP problem to other problems. In particular, we show that (when $\omega = 2$), if unweighted directed APSP requires $n^{2.5-o(1)}$ time, then Minimum Witness Product requires $n^{7/3-o(1)}$ time. 
    
    \item We show that, surprisingly, many central problems in fine-grained complexity are equivalent to their natural counting versions. In particular, we show that Min-Plus Product and Exact Triangle are subcubically equivalent to their counting versions, and 3SUM is subquadratically equivalent to its counting version. 

    \item We obtain new algorithms using new variants of the Balog-Szemer\'edi-Gowers theorem from additive combinatorics.  For example, we  get an $O(n^{3.83})$ time deterministic algorithm for exactly counting the number of shortest paths in an arbitrary weighted graph, improving the textbook $\widetilde{O}(n^{4})$ time algorithm. We also get  faster algorithms for 3SUM in preprocessed universes, and deterministic algorithms for 3SUM on monotone sets in $\{1, 2, \ldots, n\}^d$. 
\end{itemize}
\end{abstract}

\newpage

\tableofcontents
\newpage

\pagenumbering{arabic} 
\section{Introduction}
\input{intro}

\section{Problem Definitions}
\label{sec:prelim}
\input{prelim}

\section{Conditional Lower Bounds for Problems with Intermediate Complexity:\texorpdfstring{\\}{} u-dir-APSP under the \texorpdfstring{\StrongAPSP{}}{\StrongAPSPPDF{}}}
\label{sec:intapsp-lower-bound}
\input{int_apsp_reduce_preview}

\section{Equivalences Between Counting and Detection Problems:\texorpdfstring{\\}{} \#Exact-Triangle vs.\ Exact-Triangle}
\label{sec:counting:preview}
\input{counting_preview}

\section{Alternative to BSG: A Triangle Decomposition Theorem and Its Applications}
\label{sec:decomposition-zero-tri}
\input{tri_decompose}

\section{More Lower Bounds under the \texorpdfstring{\StrongAPSP{}}{\StrongAPSPPDF{}}}
\label{sec:intapsp-lower-bound:more}

\input{int_apsp_reduce_more}

\section{More Equivalences Between Counting and Detection Problems}
\label{sec:counting}

\input{counting}

\section{Counting Algorithms in Other Models}
\label{sec:other_models}
\input{co-nondet}

\section{BSG Theorems Revisited}
\label{sec:bsg}

\input{bsg}

\section{Lower Bounds for Min-Equality Convolution}
\label{sec:min-equal-conv}
\input{min_eq_conv}

\section{Acknowledgement}
We would like to thank Lijie Chen for suggesting the quantum \ThreeSUMCount{} problem. 

\bibliographystyle{alpha}
\bibliography{ref}

\appendix

\section{Still More Equivalences Between Counting and Detection Problems}
\label{sec:more_counting}
\input{more_counting}

\input{more_bsg}

\end{document}

%% file: intro.tex
There are many examples of computational problems for which a slight variant can become much harder than the original: 2-SAT (in P) vs.\ 3-SAT (NP-complete), perfect matching (in P) vs.\ its counting version (\#P-complete), and so on. Fine-grained complexity (FGC) has provided explanations for such differences for many natural problems -- e.g.\ finding a minimum weight triangle in a node-weighted graph is ``easy'' (in $O(n^\omega)$ time \cite{CzumajL09}, where $\omega<2.373$ \cite{alman2021refined} is the matrix multiplication exponent) while finding a minimum weight triangle in an edge-weighted graph is ``hard'' (subcubically equivalent to All-Pairs Shortest Paths (\APSP{}) \cite{focsy,focsyj}). Nevertheless, many fundamental questions remain unanswered:

\begin{itemize}
\item Seidel \cite{seidel1995} showed that \APSP{} in unweighted {\em undirected}  graphs can be solved  in $\tilde{O}(n^\omega)$ time\footnote{The notation $\tilde{O}(f(n))$ denotes $f(n)\poly\log(n)$.}. The fastest algorithm for \APSP{} in unweighted {\em directed}  graphs by Zwick~\cite{zwickbridge} runs in $O(n^{2.53})$ for the current bound of rectangular matrix multiplication~\cite{legallurr}. If $\omega=2$, then undirected \APSP{} would have an essentially optimal $\tilde{O}(n^2)$ time algorithm, whereas Zwick's algorithm for directed \APSP{}
would run in, slower, $\tilde{O}(n^{2.5})$ time. 

Assuming that the runtime of Zwick's algorithm is the best possible for unweighted directed  APSP (abbreviated \uAPSP{}) has been used as a hardness hypothesis (see e.g. \cite{lincoln2020monochromatic,williamsxumono}). However, so far there has been {\em no explanation} for why \uAPSP{} should be harder than its undirected counterpart.
\begin{center}
{\em \hypertarget{question:Q1}{Q1}: Does \APSP{} in directed unweighted graphs require superquadratic time, even if $\omega=2$?}
\end{center}

\item Boolean matrix multiplication (\BMM{}) asks to compute, for two $n\times n$ Boolean matrices $A$ and $B$, for every $i,j\in \{1,\ldots,n\}$ whether there is some $k$ so that $A[i,k]=B[k,j]=1$. The Min-Witness product (\MinWitness{}) is a well-studied \cite{KowalukLingas, VassilevskaWY10,shapira2011all,CohenY14,KowalukL21} generalization of \BMM{} in which for every $i,j$, one needs to instead compute the {\em minimum} $k$ for which $A[i,k]=B[k,j]=1$.  Similarly to \uAPSP{}, the fastest algorithm for \MinWitness{}  runs in $O(n^{2.53})$ time \cite{CzumajKL07}, and would run in $\tilde{O}(n^{2.5})$ time if $\omega=2$, whereas \BMM{} would have an essentially optimal algorithm in that case.

The assumption that \MinWitness{} requires $n^{2.5-o(1)}$ time has been used as a hardness hypothesis (e.g. \cite{lincoln2020monochromatic}), but similar to \uAPSP{}, so far there has been no explanation as to why \MinWitness{} should be harder than \BMM{}.
\begin{center}
{\em \hypertarget{question:Q2}{Q2}: Does \MinWitness{} require superquadratic time, even if $\omega=2$?}
\end{center}

\item For a large number of problems of interest within FGC, the known algorithms for finding a solution can also compute the {\em number} of solutions in the same time. This is true for triangle detection in graphs, \BMM{}, Min-Plus product, \ThreeSUM{}, Exact Triangle, Negative Triangle, and many more. Is this merely a coincidence, or are the decision variants of these problems equivalent to the counting variants?

A natural and important question is:
\begin{center}
{\em
\hypertarget{question:Q3}{Q3}: Are the core FGC problems like \ThreeSUM{}, Min-Plus product and Exact Triangle easier than their exact counting variants?
}
\end{center}

A recent line of work \cite{DellL21,DellLM20} gives fine-grained reductions from {\em approximate counting} to decision for several core problems of FGC, showing that if the decision versions have improved algorithms, then one can also obtain fast approximation schemes. As an approximate count can solve the decision problem, we get that
 for many key problems {\em approximate counting} and decision are equivalent.

This does not imply that {\em exact counting} would be equivalent to decision however. In fact, quite often the decision and counting versions of computational problems can have vastly different complexities. There are many examples of polynomial time decision problems (e.g.\ perfect matching \cite{Valiant79}) whose counting variants become $\#$P-complete. 
For many of these problems, polynomial time approximation schemes are possible (see e.g.\ \cite{JerrumSV04} for perfect matching), however obtaining a fast {\em exact} counting algorithm is considered infeasible.

Within FGC, there are more examples. Consider for example the case of induced subgraph isomorphism for pattern graphs of constant size $k$.  It is known (see e.g.\  \cite{CurticapeanDM17}) that for {\em every} $k$-node pattern graph $H$, {\em counting} the induced copies of $H$ in an $n$-node graph is fine-grained equivalent to {\em counting} the $k$-cliques in an $n$-node graph.
Meanwhile, the work of \cite{DellL21,DellLM20} implies that {\em approximately} counting the induced $k$-paths in a graph is fine-grained equivalent to detecting a single $k$-path. However, induced $k$-paths can be found (and hence can be approximately counted) {\em faster} than counting $k$-cliques: combinatorially, there's an $O(n^{k-2})$ time algorithm \cite{BlaserKS18}, whereas $k$-clique counting is believed to require $n^{k-o(1)}$ time; with the use of fast matrix multiplication, if $k<7$, $k$-paths can be found and approximately counted in the best known running time for $(k-1)$-clique detection \cite{WilliamsWWY15,BlaserKS18}, faster than $k$-clique.

To reiterate Question \hyperlink{question:Q3}{Q3} in this context, it asks whether the core FGC problems are like induced $k$-path above, or perhaps surprisingly, are equivalent to their counting versions?

The fine-grained problems of interest such as \ThreeSUM{}, Exact-Triangle, Orthogonal Vectors and more, admit efficient self-reductions.
For such problems, prior work~\cite{focsyj} has given generic techniques to show that the problems are fine-grained equivalent to the problem of {\em listing} any ``small'' number of solutions.
For example, \ThreeSUM{} is subquadratically equivalent to listing any truly subquadratic\footnote{Truly subquadratic means $O(n^{2-\eps})$ for some constant $\eps>0$; truly subcubic means $O(n^{3-\eps})$ for constant $\eps>0$, etc.} number of \ThreeSUM{} solutions.
 Thus, as long as the count is small, the listing problem is equivalent to the decision problem. This technique has become an important technique in FGC, used in many subsequent works (e.g.~\cite{HenzingerKNS15,backurs2017better,KunnemannPS17,cygan2019problems}).
The issue is, however, that when the count is actually ``large'', say $\Omega(n^2)$ for \ThreeSUM{}, the listing approach is too expensive. 
The question becomes, how do we count faster than listing when the count is large? Until now (more than 10 years after the conference version of~\cite{focsyj}) there has been no technique to do this.

\item
Recently, wide classes of structured instances of Min-Plus
matrix products, Min-Plus convolutions, and \ThreeSUM{} have been
identified which can be solved faster.
Chan and Lewenstein~\cite{ChanLewenstein} obtained the first truly 
subquadratic algorithm for Min-Plus convolution for monotone 
sequences over $[n]:=\{1,\ldots,n\}$ (and also integer sequences with bounded differences),
and Bringmann, Grandoni, Saha and Vassilevska W.~\cite{BringmannGSW16} 
obtained the first truly subcubic algorithm for Min-Plus 
product for integer matrices with bounded differences. 
The latter result is in some sense more general, as bounded-difference Min-Plus convolution
reduces to (row/column) bounded-difference Min-Plus matrix product.  
Both
results have subsequently been generalized or improved \cite{WilliamsX20, GuPWX21, ChiDXsoda22, ChiDXZstoc22}.  Interestingly, Chan and Lewenstein's approach made use of a
famous result in additive combinatorics known as the \emph{Balog--Szemer\'edi--Gowers Theorem} (the ``BSG Theorem''), whereas \cite{BringmannGSW16}'s approach
and all  subsequent algorithms use more direct techniques
without the need for additive combinatorics.   So far, there has been no explanation for why
there are two seemingly different approaches.
This leads to our last main question:

\begin{center}\em
\hypertarget{question:Q4}{Q4}: What is the relationship between the BSG-based approach and the direct approach to monotone Min-Plus
product/convolution, and could they be unified?
\end{center}

Question \hyperlink{question:Q4}{Q4} might appear unrelated to the preceding questions, and is more vague
or conceptual.  But the hope is that by understanding the relationship better,
we may obtain new improved algorithms, since Chan and Lewenstein's approach
has several applications, e.g.\ to \ThreeSUM{} with preprocessed universes and \ThreeSUM{} for 
$d$-dimensional monotone sets in $[n]^d$, which are not handled by the subsequent approaches.

\end{itemize}

\subsection{Summary of Our Contributions and New Tool}

\begin{enumerate}

\item {\bf Conditional hardness for \uAPSP{}.}  One of the main hypotheses of fine-grained complexity, known as the ``APSP Hypothesis'', is that \APSP{} in $n$-node graphs with polynomial integer edge weights requires $n^{3-o(1)}$ time. Meanwhile, the fastest algorithms for \APSP{} \cite{seidel1995, zwickbridge, Williams18} still fail to solve the problem in truly subcubic time when the edge weights are in $[n]$. In fact, even Min-Plus product on $n\times n$ matrices with entries in $[n]$ is not known to be solvable in truly subcubic time: the known algorithm for small integer entries runs in $\tilde{O}(Mn^\omega)$ time \cite{ALONGM1997}, and even if $\omega=2$, this is no better than the brute-force cubic time algorithm when $M=n$. Moreover, Zwick \cite{zwickbridge} explicitly asks whether there is a truly subcubic time algorithm for \APSP{} with weights in $[n]$.

We formulate a version of the APSP Hypothesis, which we call the \StrongAPSP{}, stating (for $\omega=2$) that \APSP{} in graphs with edge weights in  $[n]$ requires $n^{3-o(1)}$ time. Under this hypothesis we show that  \uAPSP{} requires $n^{7/3-o(1)}$ time, even if $\omega=2$. 

Thus, we conditionally resolve question \hyperlink{question:Q1}{Q1}: either {\em unweighted directed} APSP requires super-quadratic time and is harder than unweighted undirected \APSP{} (when $\omega=2$), or \APSP{} in weighted graphs with weights in $[n]$ is in truly subcubic time.

This is the first fine-grained connection between unweighted and weighted \APSP{}.

\item {\bf Conditional hardness for \MinWitness{}.} 
We present the first fine-grained reduction from \APSP{} in unweighted directed graphs to \MinWitness{} (which was  %
left open by \cite{lincoln2020monochromatic}). 
Our reduction implies that, when $\omega = 2$, \MinWitness{} requires $n^{7/3 - o(1)}$ time unless the current best algorithm for \uAPSP{}~\cite{zwickbridge} can be improved. 
Alternatively, working under the \StrongAPSP{},
we also obtain an $n^{11/5-o(1)}$ lower bound for \MinWitness{}.

Thus, either \MinWitness{} is truly harder than \BMM{} (if $\omega=2$), or there is a breakthrough in unweighted or weighted \APSP{} algorithms. This gives an answer to question \hyperlink{question:Q2}{Q2}.

See Section~\ref{sec:interm} for further discussion and details of our results on \MinWitness{} and \uAPSP{}.

\item  {\bf More hardness results.} We give many more fine-grained lower bounds under the \StrongAPSP{} for problems such as Batched Range Mode, All Pairs Shortest Lightest Paths, Min Witness Equality Product, dynamic shortest paths in unweighted planar graphs and more.

\item {\bf Counting is equivalent to detection.} We resolve question \hyperlink{question:Q3}{Q3} for several core problems in FGC\@. 
We show that the \MinPlus{} problem is subcubically equivalent to its counting version, the \ThreeSUM{}  problem is subquadratically equivalent to its counting version, and  the Exact-Weight Triangle  problem (\ExactTri) is subcubically equivalent to its counting version. These are the first  fine-grained equivalences between exact counting and decision problem variants, to our knowledge.
See Section~\ref{sec:count} for further results and discussion.

\item {\bf New variants of the BSG Theorem and new algorithms.}
We formulate a new decomposition theorem for zero-weight triangles of weighted graphs,
which may be viewed as a substitute to the BSG Theorem but has
a simple direct proof, providing an answer to question \hyperlink{question:Q4}{Q4}.
Besides being applicable to monotone Min-Plus convolution/products,
\ThreeSUM{} with preprocessed universes and \ThreeSUM{} with $d$-dimensional monotone sets,
the theorem yields the first truly subquartic algorithm for the counting version of the 
general weighted \APSP{} problem.  Our ideas also lead to a new bound on
the BSG Theorem itself which is an improvement for a certain range of parameters.
As a result, we obtain an improved new algorithm for
\ThreeSUM{} with preprocessed universes.
See Sections \ref{sec:newalgs} and \ref{sec:bsg} for more details.

\end{enumerate}

Surprisingly, we are able to achieve all of these results with a {\bf single new tool},
 a careful combination of two known techniques in tackling shortest paths problems in graphs: Fredman's trick \cite{fredman1976new} and Matou\v{s}ek's approach for dominance product~\cite{MatIPL}.

It has long been known that \APSP{} in general $n$-node graphs is equivalent to computing the Min-Plus product of two $n\times n$ matrices $A$ and $B$ (\MinPlus{}), defined as the matrix $C$ with $C_{ij}=\min_k (A_{ik}+B_{kj})$.
Fredman \cite{fredman1976new} introduced the following powerful ``trick'' for dealing with the above minimum: 
 to determine if $A_{ik}+B_{kj}\leq A_{i\ell}+B_{\ell j}$, we simply need to check if $$A_{ik}-A_{i\ell}\leq B_{\ell j}-B_{kj}.$$

While seemingly trivial, this idea of comparing a left-hand side that is purely in terms of entries of $A$ to a right-hand side that is purely in terms of entries of $B$, leads to a variety of amazing results. Fredman used it to show that the decision tree complexity of \MinPlus{} is $O(n^{2.5})$, and not cubic as previously thought. Practically all subcubic algorithms for \APSP{} (e.g.\  \cite{fredman1976new,Takaoka98,Chan10,Williams18}) including the current fastest $n^3/\exp(\Theta(\sqrt{\log n}))$ time algorithm by Williams \cite{Williams18} use Fredman's trick. Several new truly subcubic time algorithms for variants of \APSP{} (e.g. \cite{ChiDXsoda22} and, implicitly, \cite{BringmannGSW16}) and recent algorithms for \ThreeSUM{} (e.g.\ \cite{gronlund2014,chan3sum}) also use it.

A completely different technique is Matou\v{s}ek's truly subcubic time algorithm \cite{MatIPL} for Dominance Product (\Dominance{}). The dominance product of two $n\times n$ matrices $A$ and $B$ is the $n\times n$ matrix $C$ such that for all $i,j\in [n]$,
$C_{ij}= |\{k \in [n]: A_{ik}\leq B_{kj}\}|$. Matou\v{s}ek gave an approach that combined fast matrix multiplication with brute-force to obtain an $\OO(n^{(3+\omega)/2})$ time
algorithm for \Dominance{}. \Dominance{} is known to be equivalent to the so-called Equality Product (\Equality{}) problem\footnote{The equivalence was proven e.g.\ by \cite{labib2019hamming,vnotes}, but also Matou\v{s}ek's algorithm almost immediately works for \Equality{}.} which asks to compute $C_{ij}= |\{k \in [n]: A_{ik}=B_{kj}\}|$, so we sometimes refer to \Equality{} instead.

Matou\v{s}ek's subcubic time techniques have been used to obtain truly subcubic time algorithms for All-Pairs Bottleneck Paths \cite{VassilevskaWY07,duanpettiebott}, All-Pairs Nondecreasing Paths \cite{nondecreasingv,DuanJW19}, \APSP{} in node-weighted graphs \cite{Chan10} and more. Unfortunately, the technique has fallen short when applied directly to the general \APSP{} problem.

In this paper we give a combination of these two techniques that allows us to obtain reductions that exploit fast matrix multiplication in a new way, thus allowing us to overcome many difficulties, such as counting solutions when the number of solutions is large.
The main ideas involve:
(i)~division into ``few-witnesses'' and ``many-witnesses'' cases,
(ii)~standard witness-finding techniques to handle the ``few-witnesses'' case,
(iii)~hitting sets to hit the ``many-witnesses'',
(iv)~Fredman's trick (the obvious equivalence of $a+b=a'+b'$ with $a-a'=b'-b$), and lastly 
(v) ~Matou\v sek's technique for dominance or equality products (which also involves a division into ``low-frequency'' and
``high-frequency'' cases). Steps (iv) and (v) crucially allow us to handle the ``many-witnesses'' case, delicately exploiting the small hitting set from (iii).  Individually, each of these ideas is simple and has appeared before.  But the novelty lies in how they are pieced together (see Sections 
\ref{sec:intapsp-lower-bound},
\ref{sec:counting:preview}, and
\ref{sec:decomposition-zero-tri}), and the realization that these ideas are powerful enough to yield all the new results in the above bullets!

In the remainder of the introduction we give more details on each of the bullets above. 

\subsection{Conditional Lower Bounds for Unweighted Directed APSP, Min-Witness Product and Other Problems with Intermediate Complexity}
\label{sec:interm}
\uAPSP{}  and \MinWitness{} are both ``intermediate'' problems as dubbed by Lincoln, Polak and Vassilevska W.~\cite{lincoln2020monochromatic}, a class of matrix product and all-pairs graph problems whose running times are $\OO(n^{2.5})$ when $\omega = 2$ (and hence right in the middle between the brute-force $n^3$ and and the desired optimal $n^2$), and for which no $O(n^{2.5-\eps})$ time algorithms are known.
More examples of intermediate problems include 
Min-Equality Product and Max-Min Product.
A similar class of ``intermediate'' {\em convolution} problems that are known to be solvable in $\OO(n^{1.5})$ time but not much faster~\cite{lincoln2020monochromatic} includes Max-Min Convolution, Minimum-Witness Convolution,  Min-Equality Convolution, %
and pattern-to-text Hamming distances.
For instance, $\OO(n^{1.5})$-time algorithms were known since the 1980s for Max-Min convolution~\cite{Kosaraju89a} and 
pattern-to-text Hamming distances~\cite{Abrahamson87}, and these remain the fastest algorithms for these problems.

None of these intermediate problems currently have nontrivial conditional lower bounds under standard hypotheses in FGC.  (Two exceptions are All-Edges Monochromatic Triangle (\AEMonoTri{}) and Monochromatic Convolution (\MonoConv{}), where
a near-$n^{2.5}$ lower bound is known for the former under the 3SUM or the APSP Hypothesis~\cite{lincoln2020monochromatic,williamsxumono} and a near-$n^{1.5}$ lower bound is known for the latter under the 3SUM Hypothesis~\cite{lincoln2020monochromatic}, but one may argue that these monochromatic problems are not true matrix product or convolution problems since 
their inputs involve {\em three} matrices or sequences rather than two.)

Thus, an important research direction is to prove super-quadratic conditional lower bounds for intermediate matrix product problems, and super-linear lower bounds for intermediate convolution problems. 
Some relationships are known between intermediate problems, some of which are illustrated in  Figure~\ref{fig:previous}.
However, many questions remain; for instance, it was open whether \uAPSP{} and \MinWitness{} are related.

\input{previous_work_figure}

As \uAPSP{} and \MinWitness{} are among the ``lowest'' problems in this class, proving conditional lower bounds for these two problems is especially fundamental.
(Besides, \MinWitness{} has been extensively studied, arising in many applications \cite{KowalukLingas, VassilevskaWY10,shapira2011all,CohenY14,KowalukL21}.)
The precise time bounds of the current best algorithms for \uAPSP{} by Zwick~\cite{zwickbridge} and 
\MinWitness{} by Czumaj, Kowaluk and Lingas~\cite{CzumajKL07} are both
$\OO(n^{2+\rho})$, where $\rho\in [1/2, 0.529)$ satisfies\footnote{
$\omega(a, b, c)$ denotes the rectangular matrix multiplication exponent between an $n^a \times n^b$ matrix and an $n^b \times n^c$ matrix.
} $\omega(1,\rho,1)=1+2\rho$.

To explain the hardness of \uAPSP{}, \MinWitness{}  and more, we 
introduce a strong version of the APSP Hypothesis: %
\begin{hypothesis}[\StrongAPSP]
In the Word-RAM model with $O(\log n)$-bit words, computing APSP for an undirected\footnote{
The version of this hypothesis for directed graphs is equivalent: see Remark~\ref{rmk:strongapsp:dir}.
} graph with edge weights in $[n^{3-\omega}]$ requires $n^{3-o(1)}$ randomized time. 
\end{hypothesis}

For $\omega=2$, the above asserts that the standard  textbook cubic time algorithms for \APSP{} are near-optimal when the edge weights are integers in $[n]$.  
Due to a tight equivalence~\cite{ShoshanZwick} between undirected \APSP{} with edge weights in $[M]$ and \MinPlus{} of two matrices with entries in $[O(M)]$, 
the above hypothesis is equivalent to the following:

\begin{customhypo}{1$'$}[\StrongAPSP{}, restated]
In the Word-RAM model with $O(\log n)$-bit words, computing the Min-Plus product between two $n \times n$ matrices with entries in $[n^{3-\omega}]$ requires $n^{3-o(1)}$ randomized time.
\end{customhypo}

The current upper bound for \MinPlus{} between two $n\times n$ matrices with entries in $[M]$ is $\OO(\min\{Mn^\omega,n^3\})$~\cite{ALONGM1997}, which is cubic for $M=n^{3-\omega}$, and this has not been improved for over 30 years.

The above strengthening of the APSP Hypothesis is reminiscent of a strengthening of the 3SUM Hypothesis proposed by Amir, Chan, Lewenstein and Lewenstein~\cite{AmirCLL14}, which they called the ``Strong 3SUM-hardness assumption'',
asserting that the 3SUM Convolution problem requires near-quadratic time for integers in $[n]$.  %

We prove the following results:

\begin{theorem}
\label{thm:intro:strong-apsp}
Under the \StrongAPSP{}, 
\uAPSP{} requires $n^{7/3-o(1)}$ time, and
    \MinWitness{} requires $n^{11/5-o(1)}$ time (on a Word-RAM with $O(\log n)$-bit words).
\end{theorem}

In fact, for \uAPSP{}, we can obtain a conditional lower bound of $n^{2+\beta - o(1)}$ for graphs with 
 $n^{1+\beta}$ edges, for any $0 < \beta \le \frac{1}{3}$.  This time bound is \emph{tight} for graphs of such sparsity.  In other words, under the \StrongAPSP{}, 
the naive $O(mn)$ time algorithm for \uAPSP{} with $m$ edges by repeated BFSs is essentially optimal for sufficiently sparse graphs, and fast matrix multiplication helps only for sufficiently large densities.
Such lower bounds for sparse graphs
were known only for \emph{weighted} APSP~\cite{LincolnWW18,AgarwalR18}.

Our technique also yields new conditional lower bounds for many other problems, such as All-Pairs Shortest Lightest Paths (\APSLPIntro{})
or All-Pairs Lightest Shortest Paths (\APLSPIntro{}) for undirected small-weighted graphs (which was first studied by Zwick~\cite{ZwickAPLSP}), a batched version of the range mode problem (\BatchMode{})
(which has received considerable recent attention~\cite{CDLMW, WilliamsX20, GuPWX21, GaoHe22, JinX22}), and dynamic shortest paths in planar graphs~\cite{AbbDah}.
See Table~\ref{tab:bounded-min-plus}  for some of the specific results, and Section~\ref{sec:intapsp-lower-bound:more} for more discussion on all these problems.  As demonstrated by the applications to range mode and dynamic planar shortest paths,
our new technique will likely be useful to proving
conditional lower bounds for other data structure problems, %
serving as an alternative to existing techniques based on the Combinatorial BMM Hypothesis (see e.g. \cite{abboud2014popular}) or the OMv Hypothesis~\cite{HenzingerKNS15}.

We also consider conditional lower bounds based on the hardness of \uAPSP{} itself.
The following hypothesis has been proposed by Chan, Vassilevska W. and Xu~\cite{CVXicalp21}:

\begin{hypothesis}[\uAPSPH{}]
In the Word-RAM model with $O(\log n)$-bit words, computing APSP for an $n$-node unweighted directed graph requires at least $n^{2+\rho-o(1)}$ randomized time where $\rho$ is the constant satisfying  $\omega(1,\rho,1)=1+2\rho$.
\end{hypothesis}

As noted in Remark~\ref{rmk:uapsp},
the \uAPSPH{} implies the \StrongAPSP{} if $\omega=2$.
The \StrongAPSP{} is thus more believable in some sense, but the \uAPSPH{} allows us to prove higher lower bounds.
For example,  we prove the following conditional lower bound for \MinWitness{}: 

\begin{theorem}
\label{thm:intro:minwit}
Under the \uAPSPH{}, 
\MinWitness{}
requires $n^{2.223}$ time, or $n^{7/3-o(1)}$ time if $\omega=2$ (on a Word-RAM with $O(\log n)$-bit words).
\end{theorem}

Earlier papers~\cite{lincoln2020monochromatic,CVXicalp21} were unable to obtain such a reduction
from \uAPSP{} to \MinWitness{} (Chan, Vassilevska W. and Xu~\cite{CVXicalp21} were only able to reduce from \uAPSP{} to \MinWitnessEq{}, but not  \MinWitness{}).
We similarly obtain higher lower bounds for all the other problems, as indicated in Table~\ref{tab:bounded-min-plus}.
Our results thus show that the \uAPSPH{} is far
more versatile for proving conditional lower bounds than what previous papers~\cite{lincoln2020monochromatic,CVXicalp21} were able to show.

Lastly, to deal with convolution-type problems, we introduce a similar strong version of the Min-Plus Convolution Hypothesis:

\begin{hypothesis}[\StrongConv{}]
In the Word-RAM model with $O(\log n)$-bit words, Min-Plus convolution between two length $n$ arrays with entries in $[n]$ requires $n^{2-o(1)}$ randomized time. 
\end{hypothesis}

The best algorithm for \MinPlusConv{} with numbers in $[M]$ runs in $\OO(Mn)$ time (by combining Alon, Galil and Margalit's technique~\cite{ALONGM1997} and Fast Fourier transform\@), and no truly subquadratic time algorithm is known for $M=n$.

We do not know any direct relationship between the \StrongAPSP{} and the \StrongConv{} (although when there are no restrictions on weights, it was known that \MinPlusConv{} reduces to \MinPlus{} \cite{bremner2006necklaces}).

We prove the first conditional lower bounds for one intermediate convolution problem, \MinEqualityConv{}, under the \StrongAPSP{}, the \uAPSPH{} or the \StrongConv{}:

\begin{theorem}\label{thm:intro:convol}
\MinEqualityConv{} requires 
        $n^{1+1/6-o(1)}$ time under the \StrongAPSP{}; or
    $n^{1+\rho/2-o(1)}$ time under the \uAPSPH{}; or
    $n^{1+1/11-o(1)}$ time under the \StrongConv{}. 
\end{theorem}

In the above theorem, the lower bound under the \StrongConv{} is the most interesting, requiring the use of one of our new variants of the BSG Theorem. 

We should emphasize that the significance of Theorem~\ref{thm:intro:convol} stems from the relative rarity of nontrivial lower bounds for convolution and related string problems.  Some problems may have $n^{\omega/2}$ lower bounds by reduction from \BMM{} (e.g.\ there was a well-known reduction from \BMM{} to pattern-to-text Hamming distances, attributed to Indyk -- see also \cite{GawrychowskiU18}), but such bounds become  meaningless when $\omega=2$.  A recent paper~\cite{CVXstoc22} 
obtained 
an $n^{1.5-o(1)}$ lower bound under the OV Hypothesis for a problem called ``pattern-to-text distinct Hamming similarity'', which is far less natural and basic than \MinEqualityConv{}.

\input{bounded_min_plus_table}

\subsection{Equivalence of Counting and Detection}
\label{sec:count}
Here we summarize our main results on counting vs detection.

\paragraph{Equivalence between counting and detection for \ThreeSUM.} In \ThreeSUMCount{} we are given three sets of numbers $A,B,C$ and we want to count the number of triples $a \in A,b\in B,c\in C$ such that $a+b=c$.\footnote{There are several equivalent definitions of \ThreeSUM: the predicate can be $a+b=c$, or $a+b+c=0$; the integers could come from the same set $A$, or the input could consist of three sets $A,B,C$ and we require $a\in A, b\in B, c\in C$. We work with the version with three input sets and predicate $a+b=c$.} More complex variants include \AllThreeSUMCount{} in which we want to count for every $c\in C$, the number of $a\in A,b\in B$ such that $a+b=c$. All of \ThreeSUM, \ThreeSUMCount{} and \AllThreeSUMCount{} can be solved in $O(n^2)$ time, via the folklore algorithm. The 3SUM %
Hypothesis (see \cite{virgisurvey}) asserts that (in the word-RAM model of computation with $O(\log n)$-bit words), $n^{2-o(1)}$ time is needed to solve \ThreeSUM. As no reduction from \ThreeSUMCount{} to \ThreeSUM{} is known till now, a priori it could be that the \ThreeSUM{} Hypothesis is false, but that \ThreeSUMCount{} still requires $n^{2-o(1)}$ time.
We prove:
\begin{theorem}
\ThreeSUMCount{} and \ThreeSUM{} are equivalent under truly subquadratic randomized fine-grained reductions.
\end{theorem}

\paragraph{Counting of \APSP{} and \MinPlus.}
It is known that \APSP{} in $n$-node graphs is equivalent to \MinPlus{} of $n\times n$ matrices \cite{Fischer71}. 
The counting variant \APSPCount{} of \APSP{} asks to count for every pair of nodes in the given graph, the number of shortest paths between them.
The counting variant of \MinPlus{} asks, given two matrices $A,B$ to count for all pairs $i,j$, the number of $k$ such that $A_{ik}+B_{kj}=\min_{\ell} (A_{i\ell}+B_{\ell j})$, i.e. the number of witnesses of the \MinPlus.

While \APSP{} and \MinPlus{} are equivalent,  \APSPCount{} and \MinPlusCount{}  are not known to be. The main issue is that the shortest paths counts can be exponential, and operations on such large numbers are costly in any reasonable model of computation such as the Word-RAM (see the discussion in Section \ref{sec:newalgs}). Because of this the fastest known algorithm for \APSPCount{} is actually {\em quartic} in the number of nodes, and not cubic. While we are able to improve upon the quartic running time (see Section \ref{sec:newalgs}), the running time is still supercubic, so it is unclear whether a tight reduction is possible from \APSPCount{} to \MinPlusCount.

Chan, Vassilevska W. and Xu \cite{CVXicalp21} defined several variants of \APSPCount{} that mitigate the existence of exponential counts. One variant, \APSPCountMod{U} computes the counts modulo any $O(\poly\log n)$-bit integer $U$, and thus no computations with large integers are necessary.  The problem can be solved in $\OO(n^3)$ time, and can be reduced to \MinPlusCount{} (see Appendix~\ref{sec:apsp-count-mod}). 
We prove:
\begin{theorem}
\MinPlusCount{} and \MinPlus{} are equivalent under truly subcubic fine-grained reductions. For any $O(\poly\log n)$-bit integer $U \ge 2$, \APSPCountMod{U} and \APSP{} are equivalent under truly subcubic fine-grained reductions.
\end{theorem}

\paragraph{Counting Exact Triangles.} 
The Exact Triangle Problem (\ExactTri{}) is: given an $n$-node graph with $O(\log n)$-bit integer edge weights, to determine whether the graph contains a triangle whose edge weights sum to some target value $t$. This problem is known to be at least as hard as both \ThreeSUM{} and \APSP{}~\cite{patrascu2010towards,VWfindingcountingj,focsyj}, so that if its brute-force cubic time algorithm can be improved upon, then both the \ThreeSUM{} Hypothesis and the \APSP{} Hypothesis would be false. \ExactTri{} is among the hardest problems in fine-grained complexity.
The counting variant \ExactTriCount{} asks for the number of triangles with weight $t$. We prove:

\begin{theorem}
\ExactTri{} and \ExactTriCount{} are equivalent under truly subcubic fine-grained reductions.
\end{theorem}

Abboud, Feller and Weimann~\cite{AbboudFW20} previously considered the problem of computing the {\em parity} of the number of zero-weight triangles, and it is equivalent to the problem of 
computing the  parity of the number of triangles with weight $t$ for a given $t$. Let's call this {\sf Parity-Exact-Tri}. They showed that \ExactTri{} can be reduced to  {\sf Parity-Exact-Tri} via a randomized subcubic fine-grained reduction. They were {\em not} able to obtain an equivalence, as it is not obvious at all that the decision problem should be able to solve the parity problem. Since \ExactTriCount{} can easily solve {\sf Parity-Exact-Tri} (if one knows the count, one can take it mod $2$), we also get that {\sf Parity-Exact-Tri} is equivalent to \ExactTri.

More equivalence results can be found in Section~\ref{sec:counting} and Appendix~\ref{sec:more_counting}.  See Section~\ref{sec:counting:discuss} for
further discussion on possible implications of our results.

\paragraph{New nondeterministic and quantum counting algorithms.}
Using our new techniques, we can also provide  efficient nondeterministic algorithms for \NegTriCount{} (counting the number of triangles with negative weights), \ExactTriCount{} and \ThreeSUMCount{} even when there are real-valued inputs. 
Ours are the first nondeterministic algorithms for these problems that beat their essentially cubic and quadratic deterministic algorithms.  See Section~\ref{sec:nondet:discuss} for the complexity-theoretic implications of these results.

Our equivalence results can be used to obtain new quantum algorithms.
It is known that \ThreeSUM{} can be solved in  $\OO(n)$ quantum time by an algorithm that uses Grover search (see e.g. \cite{AmbainisL20}). However, it is not clear how to adapt that algorithm to solve \ThreeSUMCount{} since Grover search cannot count the number of solutions exactly. By combining the $\OO(n)$ quantum time \ThreeSUM{} algorithm with our subquadratic equivalence between \ThreeSUM{} and \ThreeSUMCount{} in a black box way, we immediately obtain the first truly subquadratic time quantum algorithm for \ThreeSUMCount{}.

\subsection{New Variants of the BSG Theorem with Algorithmic Applications}
\label{sec:newalgs}

As a way to address question \hyperlink{question:Q4}{Q4}, we formulate a new ``Triangle Decomposition Theorem'', which follows from our techniques.
The theorem %
can be stated roughly as follows:
\begin{quote}
For any weighted $n$-node graph $G$ and parameter $s$, there exist $O(s^3)$ subgraphs
$G^{(\lam)}$ such that the set of all triangles of weight zero (or any fixed target value) in $G$
can be decomposed as the union of the set of all triangles in $G^{(\lam)}$, plus a small remainder set of
$O(n^3/s)$ triangles.  (All triangles in each $G^{(\lam)}$ are guaranteed to have weight zero.)
\end{quote}
Thus, the set of all zero-weight triangles in a graph is highly structured in some sense (in particular,
if there are many zero-weight triangles, one can extract large subgraphs in which all triangles are zero-weight triangles).
See Theorem~\ref{thm:tri:decompose} for the precise statement.
From this theorem, we can easily rederive a subcubic algorithm for monotone Min-Plus product (as shown in Appendix~\ref{app:bd:diff}),
if we don't care about optimizing the exponent in the running time.  
The theorem also leads to several new algorithms, as listed below.

This Triangle Decomposition Theorem resembles a covering version of the BSG Theorem, as formulated
by Chan and Lewenstein~\cite{ChanLewenstein}, which can be roughly stated as follows:
\begin{quote}
For any sets $A,B,C$ in an abelian group and parameter $s$, there exist $O(s)$ pairs
of subsets $(A^{(\lam)},B^{(\lam)})$ such that $\{(a,b)\in A\times B: a+b\in C\}$ can be covered
by the union of $A^{(\lam)}\times B^{(\lam)}$, plus a small remainder set of $O(n^2/s)$ pairs,
where the total size of the sum sets\footnote{
Throughout the paper, $A+B$ denotes the sum set $\{a+b: a\in A,\, b\in B\}$, and $A-B$ denotes the difference set $\{a-b: a\in A,\, b\in B\}$.
}
$A^{(\lam)} + B^{(\lam)}$ is $O(s^6 n)$.
\end{quote}
Thus, the set of all solutions to \ThreeSUM{} is highly structured in some sense (in particular,
if there are many solutions to \ThreeSUM{}, one can extract  pairs of large subsets $(A^{(\lam)},B^{(\lam)})$ whose
sum sets are small).
See Section~\ref{sec:bsg} for the precise statement and for more background on the BSG Theorem.

We show that 
our approach, combined with some extra ideas, can actually prove a form of the BSG Theorem, namely,
with the above $O(s^6 n)$ bound replaced by $\OO(s^2 n^{3/2})$.  At first, this new bound appears weaker---indeed,
researchers in additive combinatorics were  concerned more with obtaining a linear bound in $n$ on the sum set size, and less with the dependency on $s$.  However, Chan and Lewenstein's algorithmic applications require
nonconstant choices of $s$, and the new bound lowering the $s^6$ factor to $s^2$ turns out to yield better results
in at least one application below
(namely, \ThreeSUM{} with preprocessed universes).

We now mention a few applications of the new theorems:

\paragraph{\#APSP for arbitrary weighted graphs.}
Recall the \APSPCount{} problem
(counting the number of shortest paths from $u$ to $v$,
for every pair of nodes $u,v$ in a weighted graph).  This basic problem has applications
to %
\emph{betweenness centrality}.
As  mentioned earlier, because counts may be exponentially big,
known algorithms run in near-$n^4$ time
in the
standard Word-RAM model with $O(\log n)$-bit words.
An $\OO(n^3)$-time
algorithm for \APSPCount{} was given by Chan, Vassilevska W. and Xu~\cite{CVXicalp21}, but only for \emph{unweighted} graphs.
We give the first truly subquartic \APSPCount{} algorithm for  arbitrary weighted graphs with running time $O(n^{3.83})$
(or $\OO(n^{15/4})$ if $\omega=2$).

\paragraph{3SUM in preprocessed universes.}
Chan and Lewenstein~\cite{ChanLewenstein} studied a ``preprocessed universe'' setting for the \ThreeSUM{}
problem: preprocess three sets $A,B,C$ of $n$ integers (the \emph{universes})
so that for any given query subsets $A'\subseteq A$, $B'\subseteq B$, and $C'\subseteq C$, we can
solve \ThreeSUM{} on the instance $(A',B',C')$ more quickly than from scratch.
(The problem was first stated by Bansal and Williams~\cite{BansalW12}.)
Chan and Lewenstein showed intriguingly that after preprocessing the universes in $\OO(n^2)$ expected time,
\ThreeSUM{} on the given query subsets can be solved in time truly subquadratic in $n$, namely, $\OO(n^{13/7})$.
Our techniques yield a new, simpler solution with $\OO(n^2)$ expected preprocessing time and
an improved query time of $\OO(n^{11/6})$.  Furthermore, with the same $\OO(n^{11/6})$ query time,
the new algorithm can solve a slight generalization
where the query subset $C'$ can be an arbitrary set of $n$ integers (i.e., the universe
$C$ needs not be specified).  Here, our improvement is even bigger:
Chan and Lewenstein's previous solution for this generalized case
required $\OO(n^{19/10})$ query time.

We also obtain the first
result for the analogous problem of \ExactTri{} in preprocessed universes:
 we can preprocess any weighted $n$-node graph in $\OO(n^3)$ time, so that
for any given query subgraph, we can solve \ExactTri{} in time truly subcubic in $n$,
namely, $O(n^{2.83})$ (or $\OO(n^{11/4})$ if $\omega=2$).
This result can be viewed as a generalization of the result for \ThreeSUM{}, by 
known reductions from \ExactTri{} to \ThreeSUM{} (for integers).

\paragraph{3SUM for monotone sets in $[n]^d$.}
Chan and Lewenstein~\cite{ChanLewenstein} also studied a special case of \ThreeSUM{} for
monotone sets $A,B,C$ in $[n]^d$ for a constant dimension~$d$, where a set of points is \emph{monotone} if it is of the form $\{a_1,\ldots,a_n\}$ where the $j$-th coordinates of  $a_1,\ldots,a_n$ is a monotone sequence for each $j$.  The problem is related to
an important special case of Min-Plus convolution for monotone sequences in $[n]$ (or integer sequences with bounded differences), which reduces to  monotone \ThreeSUM{} in 2 dimensions.  This is also  related to a data structure problem for strings known as \emph{jumbled indexing}, where %
$d$ corresponds to the alphabet size.  Chan and Lewenstein gave the first truly subquadratic algorithm for the problem, with running time of the form $O(n^{2-1/(d+O(1))})$ using randomization.  (See \cite{AmirCLL14,HsuUmans17} for conditional lower bounds on this and related problems.)  However, they obtained subquadratic deterministic algorithms only for %
$d\le 7$ under the current fast matrix multiplication bounds.  Our techniques give the first \emph{deterministic} algorithm for all constant $d$, with running time $O(n^{2-1/O(d)})$.
Although the new bound is not better (and for monotone Min-Plus convolution, a recent paper by Chi, Duan, Xie and Zhang~\cite{ChiDXZstoc22} presented even faster randomized algorithms), the new approach is simpler, besides being deterministic.

\subsection{Paper Organization}

In Section~\ref{sec:prelim}, we define notations and problems. 
The rest of the paper has three main threads:

\begin{itemize}
\item \emph{Conditional lower bounds for problems with intermediate complexity}:
In Section~\ref{sec:intapsp-lower-bound}, we illustrate our approach by proving
the first superquadratic lower bound for \uAPSP{} under the \StrongAPSP{}.
In Sections~\ref{sec:intapsp-lower-bound:more}--\ref{sec:uapsp-lower-bound}, we prove lower bounds for other problems, including
\MinWitness{}, under both the \StrongAPSP{} 
 and \uAPSPH{}.
\item \emph{Equivalences between counting and detection problems}:
In Section~\ref{sec:counting:preview}, we illustrate our basic idea by proving
the subcubic equivalence between \ExactTriCount{} and \ExactTri{}.  In Section~\ref{sec:counting}, we prove more results of this kind, including the subquadratic equivalence between \ThreeSUMCount{} and \ThreeSUM{}.  (Still more examples can be
found in Appendix~\ref{sec:more_counting}\@.)  
In Section~\ref{sec:other_models}, we further adapt these ideas to obtain new nondeterministic and quantum algorithms for counting problems.
\item \emph{BSG-related theorems}: In Section~\ref{sec:decomposition-zero-tri}, we present our new Triangle Decomposition Theorem and describe its applications. 
In Section~\ref{sec:bsg}, we describe our new variants of the BSG theorem.  (Still more applications and variants can be found in Appendices~\ref{app:bd:diff} and \ref{app:bsg}\@.)  Finally, in Section~\ref{sec:min-equal-conv}, we prove the conditional lower bounds for \MinEqualityConv{}, the most sophisticated of which use one of our new BSG theorems. 
\end{itemize}

Although the paper is lengthy,
Sections~\ref{sec:intapsp-lower-bound}, \ref{sec:counting:preview} and
\ref{sec:decomposition-zero-tri} should suffice to give the readers an overview of our main proof techniques.  (Readers interested in diving deeper into any of the above threads may proceed to the subsequent sections.)

%% file: previous_work_figure.tex
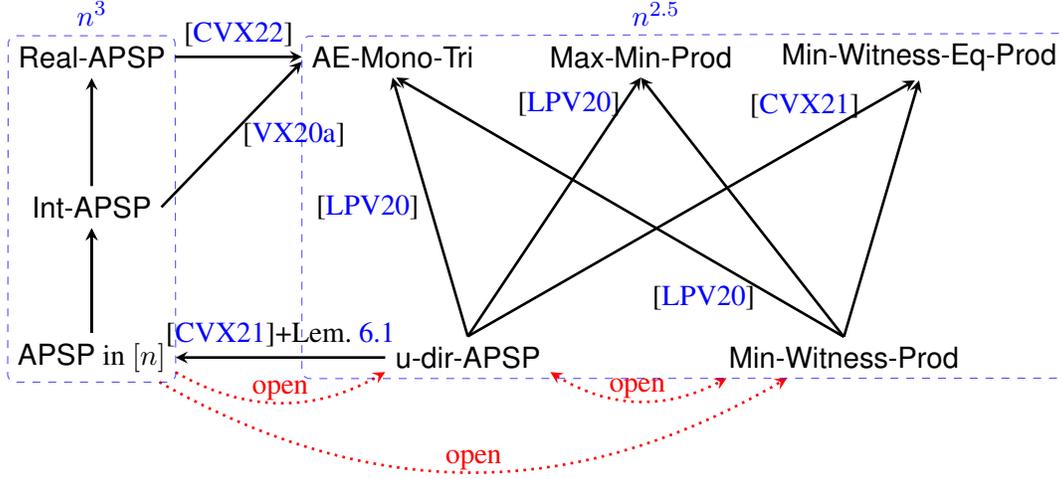
\begin{figure}[ht]
    \centering
    \begin{tikzpicture}
    
    \node at(0, 0)  [anchor=center] (apspn){\APSP{} in $[n]$};
    \node at(0, 2)  [anchor=center] (intapsp){\IntAPSP{}};
    \node at(0, 4)  [anchor=center] (realapsp){\RealAPSP{}};
    \node at(realapsp.north)  [anchor=south] (){\textcolor{blue}{$n^3$}};

    \node at(5, 0)  [anchor=center] (uapsp){\uAPSP{}};
    \node at(10, 0)  [anchor=center] (minwitness){\MinWitness{}};
    
    \node at(4, 4)  [anchor=center] (aemonotri){\AEMonoTri{}};
    \node at(7.3, 4)  [anchor=center] (maxmin){\MaxMin{}};
    \node at(11, 4)  [anchor=center] (minwitnesseq){\MinWitnessEq{}};
    \node at(7.5,0 |- maxmin.north)  [anchor=south] (){\textcolor{blue}{$n^{2.5}$}};
	
	\draw[-stealth,line width=1pt] (apspn.north) to[]  node[right] {} (intapsp.south);
	\draw[-stealth,line width=1pt] (intapsp.north) to[]  node[right] {} (realapsp.south);
	\draw[-stealth,line width=1pt] (apspn.north) to[]  node[right] {} (intapsp.south);
	\draw[-stealth,line width=1pt] (minwitness.north) to[]  node[pos=0.15,  left, xshift=-4] {\cite{lincoln2020monochromatic}} (aemonotri.280);
	\draw[-stealth,line width=1pt] (minwitness.north) to[]  node[right] {} (maxmin.south);
	\draw[-stealth,line width=1pt] (minwitness.north) to[]  node[right] {} (minwitnesseq.south);
	\draw[-stealth,line width=1pt] (uapsp.north) to[]  node[left] {\cite{lincoln2020monochromatic}} (aemonotri.south);
	\draw[-stealth,line width=1pt] (uapsp.north) to[]  node[left, pos = 0.9, xshift=3] {\cite{lincoln2020monochromatic}} (maxmin.south);
	\draw[-stealth,line width=1pt] (uapsp.north) to[]  node[left, pos = 0.9] {\cite{CVXicalp21}} (minwitnesseq.260);

	\draw[-stealth,line width=1pt] (realapsp.east) to[]  node[above] {\cite{CVXstoc22}} (aemonotri.west);
	\draw[-stealth,line width=1pt] (intapsp.east) to[]  node[right] {\cite{williamsxumono}} (aemonotri.183);
	\draw[-stealth,line width=1pt] (uapsp.west) to[]  node[above] {\cite{CVXicalp21}+Lem.~\ref{lem:down}} (apspn.east);

	\draw[-stealth,line width=1pt, dotted, color = red, bend right] (apspn.350) to[]  node[above, yshift=-3] {open} (uapsp.190);
	\draw[stealth-stealth,line width=1pt, dotted, color = red, bend right] (uapsp.350) to[]  node[above, yshift=-3] {open} (minwitness.190);
	\draw[-stealth,line width=1pt, dotted, color = red, bend right] (apspn.340) to[]  node[above, yshift=-3] {open} (minwitness.200);

    \draw[opacity=0.7, dashed, rounded corners=3, color=blue] (realapsp.north east) -- (realapsp.north west) -- (apspn.south west) -- (apspn.south east) -- cycle;
    \draw[opacity=0.7, dashed, rounded corners=3, color=blue] (minwitnesseq.north east) -- (aemonotri.north west) -- (aemonotri.north west |- uapsp.south west) -- (minwitnesseq.north east |- minwitness.south east) -- cycle;

    \end{tikzpicture}
    \vspace*{-2ex}
    \caption{Previous work on intermediate matrix product problems, assuming $\omega=2$. Here,  \APSP{} in $[n]$ denotes %
    APSP with edge weights in $[n]$, \IntAPSP{} denotes APSP with arbitrary $O(\log n)$-bit integer weights, and \RealAPSP{} denotes APSP with real weights.    All unlabelled arrows follow by trivial reductions.}
    \label{fig:previous}
\end{figure}

%% file: bounded_min_plus_table.tex
\begin{table*}[ht]
\footnotesize
\centering
\begin{tabular}{|c|cr|cr|cr|}
\hline
\diagbox{Problems}{Lower Bounds}{Hypotheses} & \multicolumn{2}{c|}{\StrongAPSPShort{}}  & \multicolumn{2}{c|}{u-dir-APSP} & \multicolumn{2}{c|}{\begin{tabular}{@{}p{2.3cm}@{}}\StrongConvShort{}\end{tabular}}\\
\hline
\uAPSP{} & $n^{7/3-o(1)}$ & Cor.~\ref{cor:strong-intapsp-imply} & $n^{5/2-o(1)}$ & trivial &  \multicolumn{2}{c|}{N/A} \\
\hline
\MinWitness{} &  $n^{11/5-o(1)}$ & Cor.~\ref{cor:strong-intapsp-imply:minwit} & $n^{7/3-o(1)}$ & Cor.~\ref{cor:lower-bounds-under-uAPSP:minwit}& \multicolumn{2}{c|}{N/A}\\
\hline
\APSLPIntro{} & $n^{11/5-o(1)}$ & Cor.~\ref{cor:strong-intapsp-imply:apslp}& $n^{7/3-o(1)}$ & Cor.~\ref{cor:lower-bounds-under-uAPSP:apslp} &  \multicolumn{2}{c|}{N/A}\\
\hline
\BatchMode{} & $n^{7/6-o(1)}$ & Cor.~\ref{cor:strong-intapsp-imply:mode}&  $n^{5/4-o(1)}$ & Cor.~\ref{cor:lower-bounds-under-uAPSP:mode} & \multicolumn{2}{c|}{N/A}\\
\hline
\MinEqualityConv{} & $n^{7/6-o(1)}$ & Thm.~\ref{thm:min-equal-conv-under-product} & $n^{5/4-o(1)}$ & Thm.~\ref{thm:min-equal-conv-under-uAPSP} &  $n^{12/11-o(1)}$ & Thm.~\ref{thm:min-equal-conv-under-conv}\\ 
\hline
\end{tabular}
\caption{New conditional lower bounds obtained under the \StrongAPSP{}, the \uAPSPH{} and the 
\StrongConv{}, %
assuming $\omega = 2$.  (We also have nontrivial lower bounds %
even if $\omega>2$; see the corresponding theorems/corollaries for the precise bounds.)}
\label{tab:bounded-min-plus}
\end{table*}

%% file: prelim.tex
For integer $n \ge 1$, we use $[n]$ to denote $\{1, 2, \ldots, n\}$ and use $[\pm n]$ to denote $\{-n, -(n-1), \ldots, n\}$. 
For a predicate $X$, we use $[X]$ to denote a $\{0, 1\}$ value which evaluates to $1$ if $X$ is true, and $0$ otherwise. 

We use $M(n_1, n_2, n_3)$ to denote the (randomized) running time to multiply an $n_1 \times n_2$ matrix with an $n_2 \times n_3$ matrix.
We consider the Word-RAM model of computation with $O(\log n)$-bit words, and all input numbers are $O(\log n)$-bit integers, unless otherwise stated.

\begin{problem}[\MinPlus{}]
Given an $n_1 \times n_2$ matrix $A$ and an $n_2 \times n_3$ matrix $B$, compute the $n_1 \times n_3$ matrix $C$ where $C_{ij} = \min_{k \in [n_2]} (A_{ik}+B_{kj})$. We denote $C$ by $A \star B$. 
\end{problem}

Furthermore, we use $M^*(n_1, n_2, n_3 \mid \ell)$ to denote the (randomized) running time to compute the Min-Plus product between an $n_1 \times n_2$ matrix and an $n_2 \times n_3$ matrix where the entries of both matrices  are from $[\ell] \cup \{\infty\}$.

\begin{problem}[\APSP{}]
Given an edge-weighted $n$-node graph, compute its all-pairs shortest path distances. 
\end{problem}

\begin{problem}[\MinPlusConv{}]
Given two length $n$ arrays $A, B$, compute the length $n$ array $C$, where $C_i = \min_{k \in [i-1]} (A_k + B_{i-k})$.
\end{problem}

\begin{problem}[\ExactTri{}]
Given an edge-weighted $n$-node graph $G$ and a target value $t$, determine whether $G$ contains a triangle whose edge weights sum up to $t$. In its All-Edges version (\AEExactTri{}), we need to determine for every edge $e$, whether $e$ is contained in a triangle whose edge weights sum up to $t$. 
\end{problem}

\begin{problem}[\NegTri{}]
Given an edge-weighted $n$-node graph $G$, determine whether $G$ contains a triangle whose edge weights sum up to a negative value. In its All-Edges version (\AENegTri{}), we need to determine for every edge $e$, whether $e$ is contained in a triangle whose edge weights sum up to a negative value. 
\end{problem}

\begin{problem}[\ExactKClique{}]
Given an edge-weighted $n$-node graph $G$ and a target value $t$, determine whether $G$ contains a $k$-clique whose edge weights sum up to $t$. 
\end{problem}

\begin{problem}[\MinKClique{}]
Given an edge-weighted $n$-node graph $G$, determine the minimum edge weight sum over all $k$-cliques in $G$. 
\end{problem}

\begin{problem}[\ThreeSUM{}]
Given three sets of numbers $A, B, C$ of size $n$, determine if there exist $a \in A, b \in B, c \in C$ such that $a + b  = c$. In its All-Numbers version (\AllThreeSUM{}), we need to determine for each $c \in C$, whether there exist $a \in A, b \in B$ such that $a + b  = c$.
\end{problem}

\begin{problem}[\ThreeSUMConv{}]
Given three arrays of numbers $A, B, C$ of lengths $n$, determine if there exist $i, j \in [n]$ such that $A_i + B_j = C_{i + j}$. In its All-Numbers version (\AllThreeSUMConv{}), we need to determine for each $k \in [n]$, whether there exist $i, j \in [n]$ such that $i + j = k$ and $A_i + B_j = C_k$. 
\end{problem}

\begin{problem}[\MonoConv{}]
Given three arrays of numbers $A, B, C$ of length $n$, determine for every $k \in [n]$, whether there exist $i, j \in [n]$ such that $i + j = k$ and $A_i = B_j = C_k$. 
\end{problem}

All the problems defined above have natural counting variants, which will be denoted by adding a ``\#'' prefix to the problems' names. For a detection problem that outputs a single bit or multiple bits, each bit represents whether an object that satisfies some requirements exists among a set of candidates; in its counting variant, we need to output a number in place of each bit to denote the number of objects that satisfy the requirements among the same set of candidates. For instance, \ExactTriCount{} asks to count the number of triangles in the graph whose edge weights sum up to $t$, and \AEExactTriCount{} asks to count, for each edge $e$, the number of triangles containing $e$ whose edge weights sum up to $t$. 
For a minimization problem that outputs one or multiple optimal values, its counting variant needs to output the number of optimal solutions in place of each optimal value. For instance, \MinPlusCount{} asks for each $(i, j) \in [n_1] \times [n_3]$, the number of $k$ where $A_{ik}+B_{kj}=(A\star B)_{ij}$. One special problem is the \APSPCountMod{U} problem, in which we need to compute the number of shortest paths mod $U$ for every pair of nodes. 

Clearly, for any of the detection problems, its counting variant is (not necessarily strictly) harder than its original version. The same is not obviously true for minimization problems.

\begin{problem}[\uAPSP{}]
Given an unweighted directed  graph, compute its all-pairs shortest path distances. 
\end{problem}

\begin{problem}[\MinWitness{}]
Given an $n_1 \times n_2$ Boolean matrix $A$ and an $n_2 \times n_3$ Boolean matrix  $B$, compute $\min \{k \in [n_2]: A_{ik} \wedge B_{kj}\}$ 
for every $(i, j) \in [n_1] \times [n_3]$.
\end{problem}

\begin{problem}[\Equality{}]
Given an $n_1 \times n_2$ matrix $A$ and an $n_2 \times n_3$ matrix $B$, compute the number of  $k \in [n_2]$ such that $A_{ik}=B_{kj}$ for every $(i, j) \in [n_1] \times [n_3]$.
\end{problem}

\begin{problem}[\MinWitnessEq{}]
Given an $n_1 \times n_2$ matrix $A$ and an $n_2 \times n_3$ matrix $B$, compute $\min \{k \in [n_2]: A_{ik} = B_{kj}\}$ 
for every $(i, j) \in [n_1] \times [n_3]$.
\end{problem}

\begin{problem}[\MinEqualityProd{}]
Given an $n_1 \times n_2$ matrix $A$ and an $n_2 \times n_3$ matrix $B$, compute $\min \{A_{ik}: k \in [n_2] \wedge A_{ik} = B_{kj}\}$ 
for every $(i, j) \in [n_1] \times [n_3]$.
\end{problem}

\begin{problem}[\APSLP{}]
Given an $n$-node undirected graph whose edge weights are either $1$ or $2$, compute for every pair of nodes $s, t$, the smallest number of edges required to travel from $s$ to $t$, and the minimum weight over all paths using the smallest number of edges. 
\end{problem}

\begin{problem}[\APLSP{}]
Given an $n$-node undirected graph whose edge weights are either $1$ or $2$, compute for every pair of nodes $s, t$, the shortest path distance from $s$ to $t$, and the smallest number of edges among all shortest paths. 
\end{problem}

\begin{problem}[\BatchMode{}]
Given an length $N$ array and $Q$ intervals of the array, compute the element that appears the most (breaking ties arbitrarily) for each of the intervals. 
\end{problem}

\begin{problem}[\MinEqualityConv{}]
Given two length $n$ arrays $A, B$, compute the length $n$ array $C$, where $C_i = \min\{A_j: j \in [i-1] \wedge A_j = B_{i-j}\}$. 
\end{problem}

%% file: int_apsp_reduce_preview.tex
In this section, we apply the idea of combining Fredman's trick with Equality Product to prove conditional lower bounds
for problems with intermediate complexity.
To illustrate the idea,
we focus on lower bounds for the \uAPSP{} problem
under the \StrongAPSP{}, but the approach also leads to lower bounds for \MinWitness{} and other problems, under the \StrongAPSP{} as well as the
\uAPSPH{}, as we will explain later in Sections~\ref{sec:intapsp-lower-bound:more}--\ref{sec:uapsp-lower-bound}.
As noted earlier, the \StrongAPSP{} is equivalent to the hypothesis
that $M^*(n,n,n\mid n^{3-\omega})$ is not truly subcubic.
Thus, it suffices to describe fine-grained reductions from the \MinPlus{} problem for two $n\times n$ matrices 
with bounded integer entries from $[n^{3-\omega}]$, to the \uAPSP{}  problem.

When devising a reduction from one 
problem to another (as in typical NP-hardness proofs), we often concentrate on understanding the power of the latter problem.
In our new reductions, we will focus mostly on the former problem instead, interestingly:
we will attempt to design an algorithm to solve the \MinPlus{} problem for bounded integers in subcubic time, and 
only at the end, reveal how an oracle for \uAPSP{} (or \MinWitness{} and other problems) could help.

\subsection{Preliminaries: Generalized Equality Products}\label{sec:prelim:eq:prod}

By Matou\v sek's technique for dominance product~\cite{MatIPL, YusterDom},
the equality product of an $n_1\times n_2$
and an $n_2\times n_3$ matrix can be computed in
time $\OO\big(\min_r (n_1n_2n_3/r \,+\, M(n_1,rn_2,n_3)) \big)$.
For example, in the case $n_1=n_2=n_3=n$, the bound is at most $\OO(\min_r (n^3/r + rn^\omega))=\OO(n^{(3+\omega)/2})$, as we have stated before, if we don't use rectangular matrix multiplication exponents.
We begin with the following lemma describing a straightforward generalization, which will be useful later:

\begin{lemma}\label{lem:geneqprod}
Given $n_1\times n_2$ matrices $A$ and $A'$, and
$n_2\times n_3$ matrices $B$ and $B'$,
define the \emph{generalized equality product} of $(A,A')$ and $(B,B')$ to be the following $n_1\times n_3$ matrix $E$:
\[ E_{ij}\ := \min_{k:\,A_{ik}=B_{kj}} (A'_{ik} + B'_{kj}).
\]
Suppose that all matrix entries of $A'$ and $B'$ are in $[\pm \ell] \cup \{\infty\}$.  For any $r$, we can compute $E$ in time
\[ \OO\big(n_1n_2n_3/r \,+\, M^*(n_1,rn_2,n_3\mid \ell) \big).
\]
\end{lemma}
\begin{proof}
Sort $B_k=(B_{kj})_{j\in[n_3]}$, i.e., the $k$-th column of $B$.
Let $F_k$ be the set of elements that have frequency more than $n_3/r$ in $B_k$.
Note that $|F_k|\le r$.
We divide into two cases, computing two $n_1\times n_3$ matrices
$E^L$ and $E^H$;
the answers will be $E_{ij}=\min\{E_{ij}^L,E_{ij}^H\}$.

\begin{itemize}
\item[\bf--]
{\bf Low-frequency case: computing $\displaystyle E_{ij}^L=\min_{k:\, A_{ik}=B_{kj}\not\in F_k} (A'_{ik} + B'_{kj})$.}
Initially set $E_{ij}^L=\infty$.
For each $i\in [n_1]$ and $k\in [n_2]$,
if $A_{ik}\not\in F_k$, we examine each of the at most $n_3/r$ indices $j$ with $A_{ik}=B_{kj}$, and
reset $E_{ij}^L=\min\{E_{ij}^L,\, A'_{ik}+B'_{kj}\}$.
All this takes $O(n_1n_2\cdot n_3/r)$ time.

\item[\bf--]
{\bf High-frequency case: computing $\displaystyle E_{ij}^H=\min_{k:\, A_{ik}=B_{kj}\in F_k} (A'_{ik} + B'_{kj})$.}
For each $i\in [n_1]$, $k\in[n_3]$, and $p\in F_k$,
let $\hat{A}_{i,(k,p)}=A'_{ik}$ if $A_{ik}=p$, and $\hat{A}_{i,(k,p)}=\infty$ otherwise.
For each $k\in[n_3]$, $j\in [n_2]$, and $p\in F_k$,
let $\hat{B}_{(k,p),j}=B'_{kj}$ if $B_{kj}=p$, and $\hat{B}_{(k,p),j}=\infty$ otherwise.
We let $E_{ij}^H = \min_{k\in [n_2],p\in F_k} (\hat{A}_{i,(k,p)} + \hat{B}_{(k,p),j})$.
This can be computed by a Min-Plus product in 
$O(M^*(n_1,rn_2,n_3\mid \ell))$ time.
\end{itemize}

\negbigskip
\end{proof}

\subsection{The Key Reduction}

We now present an approach to solve the \MinPlus{} problem for two matrices with
integer entries in $[\ell]$, by reducing it to 
instances that have simultaneously
a smaller inner dimension~$s$ and smaller integer entries in $[t]$ with $t\le \ell$: 

\begin{theorem}\label{thm:main} 
For any $r,s,t$ with $s\le n_2$ and $t\le \ell$,
\[ M^*(n_1,n_2,n_3\mid\ell)\ =\ 
\OO\big( (n_2/s) M^*(n_1,s,n_3\mid t)
\,+\, sn_1n_2n_3/r \,+\,  sM^*(n_1,rn_2,n_3\mid \ell/t) \big).
\]
\end{theorem}
\begin{proof}
Let $g=\lceil \ell/t\rceil$.
Let $A$ be an $n_1\times n_2$ matrix and $B$ be an $n_2\times n_3$ matrix, where all matrix entries are in $[\ell]\cup\{\infty\}$.
We describe an algorithm to compute the Min-Plus product of $A$ and $B$.
Without loss of generality, we may assume that $(A_{ik}\bmod g) < g/2$ for all $i,k$ with $A_{ik}$ finite, since we can separate the problem into two instances $(A^<,B)$ and $(A^\ge,B)$ for
two matrices $A^<$ and $A^\ge$, where
$A^<_{ik} = A_{ik}$ and $A^{\ge}_{ik}=\infty$ if $(A_{ik}\bmod g) < g/2$,
and $A^{<}_{ik} = \infty$ and $A^{\ge}_{ik}=A_{ik}-g/2$ if $(A_{ik}\bmod g) \ge g/2$.
Similarly, we may assume that $(B_{kj}\bmod g) < g/2$ for all $k,j$ with $B_{kj}$ finite.

For each $i,k$, write $A_{ik}$ as $A'_{ik}g + A''_{ik}$
with $0 \le A'_{ik} \le t$ and $0 \le A''_{ik} < g/2$.  
Similarly, for each $k,j$, write $B_{kj}$ as $B'_{kj}g + B''_{kj}$
with $0 \le B'_{kj} \le t$ and $0 \le B''_{kj} < g/2$.
(Set $A'_{ik}=A''_{ik}=\infty$ if $A_{ik}=\infty$, and $B'_{kj}=b''_{kj}=\infty$ if $B_{kj}=\infty$.)

We first compute the Min-Plus product $C'$ of $A'$ and $B'$ (i.e., $C'_{ij} = \min_{k} (A'_{ik}+B'_{kj})$), in time
$O(M^*(n_1,n_2,n_3\mid t)) \le O( (n_2/s)\cdot M^*(n_1,s,n_3\mid t))$.

Let $W_{ij}=\{k\in [n_2]: A'_{ik}+B'_{kj}=C'_{ij}\}$; the elements in $W_{ij}$
are the \emph{witnesses} for $C'_{ij}$.
The Min-Plus product $C$ of $A$ and $B$ is given by $C_{ij} = C'_{ij}g + C''_{ij}$, where
\[ C''_{ij}\: :=\: \min_{k\in W_{ij}} (A''_{ik} + B''_{kj}).
\]
It suffices to describe how to compute $C''$.  We divide into two cases:

\begin{itemize}
\item
{\bf Few-witnesses case: computing $C''_{ij}$ for all $i,j$ with $|W_{ij}|\le n_2/s$.}
For each such $(i,j)$, we will explicitly enumerate all witnesses in $W_{ij}$.
This can be done by standard techniques for witness finding~\cite{AlonGMN92, seidel1995}:
first, observe that if the witness is unique (i.e., $|W_{ij}|= 1$), it can be found
by performing $O(\log n_2)$ Min-Plus products (namely,
for each $\ell\in [\log n_2]$, the $\ell$-th bit of
the witness for $C'_{ij}$ is 1 iff
$\min_{k\in K_\ell} (A'_{ik}+B'_{kj})=C_{ij}$, where $K_{\ell}:=\{k\in [n_2]:
\mbox{the $\ell$-th bit of $k$ is 1}\}$).
We take a random subset $R\subseteq [n_2]$ of $s$ indices,
and find witnesses for the Min-Plus product of
$(A'_{ik})_{i\in [n_1],k\in R}$ and 
$(B'_{kj})_{k\in R,j\in [n_3]}$ if the witnesses are unique.
This takes $\OO(M^*(n_1,s,n_3\mid t))$ time.
Fix $i,j$ with $|W_{ij}|\le n_2/s$.
For a fixed element $w\in W_{ij}$, the probability that $w$ is found,
i.e., $w$ is in $R$ but no
other element of $W_{ij}$ is in $R$, is $\Omega((s/n_2)\cdot (1-s/n_2)^{n_2/s})
=\Omega(s/n_2)$.  By repeating $O((n_2/s)\log(n_1n_2n_3))$ times (with
different choices of $R$),
all witnesses in $W_{ij}$ would be found w.h.p.
Once the entire witness set $W_{ij}$ is found, we can compute each $C''_{ij}$ naively in 
$O(|W_{ij}|)$ time.
The total running time is $\OO((n_2/s)\cdot M^*(n_1,s,n_3\mid t))$.

\item
{\bf Many-witnesses case: computing $C''_{ij}$ for all $i,j$ with $|W_{ij}| > n_2/s$.}
Pick a random subset $H$ of size $c_0s\log(n_1n_2n_3)$ for a sufficiently large constant~$c_0$.  Then $H$ hits (i.e., intersects) every witness set $W_{ij}$ with $|W_{ij}|>n_2/s$ w.h.p.

We do the following: for each $k_0\in H$ and
for each $i\in[n_1],j\in [n_3]$, if 
$A'_{ik_0}+B'_{k_0j}=C'_{ij}$ (i.e., $k_0\in W_{ij}$), set
\begin{equation}\label{eqn:eqprod}
C''_{ij}\ = \min_{k:\, A'_{ik}-A'_{ik_0}=B'_{k_0j}-B'_{kj}} (A''_{ik} + B''_{kj}).
\end{equation}
Correctness of (\ref{eqn:eqprod}) follows immediately from ``Fredman's trick'': $A'_{ik}-A'_{ik_0}=B'_{k_0j}-B'_{kj}$ is equivalent to $A'_{ik}+B'_{kj}=A'_{ik_0}+B'_{k_0j}$, which is equivalent to $k\in W_{ij}$, assuming
$A'_{ik_0}+B'_{k_0j}=C'_{ij}$.
Thus, the above correctly computes $C''_{ij}$ for every $i,j$ with $|W_{ij}|>n_2/s$, since
$H$ hits $W_{ij}$.

Finally, we observe that for a fixed $k_0$,
the right-hand side in (\ref{eqn:eqprod}) corresponds
precisely to a generalized equality product!  By
Lemma~\ref{lem:geneqprod}, they can be computed in
$\OO(n_1n_2n_3/r + M^*(n_1,rn_2,n_3\mid g))$ time, for each of
the $\OO(s)$ choices of $k_0$.
\end{itemize}

\negbigskip
\end{proof}

Note that Fredman's trick was originally introduced to solve APSP or compute Min-Plus products for arbitrary \emph{real}-valued inputs.  It is interesting that the trick is useful even when input values are in a restricted integer range
(in $[t]$).

Note also that a more naive attempt to prove the above theorem is to just bound $M^*(n_1,n_2,n_3\mid \ell)$
by $(n_2/s)M^*(n_1,s,n_3\mid \ell)$,
and once the inner dimension $n_2$ is reduced to $s$ in the subproblems, 
try to use hashing to reduce the range of the integers to, say, $[\OO(s^2)]$.  However, while such a hashing approach
might work for equality-type problems (e.g., \AEExactTri{}), it does not work at all
for \MinPlus{}.

\subsection{Consequences}

In Theorem~\ref{thm:main}, we can directly bound the third term by using existing 
matrix multiplication results~\cite{AlonGMN92}, leading to the following corollary:

\begin{corollary}\label{cor:intapsp0}
For any constant $0<\beta\le (3-\omega)/2$,
if $M^*(n,n^\beta,n\mid n^{2\beta})=O(n^{2+\beta-\eps})$,
then $M^*(n,n,n\mid n^{3-\omega})=O(n^{3-\Omega(\eps)})$.
\end{corollary}
\begin{proof}
By Theorem~\ref{thm:main},
\begin{eqnarray*}
 M^*(n,n,n\mid n^{3-\omega}) &=& \OO\left((n/s)M^*(n,s,n\mid t) + sn^3/r + (sn^{3-\omega}/t) M(n,rn,n)\right)\\
&\le&  \OO\left((n/s)M^*(n,s,n\mid t) + sn^3/r + rsn^3/t\right).
\end{eqnarray*}
Setting $r=n^\beta$, $s=n^{\beta-\eps'}$, and $t=n^{2\beta}$ with $\eps'=\eps/2$ yields $M^*(n,n,n\mid n^{3-\omega})=O(n^{3-\Omega(\eps)})$.
\end{proof}

The above corollary
establishes
a conditional lower bound of $n^{2+\beta-\eps}$ for
 the subproblem
of Min-Plus product for rectangular matrices of dimension
$n\times n^\beta$ and $n^\beta\times n$ for integers bounded by $n^{2\beta}$, under the \StrongAPSP{}.
This lower bound is tight in the sense that $O(n^{2+\beta})$ is
an obvious upper bound (though the range of allowed integer values $[n^{2\beta}]$ may not
be tight).  We will now use this corollary to derive conditional
lower bounds for \uAPSP{}.

Let $\TunwtdirAPSP(n)$ be the time complexity of \uAPSP{} on $n$-node graphs.
More generally, let $\TunwtdirAPSP(n,m)$ be the time complexity of APSP on an unweighted directed graph with $n$ nodes and $m$ edges.
Chan, Vassilevska W., and Xu~\cite{CVXicalp21} have given
a simple reduction of Min-Plus product for rectangular matrices of certain inner dimensions and
integer ranges to \uAPSP{}, as summarized by the following lemma
(it is easy to check that the graph in their reduction has $O(nx)$ edges).

\begin{lemma}\label{lem:reduce:from:minplus}
For any $x,y$, we have
$M^*(n,x,n\mid y)= O(\TunwtdirAPSP(n,nx))$ if $xy\le n$.
\end{lemma}

Combining Corollary~\ref{cor:intapsp0} and Lemma~\ref{lem:reduce:from:minplus} immediately gives the following:

\begin{corollary}\label{cor:strong-intapsp-imply}
For any constant $\beta\le \min\{1/3, (3-\omega)/2\}$, if $\TunwtdirAPSP(n,n^{1+\beta})=O(n^{2+\beta-\eps})$, 
then $M^*(n,n,n\mid n^{3-\omega}) = O(n^{3-\Omega(\eps)})$.
\end{corollary}

By setting $\beta=1/3$, we have thus proved that \uAPSP{} cannot be solved
in $O(n^{7/3-\eps})$ time under the \StrongAPSP{}, assuming $1/3\le (3-\omega)/2$ (this assumption can be removed, as we observe later in Section~\ref{sec:uapsp-lower-bound}).  In particular, if $\omega=2$, this implies that
 \uAPSP{} is strictly harder than unweighted undirected APSP,
as the latter problem can be solved in $\OO(n^\omega)$ time~\cite{seidel1995}.

Furthermore, \uAPSP{} for a graph with $n^{1+\beta}$ edges cannot be solved in $O(n^{2+\beta-\eps})$ time for any $\beta\le 1/3$ under the same hypothesis, assuming $1/3\le (3-\omega)/2$ (again this assumption can be removed).  In other words, the naive algorithm by repeated BFSs is essentially \emph{optimal} for sufficiently sparse graphs.

If we assume a weaker hypothesis that APSP does not
have truly subcubic algorithms for edge weights in $[n]$ 
instead of $[n^{3-\omega}]$, it can be checked that we still
get a lower bound near $n^{2+(3-\omega)/3}$ for \uAPSP{}.
In fact, assuming that APSP does not
have truly subcubic algorithms for edge weights in $[n^\lambda]$,
we can still obtain a super-quadratic lower bound for \uAPSP{} for
$\lambda$ as large as 1.99, if $\omega=2$.
For simplicity, we will concentrate only on the version of the \StrongAPSP{} with $\lambda=3-\omega$ throughout the paper.

In Sections~\ref{sec:intapsp-lower-bound:more}--\ref{sec:uapsp-lower-bound}, we will use the same approach to derive further conditional lower bounds for \MinWitness{}, \APLSP{}, and \BatchMode{},
from both the \StrongAPSP{} and the \uAPSPH{}.

%% file: counting_preview.tex
In this section, we describe a simple approach to proving equivalence between
counting and detection problems, by combining Fredman's trick with Equality Product.  To illustrate the basic idea, we focus on
the equivalence of \AEExactTriCount{} and \AEExactTri{}.  With more work and additional ideas, the approach can also establish the equivalence of \ThreeSUMCount{} and \ThreeSUM{}, as we will later explain in Section~\ref{sec:counting}.

We use $G$ to denote the input graph of an \AEExactTri{}  or \AEExactTriCount{} instance, we use $w$ to denote the weight function, and we use $W_{ij}$ to denote the set of $k$ where $(i, j, k)$ forms a triangle whose edge weights sum up to the target value $t$.

\begin{lemma}
\label{lem:exact-tri}
Given an $n$ node graph $G$, a target value $t$, and a subset $S \subseteq V(G)$, we can compute a matrix $D$ in $\OO(|S| \cdot n^{(3+\omega)/2})$ time such that $D_{ij} = |W_{ij}|$ whenever $S \cap W_{ij} \ne \emptyset$. 
\end{lemma}
\begin{proof}
For every $s \in S$, we do the following. Let $A^{(s)}$ be a matrix where $A^{(s)}_{ik} = w(i, k) - w(i, s)$ and $B^{(s)}_{kj} = w(s, j) - w(k, j)$. Then we use compute the equality product $C^{(s)}$ of $A^{(s)}$ and $B^{(s)}$ in $\OO(n^{(3+\omega)/2})$ time for each $s$ \cite{MatIPL}. Finally, if there exists $s \in S$ such that $A_{is} + B_{sj} + w(i, j) = t$, we let $D_{ij}$ be $C^{(s)}_{ij}$ for an arbitrary $s$ with the property; otherwise, we let $D_{ij}$ be $0$ (we don't care about its value in this case). The running time for computing $D$ is clearly $\OO(|S| \cdot n^{(3+\omega)/2})$. 

Suppose $S \cap W_{ij} \ne \emptyset$ for some $(i, j)$. Then $D_{ij}$ equals $C^{(s)}_{ij}$ for some $s$ where $A_{is} + B_{sj} + w(i, j) = t$. By Fredman's trick, $A^{(s)}_{ik} = B^{(s)}_{kj}$ if and only if $w(i, k) + w(k, j) + w(i, j) = w(i, s) + w(s, j) + w(i, j) = t$. Therefore, $D_{ij} = C^{(s)}_{ij} = |W_{ij}|$.
\end{proof}

\begin{theorem}
\label{thm:exact-tri-count}
If \AEExactTri{} for $n$-node graphs has an $O(n^{3-\eps})$ time algorithm for some $\eps > 0$, then \AEExactTriCount{} for $n$-node graphs  has an $O(n^{3-\eps'})$ time algorithm for some $\eps' > 0$
\end{theorem}

\begin{proof}
Given a \AEExactTriCount{} instance on an $n$-node graph $G$, we first list up to $n^{0.99}$ elements in $W_{ij}$ for every $i, j$. 
By well-known techniques (e.g. \cite{focsyj}), an $O(n^{3-\eps})$ time \AEExactTri{} algorithm implies an $O(n^{3-\eps''})$ for $\eps'' > 0$ time algorithm for listing up to $n^{0.99}$ witnesses for each $(i, j)$ in an \AEExactTri{} instance. 

If we list less than $n^{0.99}$ elements for some $(i, j)$, we can output the number of elements we list as the exact witness count for $(i, j)$. By the standard greedy algorithm for hitting set, in $\OO(n^{2.99})$ time, we can find a set $S$ of size $\OO(n^{0.01})$ that intersect with $W_{ij}$ for the remaining pairs $(i, j)$. Therefore, we can apply Lemma~\ref{lem:exact-tri} to compute the number of witnesses for these remaining $(i, j)$ pairs in $\OO(|S| \cdot n^{(3+\omega)/2}) \le O(n^{2.70})$ time. 

The total running time for the \AEExactTriCount{} instance is thus $\OO(n^{3-\eps''} + n^{2.99}+n^{2.70})$, which is truly subcubic. 
\end{proof}

The reduction from \AEExactTri{} to \AEExactTriCount{} is trivial. 

\begin{remark}
\label{rem:exact-tri-count} \rm
Given Theorem~\ref{thm:exact-tri-count}, it is simple to derive a subcubic equivalence between \ExactTri{} and \ExactTriCount{}. First, the reduction from \ExactTri{} to \ExactTriCount{} is trivial. To reduce \ExactTriCount{} to \ExactTri{}, we first reduce \ExactTriCount{} to \AEExactTriCount{} in the trivial way, then use Theorem~\ref{thm:exact-tri-count} to further reduce it to \AEExactTri{}, and finally reduce it to \ExactTri{} by known reductions~\cite{focsyj}. 
\end{remark}

In Section~\ref{sec:counting}, we will use a similar approach to obtain other equivalence
results between counting and detection problems.  In particular, the proof of subquadratic equivalence
between \ThreeSUMCount{} and \ThreeSUM{} will require further technical ideas: we will need to 
exploit or modify known reductions from \AllThreeSUMConv{} to \AEExactTri{},
and \ThreeSUM{} to \ThreeSUMConv{}.

%% file: tri_decompose.tex
\newcommand{\Triangles}{\mbox{\rm Triangles}}
\newcommand{\ZeroTriangles}{\mbox{\rm Zero-Triangles}}

In this section, we introduce a decomposition theorem
for zero-weight triangles, which encapsulates some of the key ideas
we have used, and which may be viewed as an alternative to the BSG Theorem.
We will describe applications of this decomposition theorem to some algorithmic problems that
were previously solved via the BSG Theorem by Chan and Lewenstein~\cite{ChanLewenstein},
as well as a new application to the \#APSP problem for arbitrary weighted graphs.

In a weighted tripartite graph $G$ with node sets $U$, $X$, and $V$, let $\Triangles(G)$ denote
the set of all triangles in $U\times X\times V$ in $G$, and 
let $\ZeroTriangles(G)$ denote
the set of all zero-weight triangles in $U\times X\times V$ in~$G$.

\begin{theorem}\label{thm:tri:decompose}
{\bf (Triangle Decomposition)}
Given a real-weighted tripartite graph $G$ with $n_1$, $n_2$, and $n_3$ nodes in its three parts $U$, $X$, and $V$, and given a parameter $s$,
there exist a collection of $\ell=\OO(s^3)$ subgraphs $G^{(1)},\ldots,G^{(\ell)}$ of $G$, and a set $R$ of $\OO(n_1n_2n_3/s)$ triangles,
such that 
\[ \ZeroTriangles(G)\: =\: R \,\cup\, \bigcup_{\lam=1}^\ell \Triangles(G^{(\lam)}).
\]
The subsets in the union above are disjoint.
The $G^{(\lam)}$'s and $R$ can be constructed in 
$\OO(n_1n_2n_3+sn_1n_2+sn_2n_3+s^2n_1n_3)$ deterministic time.
And if the edge weights between $U$ and $V$  later change, then the $G^{(\lam)}$'s and $R$ can be updated in $\OO(n_1n_2n_3/s + s^2n_1n_3)$ time.

Furthermore, the subgraphs can be grouped into $\OO(s)$ categories of
$\OO(s^2)$ subgraphs each, such that if an edge $uv\in U\times V$ is
present in one subgraph, it is present in all subgraphs of the same category.
\end{theorem}
\begin{proof}
Let $\{u_i:i\in [n_1]\}$, $\{x_k:k\in [n_2]\}$, and $\{v_j:j\in[n_3]\}$ be the nodes of the three parts.
Let the weight of $u_ix_k$ be $a_{ik}$, the weight of $x_kv_j$ be $b_{kj}$, and the weight of $u_iv_j$ be $-c_{ij}$.

As a preprocessing step, we sort the multiset
$\{a_{ik}+b_{kj}: k\in [n_2]\}$
for each $i\in [n_1]$ and $j\in [n_3]$, in $\OO(n_1n_2n_3)$ total time.
Let
$W_{ij}=\{k\in [n_3]: a_{ik}+b_{kj}=c_{ij}\}$.  Note that we can generate
the elements in $W_{ij}$ by searching in the sorted lists in $\OO(1)$ time per element.

\begin{itemize}
\item {\bf Few-witnesses case.}
For each $i,j$ with $|W_{ij}|\le n_2/s$,
add the triangle $u_ix_kv_j$ to $R$ for every $k\in W_{ij}$.
The number of triangles added to $R$ is $O(n_1n_3\cdot n_2/s)$. The
running time of this step is also $\OO(n_1n_3 \cdot n_2/s)$.

\item {\bf Many-witnesses case.}
Find a set $H\subseteq [n_2]$ of $O(s\log(n_1n_2n_3))$ nodes that
hit all $W_{ij}$ with $|W_{ij}| > n_2/s$.
Here, we can use the standard greedy algorithm for hitting sets, which is deterministic and 
takes time linear in the total size of the sets $W_{ij}$; as we can
reduce each set's size to $n_2/s$ before running the hitting set algorithm, 
this takes $\OO(n_1n_2n_3/s)$ time.
For each $i,j$ with $|W_{ij}|>n_2/s$, let $k_0[i,j]$ be some $k_0\in W_{ij}\cap H$.

For each $k_0\in H$ and $k\in[n_2]$,
let $L_{k_0k}$ be the multiset $\{b_{k_0j}-b_{kj}: j\in[n_3]\}$,
and let $F_{k_0k}$ be the elements that have frequency
more than $n_3/r$ in $L_{k_0k}$.
Note that $|F_{k_0k}|\le r$.

\begin{itemize}
\item {\bf Low-frequency case.}
For each $k_0\in H$ and $i\in[n_1]$ and $k\in[n_2]$,
if $a_{ik}-a_{ik_0}\not\in F_{k_0k}$,
we examine each of the at most $n_3/r$ indices $j$
with $a_{ik}-a_{ik_0}=b_{k_0j}-b_{kj}$,
and add $u_ix_kv_j$ to $R$ if it is a zero-weight triangle and $|W_{ij}|>n_2/s$.
The number of triangles added to $R$ is
$\OO(sn_1n_2\cdot n_3/r)=\OO(n_1n_2n_3/s)$, by choosing $r:=s^2$.
The running time of this step is also bounded by $\OO(n_1n_2n_3/s)$.

\item {\bf High-frequency case.}
For each $k_0\in H$ and $p\in [r]$, create a subgraph $G^{(k_0,p)}$
of $G$:
\begin{itemize}
\item For each $i\in[n_1]$ and $j\in[n_3]$, keep edge $u_iv_j$ iff
$k_0[i,j]=k_0$ (in particular, $a_{ik_0}+b_{k_0j}=c_{ij}$).
\item For each $i\in[n_1]$ and $k\in[n_2]$, keep edge $u_ix_k$ iff $a_{ik}-a_{ik_0}$ is the $p$-th element of
$F_{k_0k}$.
\item For each $j\in[n_3]$ and $k\in[n_2]$, keep edge $x_kv_j$ iff $b_{k_0j}-b_{kj}$ is the $p$-th element of $F_{k_0k}$.
\end{itemize}
Note that if $u_ix_kv_j$ is a triangle in $G^{(k_0,p)}$, then
$a_{ik}-a_{ik_0}=b_{k_0j}-b_{kj}$ and $a_{ik_0}+b_{k_0j}=c_{ij}$,
and by Fredman's trick, $a_{ik}+b_{kj}=a_{ik_0}+b_{k_0j}=c_{ij}$,
implying that $u_ix_kv_j$ is a zero-weight triangle.
The running time of this step is  bounded by $\OO(sn_1n_2+sn_2n_3+rn_1n_3)$.
\end{itemize}
\end{itemize}
The number of subgraphs created is $\OO(sr)=\OO(s^3)$.

\smallskip
\emph{Correctness.} Consider a zero-weight triangle $u_ix_kv_j$ in $G$.
If $|W_{ij}|\le n_2/s$, the triangle is in $R$ due to the ``few-witnesses'' case.
So assume $|W_{ij}|>n_2/s$.  Let $k_0=k_0[i,j]$.
We know that $a_{ik}+b_{kj}=c_{ij}=a_{ik_0}+b_{k_0j}$ and by Fredman's trick,
$a_{ik}-a_{ik_0}=b_{k_0j}-b_{kj}$.
If $a_{ik}-a_{ik_0}\not\in F_{k_0k}$, then
the triangle is in $R$ due to the ``low-frequency'' case.
Otherwise, it is a triangle in $G^{(k_0,p)}$ for some $p\in [r]$
due to the ``high-frequency'' case.

\emph{Update edge weights.} If the edge weights between $U$ and $V$ change, we only need to rerun the parts of our algorithms that use the weights $c$. In particular, we need to rerun the low-frequency case, which has running time $\OO(n_1 n_2 n_3 / s)$, and the first subcase of the high-frequency case, which has running time $\OO(r n_1 n_3) = \OO(s^2 n_1 n_3)$. Overall, the running time for updating the edge weights between $U$ and $V$ is $\OO(n_1n_2n_3/s + s^2n_1n_3)$.
\end{proof}

\subsection{Application 1: Exact Triangle in Preprocessed Universes}

As one simple application of the decomposition theorem, we can solve \AEExactTri{}
in truly subcubic time in  a ``preprocessed universe'' setting, where
the input is a subgraph of a preprocessed graph. 
It is convenient to define a variant of the problem,
\AEExactTri{}', which is \AEExactTri{} where the input graph is tripartite with three parts $U$, $X$, and $V$, but we only need to report if each edge between $U$ and $V$ is in an exact triangle. If $|U|=|V|=|X|$, \AEExactTri{}' is equivalent to \AEExactTri{}.

\begin{corollary}\label{cor:zerotri:prep}
Given a real-weighted tripartite graph $G$ with $n_1$, $n_2$, and $n_3$ nodes in its three parts $U$, $X$, and $V$, and given $s$, we can preprocess $G$
in $\OO(n_1n_2n_3+sn_1n_2+sn_2n_3+s^2n_1n_3)$ time, so that for any given subgraph $G'$ of $G$, 
we can 
solve \AEExactTri{}' 
on $G'$ in
$\OO(n_1n_2n_3/s + s M(n_1,s^2n_2,n_3))$ time.

For example, for $n_1=n_2=n_3=n$, after preprocessing in $\OO(n^3)$ time,
we can solve \AEExactTri{} on $G'$ in time $O(n^{2.83})$, or if $\omega=2$, $\OO(n^{11/4})$.
\end{corollary}
\begin{proof}
During preprocessing, we apply Theorem~\ref{thm:tri:decompose} 
 to compute the subgraphs $G^{(\lam)}$ and $R$ in $\OO(n_1n_2n_3/s + sn_1n_2+sn_2n_3+s^2n_1n_3)$ time.

During a query for a given subgraph $G'$ and a target value $t$, if $t$ has changed, we first subtract $t$ from all the edge weights between $U$ and $V$ %
to transform the problem to detecting zero-weight triangles.
We can update the $G^{(\lam)}$'s and $R$ in $\OO(n_1n_2n_3/s + s^2 n_1n_3)$ time.

Next, for each $\lam$, we check whether after removing edges not present in $G'$, the subgraph $G^{(\lam)}$
has a triangle (which would automatically be a zero-weight triangle)
through each edge in $U\times V$.
Since triangle finding (without weights) reduces to matrix multiplication,
the running time is $\OO(s^3 M(n_1,n_2,n_3))$.

We can do slightly better using the grouping of the subgraphs: for each category $\Lambda$, define
a tripartite graph $G^{(\Lambda)}$ with parts $U$, $X\times \Lambda$, and $V$,
and for each $\lam\in\Lambda$, include an edge between $u$ and $(x,\lam)$ if $ux\in G^{(\lam)}\cap G'$,
and between $(x,\lam)$ and $v$ if $xv\in G^{(\lam)}\cap G'$.  For each category $\Lambda$,
we check whether the subgraph $G^{(\Lambda)}$ has a triangle through
each edge.  The running time becomes $\OO(s M(n_1,s^2n_2,n_3))$.
Also, the term $O(s^2n_1n_3)$ for
the update cost can be lowered to $O(n_1n_3)$ now when working with the $G^{(\Lambda)}$'s instead of the $G^{(\lambda)}$'s, as can be checked from our construction (since each edge $u_iv_j$ occurs in one category).

For $n_1=n_2=n_3=n$, we choose $s=n^{0.17+\eps}$ (using the fact that
$\omega(1,1.34,1)<2.657$~\cite{legallurr}), or if $\omega=2$, $s=n^{1/4}$.
\end{proof}

\subsection{Application 2: 3SUM in Preprocessed Universes}

The above implies also a new algorithm for 
\ThreeSUM{} with a preprocessed universe, previously studied
by Chan and Lewenstein~\cite{ChanLewenstein}, who obtained
$\OO(n^{13/7})$ query time, after preprocessing in $\OO(n^2)$
randomized time or $\OO(n^\omega)$ deterministic time.  Our query time is strictly better if $\omega$ is 2 (or is sufficiently close to 2), and we also obtain \emph{deterministic} $\OO(n^2)$ preprocessing time regardless of
$\omega$.

\begin{corollary}\label{cor:3sum:prep0}
We can preprocess sets $A,B,C$ of $n$ integers in $[n^{O(1)}]$
in $\OO(n^2)$ deterministic time, so that given any subsets $A'\subseteq A$,
$B'\subseteq B$, and $C'\subseteq C$, we can solve
\AllThreeSUM{} on $(A',B',C')$ in time $O(n^{1.891})$, or if $\omega=2$, $\OO(n^{11/6})$.
\end{corollary}
\begin{proof}
We use a known reduction from (All-Nums-) \ThreeSUM{} to (All-Nums-) \ThreeSUMConv{} (one of the reductions by Chan and He~\cite{ChanHe}
is deterministic and increases running time only by polylogarithmic
factors when the input numbers are polynomially bounded),
in combination with a known reduction 
from \AllThreeSUMConv{} to  \AEExactTri{}~\cite{VassilevskaW09}.
The problem is reduced to $\OO(1)$ instances of 
the problem in Corollary~\ref{cor:zerotri:prep}
with $n_1=n/q$, $n_2=q$ and $n_3=n$ for a parameter $q$
(the original reduction~\cite{VassilevskaW09} has $q=\sqrt{n}$, but we will
do better with a different choice of $q$).
It is straightforward to check that these
reductions carry over to the preprocessed universe setting. 
We then obtain preprocessing time $\OO(n^2+snq+s^2n^{2}/q)=\OO(n^2+s^3n^{1.5})$ and
query time $\OO(n^2/s + s M(n/q,s^2q,n))
=\OO(n^2/s + s M(s\sqrt{n},s\sqrt{n},n))$ by setting $q=\sqrt{n}/s$.
We choose $s=n^{0.109+\eps}$  (using the fact that
$\omega(0.609,0.609,1)<1.781$~\cite{legallurr}), or if $\omega=2$, $s=n^{1/6}$.
\end{proof}

\begin{remark}\rm
Corollary~\ref{cor:zerotri:prep} and Corollary~\ref{cor:3sum:prep0} can also be used to solve \AEExactTriCount{} and \AllThreeSUMCount{} in the preprocessed universe. This is because Theorem~\ref{thm:tri:decompose} provides a decomposition, so we can sum up the counts in all the cases (In Corollary~\ref{cor:3sum:prep0}, we also have to be careful when applying the reduction in \cite{ChanHe}, to make sure we do not over count, by using inclusion-exclusion). Prior method for \ThreeSUM{} in the preprocessed universe \cite{ChanLewenstein} cannot compute counts, because it relies on the BSG theorem, which only provides a covering. 
\end{remark}

In Section~\ref{sec:3sum:prep:rand}, we will describe a still better solution to 3SUM in preprocessed universes, with randomization, using FFT instead of fast matrix multiplication.

\subsection{Application 3: A Deterministic 3SUM Algorithm for Bounded Monotone Sets}

Another interesting application is the following:

\begin{corollary}
Given monotone sets $A,B,C\subseteq [n]^d$ for any constant $d\ge 2$,
we can solve \AllThreeSUM{} on $A, B, C$ in
$O(n^{2-1/O(d)})$ deterministic time. 
\end{corollary}
\begin{proof}
Chan and Lewenstein~\cite{ChanLewenstein} observed that
\ThreeSUM{} for bounded monotone sets in any constant dimension $d$
reduces to \ThreeSUM{} for ``clustered'' sets, which in turn
reduces to \ThreeSUM{} with preprocessed universes on $O(n/\ell)$ elements
and $O(\ell^{3d})$ queries.
Using Corollary~\ref{cor:3sum:prep0} gives
total deterministic time $\OO((n/\ell)^2 + \ell^{3d} (n/\ell)^{1.891})$,
which is $\OO(n^{2-0.218/(3d+0.109)})$ by setting $\ell=n^{0.109/(3d+0.109)}$.
(The exponent here is certainly improvable, by solving the problem
using our techniques more directly, instead of applying a black-box reduction to \ThreeSUM{} with preprocessed universes.)
\end{proof}

The above result provides the first truly subquadratic
\emph{deterministic} algorithm for bounded monotone 3SUM in arbitrary constant dimensions---Chan and Lewenstein~\cite{ChanLewenstein} gave subquadratic randomized algorithms 
with $O(n^{2-1/(d+O(1))})$ running time, but they had
nontrivial deterministic algorithms only for $d\le 7$ under
the current matrix multiplication bounds.

We can also apply the triangle decomposition theorem to obtain 
subquadratic algorithms for monotone or bounded-difference Min-Plus convolution
(which were first obtained by Chan and Lewenstein~\cite{ChanLewenstein}, and followed by \cite{ChiDXZstoc22}),
and subcubic algorithms for monotone or bounded-difference  Min-Plus products
(which were first obtained by Bringmann et al.~\cite{BringmannGSW16}, and followed by \cite{WilliamsX20,GuPWX21, ChiDXsoda22,ChiDXZstoc22}).
Since previous algorithms have been found for these problems, we will omit the details here and refer to Appendix~\ref{app:bd:diff}.  

The main message is that many of the results in Chan and Lewenstein's
paper can be obtained alternatively using our decomposition theorem, which is simpler
and more elementary than the BSG Theorem, if we are interested in
subquadratic algorithms but don't care about the
precise values in the exponents.  The advantage is simplicity---additive combinatorics is not needed after all!  (However, the BSG Theorem is still potentially useful in optimizing those exponents.)

\subsection{Application 4: A Truly Subquartic \#APSP Algorithm}

As another simple, interesting application of the triangle decomposition theorem,
we obtain the first truly subquartic algorithm for \APSPCount{}
for arbitrary weighted graphs:

\begin{theorem}
\label{thm:apsp-count}
\APSPCount{} for $n$-node graphs with positive edge weights has an algorithm running in
$O(n^{3.83})$ time, or if $\omega=2$, $\OO(n^{15/4})$ time.
\end{theorem}
\begin{proof}
We define the following ``funny'' matrix product $\otimes$: if $(C, C') = (A, A') \otimes (B, B')$, then $C = A \star B$ and $C'_{ij} = \sum_{k \in [n]: C_{ij} = A_{ik} + B_{kj}} A'_{ik}B'_{kj}$. 

\begin{claim}
\label{cl:apsp-count}

Let $(A, A')$ and $(B, B')$ be two pairs of $n \times n$ matrices where the entries of $A'$ and $B'$ are (large) $\ell$-bit integers. Then we can compute $(A, A') \otimes (B, B')$ in time $\OO(n^3 + \ell n^{2.83})$, or if $\omega=2$, $\OO(n^3 + \ell n^{11/4})$.
\end{claim}
\begin{proof}
First compute $C=A\star B$ naively in $O(n^3)$ time.  Initialize the entries of $C'$ to 0.
Consider the tripartite graph
with nodes $\{u_i:i\in[n]\}$, $\{x_k:k\in[n]\}$, and $\{v_j:j\in[n]\}$, where $u_ix_k$ has weight $A_{ik}$, and $x_kv_j$ has weight $B_{kj}$, and $u_iv_j$ has weight $-C_{ij}$.
Apply Theorem~\ref{thm:tri:decompose} to obtain subgraphs $G^{(\lam)}$ and a set $R$ in $\OO(n^3 + s^2n^2)$ time.

We first examine each triangle $u_ix_kv_j\in R$ and add
$A'_{ik}B'_{kj}$ to $C'_{ij}$.  This takes $\OO(\ell n^3/s)$ time.

Next, for each $\lam$, we compute $\sum_{k}  [u_ix_k\in G^{(\lam)}] A'_{ik} \cdot  [x_kv_j\in G^{(\lam)}] B'_{kj}$ and add it to $C'_{ij}$ for every $u_iv_j\in G^{(\lam)}$.  
This reduces to a standard matrix product on $\ell$-bit integers
and takes $\OO(\ell n^\omega)$ time for each $\lam$.  The total time
is $\OO(\ell\cdot s^3 n^\omega)$.
As before, we can do slightly better using the grouping of the subgraphs, which improve the running time to
$\OO(\ell\cdot s M(n,s^2n,n))$. 

As in Corollary~\ref{cor:zerotri:prep}, we choose $s=n^{0.17+\eps}$, or if $\omega=2$, $s=n^{1/4}$.
\end{proof}

Given an input graph $G$ with positive edge weights, let $D^{(=2^i)}$ be the distance matrix for paths of (unweighted) length exactly $2^i$, and $D'^{(=2^i)}$ be the number of paths of (unweighted) length exactly $2^i$ that match the distance in $D^{(=2^i)}$. Similarly, we define $D^{(<2^i)}$ as the distance matrix for paths of (unweighted) length less $2^i$, and define $D'^{(<2^i)}$ similarly. 

For $i=0$, it is easy to see that $D^{(=1)}$ is exactly the weight matrix of $G$, and $D'^{(=1)}$ is the adjacency matrix of $G$. Also, $D^{(<1)}$ is the matrix whose diagonal entries are all $0$, and other entries are all $\infty$. Finally $D'^{(<1)}$ equals the $n \times n$ identity matrix. 

For $i>0$, we can use the following recurrences:
\begin{equation*}
    \begin{split}
        (D^{(=2^i)}, D'^{(=2^i)}) &= (D^{(=2^{i-1})}, D'^{(=2^{i-1})}) \otimes  (D^{(=2^{i-1})}, D'^{(=2^{i-1})}),\\
        (D^{(<2^i)}, D'^{(<2^i)}) &= (D^{(<2^{i-1})}, D'^{(<2^{i-1})}) \otimes  (D^{(=2^{i-1})}, D'^{(=2^{i-1})}).
    \end{split}
\end{equation*}
It is not difficult, though a bit tedious, to verify the correctness of these recurrences. 

The matrix $D'^{(<2^i)}$ gives the result for \APSPCount{} when $2^i > n$. 
Therefore, 
\APSPCount{}
reduces to $O(\log n)$ instances of the funny product $\otimes$, when the matrices $A'$ and $B'$ are  $\OO(n)$-bit numbers. Then applying Claim~\ref{cl:apsp-count} with $\ell=\OO(n)$ yields the theorem.
\end{proof}

In Section~\ref{sec:bsg}, we will return to the BSG Theorem and describe
more variants and applications.

%% file: int_apsp_reduce_more.tex
Continuing the approach in Section~\ref{sec:intapsp-lower-bound}, we now derive
conditional lower bounds for more problems with intermediate complexity from the \StrongAPSP{} (which as noted before 
is equivalent to the hypothesis
that $M^*(n,n,n\mid n^{3-\omega})$ is not truly subcubic).
We begin with a useful lemma:

\begin{lemma}\label{lem:down}
For any positive constants $\beta,\gamma,c$ with $0<\gamma<\beta$,
if $M^*(n,n^{\beta},n\mid n^{c\beta})=O(n^{2+\beta-\eps})$,
then $M^*(n,n^{\gamma},n\mid n^{c\gamma})=O(n^{2+\gamma-\Omega(\eps)})$.
\end{lemma}
\begin{proof}\

\negbigskip\negbigskip
\begin{eqnarray*}
M^*(n,n^\gamma,n\mid n^{c\gamma}) &\le & 
O(n^{2(1-\gamma/\beta)} M^*(n^{\gamma/\beta},n^\gamma,n^{\gamma/\beta}\mid n^{c\gamma}))\\
&\le &  \OO(n^{2(1-\gamma/\beta)}\cdot (n^{\gamma/\beta})^{2+\beta-\eps})
\ =\ O(n^{2+\gamma-\Omega(\eps)}).
\end{eqnarray*}

\negbigskip
\end{proof}

Lemma~\ref{lem:down} allows us to remove the assumption $\beta\le (3-\omega)/2$ in Corollary~\ref{cor:intapsp0} (since we can replace $\beta$ with any sufficiently small positive constant $\gamma$).  Consequently, Corollary~\ref{cor:strong-intapsp-imply}
holds for all $\beta\le 1/3$, regardless of the value of $\omega$.  For convenience, we repeat the statement of Corollary~\ref{cor:intapsp0} below, with the assumption removed:

\begin{corollary}\label{cor:intapsp}
For any constant $\beta>0$,
if $M^*(n,n^\beta,n\mid n^{2\beta})=O(n^{2+\beta-\eps})$,
then $M^*(n,n,n\mid n^{3-\omega})=O(n^{3-\Omega(\eps)})$.
\end{corollary}

\begin{remark}\label{rmk:uapsp}\rm
By applying Lemma~\ref{lem:down} with $\beta=1$, $\gamma=1/2$, and $c=1$, we see that
$M^*(n,n,n\mid n)=O(n^{3-\eps})$ implies
$M^*(n,\sqrt{n},n\mid\sqrt{n})=O(n^{5/2-\Omega(\eps)})$.  If $\omega=2$, Chan, Vassilevska W. and Xu~\cite{CVXicalp21} showed that the \uAPSPH{} is equivalent to the claim that $M^*(n,\sqrt{n},n\mid\sqrt{n})$ is not in $O(n^{5/2-\eps})$ for any $\eps>0$.  Consequently, the \uAPSPH{} implies the \StrongAPSP{} when $\omega=2$.
\end{remark}

\begin{remark}\label{rmk:strongapsp:dir}\rm
In the definition of the \StrongAPSP{}, it does not matter whether the input graph is undirected or directed---the directed version is also equivalent to the statement that
$M^*(n,n,n\mid n^{3-\omega})$ is not truly subcubic: 
Directed APSP for integer weights in $[n^{3-\omega}]$ can be 
solved by Zwick's algorithm~\cite{zwickbridge,CVXicalp21} in
time $\OO(\max_\ell M^*(n,n/\ell,n\mid \ell n^{3-\omega}))$.
If $M^*(n,n,n\mid n^{3-\omega})=O(n^{3-\eps})$, then for $\ell\le n^\delta$,
we have $M^*(n,n/\ell,n\mid \ell n^{3-\omega})\le \OO(M^*(n,n,n\mid n^{3-\omega+\delta}) \le O(n^{(3-\eps)(3-\omega+\delta)/(3-\omega)})=O(n^{3-\Omega(\eps)})$ for a sufficiently small $\delta$.  On the other hand,
for $\ell>n^\delta$, we trivially have  $M^*(n,n/\ell,n\mid \ell n^{3-\omega})\le O(n^3/\ell) = O(n^{3-\delta})$.  
\end{remark}

\subsection{Min-Witness Product}

Let $\Tminwit(n)$ be the time complexity of computing min-witness product
for two  $n\times n$ Boolean matrices.
More generally,
let $\Tminwit(n_1,n_2,n_3)$  be the time complexity of \MinWitness{}
for an $n_1\times n_2$ Boolean matrix and an $n_2\times n_3$ Boolean matrix.  We observe that 
Min-Plus product for rectangular matrices of sufficiently small inner dimensions and
integer ranges can be reduced to \MinWitness{}:

\begin{lemma}\label{lem:reduce:from:minplus:minwit}
For any $x,y$, we have
$M^*(n,x,n\mid y)= O(\Tminwit(n,xy^2,n))$.
\end{lemma}
\begin{proof}
Suppose that we are given an $n\times x$ matrix $A$
and an $x\times n$ matrix $B$ where all matrix entries are in $[y]$.
For each $i\in[n]$, $k\in [x]$, and $u,v\in [y]$, let 
$A'_{i,(k,u,v)}=1$ if $A_{ik}=u$, and $A'_{i,(k,u,v)}=0$ otherwise.
For each $j\in[n]$, $k\in [x]$, $u,v\in [y]$, let 
$B'_{(k,u,v),j}=1$ if $B_{kj}=v$, and $B'_{(k,u,v),j}=0$ otherwise.
The number of triples $\tau=(k,u,v)$ is $xy^2\le n$.  
By sorting the triples in increasing order of $u+v$,
 the Min-Plus product of $A$
and $B$ can be computed from the Min-Witness product of $(A'_{i,\tau})_{i\in[n],\tau\in [x]\times[y]^2}$ and $(B'_{\tau, j})_{\tau\in [x]\times[y]^2, j\in[n]}$. 
\hspace*{\fill}
\end{proof}

Combining Corollary~\ref{cor:intapsp} and Lemma~\ref{lem:reduce:from:minplus:minwit} immediately gives the following:

\begin{corollary}\label{cor:strong-intapsp-imply:minwit}
If $\Tminwit(n,n^{5\beta},n)=O(n^{2+\beta-\eps})$, then
$M^*(n,n,n\mid n^{3-\omega})= O(n^{3-\Omega(\eps)})$.
\end{corollary}

By setting $\beta=1/5$, we have thus proved that 
\MinWitness{} of two $n \times n$ Boolean matrices
cannot be computed in $O(n^{11/5-\eps})$ time under the \StrongAPSP{}.
In particular, if $\omega=2$,
this implies that \MinWitness{} is strictly harder than  Boolean matrix multiplication.

Furthermore, by setting $\beta=\gamma/5$, 
\MinWitness{} of an $n\times n^{\gamma}$ and an $n^{\gamma}\times n$ Boolean matrix cannot be computed in $O(n^{2+\gamma/5-\eps})$ time for any $\gamma$ under the same hypothesis.  This implies that for any $\gamma\le 0.3138$, \MinWitness{} is strictly harder than  Boolean matrix multiplication, as $\omega(1,0.3138,1)=2$~\cite{legallurr}.  This result interestingly rules out
the possibility that the polynomial method~\cite{Williams18,AbboudWY15,ChanW21} could be used to transform
the Min-Witness product of an $n\times d$ and a $d\times n$ matrix
into a standard product of an $n\times d^{O(1)}$ and a $d^{O(1)}\times n$ matrix.  It also contrasts \MinWitness{} with, for example, \Dominance{} or \Equality{},
which does have near-quadratic time complexity when the inner dimension $d$ is smaller than $n^{0.1569}$
(since as mentioned in Section~\ref{sec:prelim:eq:prod}, Matou\v sek's technique~\cite{MatIPL,YusterDom} yields a time bound 
of $\OO(\min_r (dn^2/r + M(n,dr,n))\le \OO(n^2 + M(n,d^2,n))$).

\subsection{All-Pairs Shortest Lightest Paths}

Let $\TundirAPSLP(n,m)$ be the time complexity of \APSLP{} on graphs with $n$ nodes and $m$ edges.  We observe the following:

\begin{lemma}\label{lem:reduce:from:minplus:apslp}
For any $x,y$, we have
$M^*(n,x,n\mid y)=O(\TundirAPSLP(n,nx))$ if $xy^2\le n$.
\end{lemma}
\begin{proof}
Suppose that we are given an $n\times x$ matrix $A$
and an $x\times n$ matrix $B$ where all matrix entries are in $[y]$.
We construct an undirected graph as follows:
\begin{itemize}
\item For each $i\in [n]$, create a node $s[i]$.
\item For each $k\in [x]$ and $u\in [y]$, create a node $w_1[k,u]$.
\item For each $k\in [x]$, create a node $w_2[k]$.
\item For each $k\in [x]$ and $v\in [y]$, create a node $w_3[k,v]$.
\item For each $j\in [n]$, create a node $t[j]$.
\item For each $i\in [n]$ and $k\in [x]$, create an edge $s[i]w_1[k,u]$ of weight~1 where $u=a_{ik}$.
\item For each $k\in [x]$ and $u\in [y]$, create a
 path between $w_1[k,u]$ and $w_2[k]$
that has $u$ edges of weight 2 and $y-u$ edges of weight 1 (so that the path has $y$ edges and weight $y+u$).
\item For each $k\in [x]$ and $v\in [y]$, create a
 path between $w_2[k]$ and $w_3[k,v]$ 
that has $v$ edges of weight 2 and $y-v$ edges of weight 1 (so that the path has $y$ edges and weight $y+v$).
\item For each $j\in [n]$ and $k\in [x]$, create an edge $w_3[k,v]t[j]$ of weight~1 where $v=b_{kj}$.
\end{itemize}
The number of nodes in the graph is $O(n+xy^2)=O(n)$,
and the number of edges is $O(nx+xy^2)=O(nx)$.
For each $i,j\in [n]$, all lightest paths from $s[i]$ to $t[j]$
have $2y+2$ edges, and among them, the shortest has weight $2y+2+\min_k (a_{ik}+b_{kj})$.
\end{proof}

Combining Corollary~\ref{cor:intapsp} and Lemma~\ref{lem:reduce:from:minplus:apslp}
immediately gives the following:

\begin{corollary}\label{cor:strong-intapsp-imply:apslp}
For any constant $\beta\le 1/5$,
if $\TundirAPSLP(n,n^{1+\beta})=O(n^{2+\beta-\eps})$, then
$M^*(n,n,n\mid n^{3-\omega})= O(n^{3-\Omega(\eps)})$.
\end{corollary}

By setting $\beta=1/5$, we have thus proved that \APSLP{} cannot be solved
in $O(n^{11/5-\eps})$ time under the \StrongAPSP{}.  If $\omega=2$, this implies that
\APSLP{} is strictly harder than undirected APSP for such weighted graphs.  (The current best algorithm for \APSLP{} has running time $\OO(n^{2.58})$, or if $\omega=2$, $\OO(n^{5/2})$~\cite{CVXicalp21}.)  The same result holds for \APLSP{}.

Furthermore, \APSLP{} 
with $n^{1+\beta}$ edges cannot be solved in $O(n^{2+\beta-\eps})$ time for any $\beta\le 1/5$ under the same hypothesis.  Thus, the naive algorithm is essentially \emph{optimal} for sufficiently sparse graphs.

The same results hold for the similar problem of \APLSP{} (in the proof of Lemma~\ref{lem:reduce:from:minplus:apslp}, we just modify the path from $w_1[k,u]$ and $w_2[k]$
to use $y-u$ edges of weight 2 and $2u$ edges of weight 1, so that the
path has $y+u$ edges and weight $2y$).

\subsection{Batched Range Mode}

For a different application about data structures, 
we next consider  the range mode problem, which has received considerable attention and has been extensively studied in the literature \cite{KrizancMS05,CDLMW,WilliamsX20,GaoHe22,JinX22}.
Let $\Trangemode(N,Q\mid \sigma)$ be the total time complexity of \BatchMode{} for $Q$ range mode queries in an array of $N$ elements in $[\sigma]$.

\begin{lemma}\label{lem:reduce:from:minplus:mode}
For any $x,y$, we have
$M^*(n,x,n\mid y)= O(\Trangemode(nxy,n^2\mid x))$.
\end{lemma}
\begin{proof}
\newcommand{\ssigma}{\sigma'}
\newcommand{\ttau}{\tau'}
Suppose that we are given an $n\times x$ matrix $A$
and an $x\times n$ matrix $B$ where all matrix entries are in $[y]$.
We create an array holding a string $S$ over the alphabet $[x]$, defined as follows: 
\begin{itemize}
\item Let $\sigma_i = 1^{A_{i1}}2^{A_{i2}}\cdots x^{A_{ix}}$
and $\ssigma_i = 1^{y-A_{i1}}2^{y-A_{i2}}\cdots x^{y-A_{ix}}$,
which have length $O(xy)$.
\item Let $\tau_j = 1^{B_{1j}}2^{B_{2j}}\cdots x^{B_{xj}}$
and $\ttau_i = 1^{y-B_{1j}}2^{y-B_{2j}}\cdots x^{y-B_{xj}}$,
which have length $O(xy)$.
\item Let $S=\sigma_n\ssigma_n\cdots \sigma_2\ssigma_2\sigma_1\ssigma_1
\ttau_1\tau_1\ttau_2\tau_2\cdots\ttau_n\tau_n,$
which has length $O(nxy)$.
\end{itemize}

For each $i,j\in [n]$, 
consider the substring $S_{ij}=\ssigma_i\sigma_{i-1}\ssigma_{i-1}\cdots
\sigma_1\ssigma_1\ttau_1\tau_1\cdots\ttau_{j-1}\tau_{j-1}\ttau_j$.
For each $k\in[x]$, the frequency of $k$ in $S_{ij}$
is precisely $iy+jy - A_{ik}-B_{kj}$.
Thus, the mode of $S_{ij}$ is an index $k$ minimizing $A_{ik}+B_{kj}$.
So, the Min-Plus product can be computed by answering $O(n^2)$ range mode queries on $S$.  
\end{proof}

Combining Corollary~\ref{cor:intapsp} and Lemma~\ref{lem:reduce:from:minplus:mode}
immediately gives the following:

\begin{corollary}\label{cor:strong-intapsp-imply:mode}
For any constant $\beta$, 
if $\Trangemode(n^{1+3\beta},n^2\mid n^\beta)=O(n^{2+\beta-\eps})$, then
$M^*(n,n,n\mid n^{3-\omega})= O(n^{3-\Omega(\eps)})$.
\end{corollary}

By setting $\beta=1/3$ and $n=\sqrt{N}$, we have thus proved that 
\BatchMode{} for $N$ queries on $N$ elements cannot be
solved in $O(N^{7/6-\eps})$ time under the \StrongAPSP{}.
Previously, Chan et al.~\cite{CDLMW} gave a reduction
from  Boolean matrix multiplication implying a better, near-$N^{3/2}$ conditional lower bound for 
\emph{combinatorial} algorithms (under the Combinatorial BMM Hypothesis); this matched upper bounds of known
combinatorial algorithms~\cite{KrizancMS05,CDLMW}.
However, for noncombinatorial
algorithms, their lower bound was near $N^{\omega/2}$,
which is trivial if $\omega=2$.
The distinction between combinatorial vs.\ noncombinatorial algorithms is especially important for the range mode problem, as it is actually possible to beat $N^{3/2}$ using fast matrix multiplication, as first shown by Vassilevska W. and Xu~\cite{WilliamsX20}.  The current fastest algorithm by Gao and He~\cite{GaoHe22} runs in  $O(N^{1.4797})$ time.  Our new 
lower bound reveals that there is a limit on how much fast matrix multiplication can help.

(For still more recent work on range mode, see \cite{JinX22} for
a conditional lower bound for the dynamic version of the range mode problem, but again this is only for combinatorial algorithms.)

Furthermore, by  using the fact that $\Trangemode(n^{1+3\beta},n^2\mid n^\beta)\le O(n^{1-3\beta})\cdot $\\ $
\Trangemode(n^{1+3\beta},n^{1+3\beta}\mid n^\beta)$ and setting $N=n^{1+3\beta}$ and $\gamma=\beta/(1+3\beta)$, we see that
\BatchMode{} for $N$ queries on $N$ elements in a universe 
of size $\sigma=N^\gamma$
cannot be answered in $O(N^{1+\gamma-\eps})$ time for
any $\gamma\le 1/6$, under the same hypothesis.  This lower bound is \emph{tight}, since an $O(N\sigma)$
upper bound is known~\cite{CDLMW}.

Furthermore, by setting $n=\sqrt{Q}$ and $n^\beta=(N/\sqrt{Q})^{1/3}$, \BatchMode{} for $Q$ queries on $N$ elements cannot be solved in $O(Q^{5/6}N^{1/3-\eps})$ time for any $Q\le N^2$ under the same hypothesis.
For example, for $Q=N^{1.6}$, the lower bound is near $N^{1.666}$
(in other words, we need at least $N^{0.066}$ time per query).
In contrast, the previous reduction by Chan et al.~\cite{CDLMW} from
Boolean matrix multiplication
gives a lower bound of $M(\sqrt{Q},N/\sqrt{Q},\sqrt{Q})$,
which is only near linear in $Q$ when $Q=N^{1.6}$, as $\omega(0.8,0.2,0.8)=1.6$.
(Known combinatorial algorithms have running time near $O(\sqrt{Q}N)$ as a function of $N$ and $Q$~\cite{KrizancMS05,CDLMW}.)

The same results hold for the similar problem 
of \emph{range minority}~\cite{ChanDSW15} (finding a least frequent element in a range).

\subsection{Dynamic Shortest Paths in Unweighted Planar Graphs}

\newcommand{\Tdynplanar}{\mbox{\sc planar-dyn-sp}}

As another example of an application to data structure problems, we now consider the dynamic shortest path problem for unweighted planar graphs.
Let $\Tdynplanar(N,Q,U)$ be the time complexity of 
performing an offline sequence of $Q$ shortest path distance queries and
$U$ edge updates on an unweighted, undirected planar graph 
with $N$ nodes.  

\begin{lemma}\label{lem:dynplanar}
For any $\alpha,\beta\le 1$, 
\[ M^*(n,n^\beta,n\mid X)=
O(n^{1-\alpha}\cdot \Tdynplanar((n^{2\alpha+\beta} + n^{\alpha+2\beta})X, n^{1+\alpha}, n^{1+\beta})).\]
\end{lemma}
\begin{proof}
Abboud and Dahlgaard~\cite[Proof of Theorem~1]{AbbDah} reduced the
computation of the Min-Plus product of an $n\times n^\beta$ and
an $n^\beta\times n^\alpha$ matrix with entries from $[X]$, to
the problem of performing an offline sequence of $O(n^{1+\alpha})$ 
shortest path queries
and $O(n^{1+\beta})$ edge-weight changes on a \emph{weighted} planar graph with $O(n^{\alpha+\beta})$
nodes.  
In their graph construction, all edges have integer weights bounded by $O(n^\alpha X)$, except
for $O(n^\beta)$ edges having integer weights bounded by $O(n^{\alpha+\beta} X)$.
(The $X$ factor was stated as $X^2$ in their paper, but as they remarked at the end of their Section~2,
$X^2$ can be lowered to $X+1$.)
The weighted graph can be turned into an unweighted graph, simply by
subdividing each edge.  More precisely, we create a path $\pi_e$ of length $\ell$ 
for an edge $e$ with weight upper-bounded by $\ell$.  Whenever the weight 
of $e$ changes,
we can redirect an endpoint of $e$ to an appropriate node in the path $\pi_e$,
using $O(1)$ updates in this unweighted graph.
The resulting unweighted planar graph has  $O((n^{2\alpha+\beta} + n^{\alpha+2\beta})X)$ nodes.  Hence,
Abboud and Dahlgaard's reduction implies that 
$M^*(n,n^\beta,n^\alpha\mid X)=
O(\Tdynplanar((n^{2\alpha+\beta} + n^{\alpha+2\beta})X, n^{1+\alpha}, n^{1+\beta}))$.

The lemma then follows, as $M^*(n,n^\beta,n\mid X) = O(n^{1-\alpha}\cdot 
M^*(n,n^\beta,n^\alpha\mid X))$.
\end{proof}

Combining Corollary~\ref{cor:intapsp} and Lemma~\ref{lem:dynplanar} (with $\alpha=\beta$)
immediately gives the following:

\begin{corollary}\label{cor:strong-intapsp-imply:planar}
For any constant $\beta$,
if $\Tdynplanar(n^{5\beta}, %
n^{1+\beta},n^{1+\beta})=O(n^{1+2\beta-\eps})$, then
$M^*(n,n,n\mid n^{3-\omega})= O(n^{3-\Omega(\eps)})$.
\end{corollary}

By setting $\beta=1/4$ and $N=n^{5/4}$, %
we have thus proved that 
an offline sequence
of $N$ shortest path queries and $N$ updates on an unweighted, undirected $N$-node planar
graph cannot be processed in $O(N^{6/5-\eps})$ time, under the
\StrongAPSP{}.  This rules out the existence of data structures with $N^{o(1)}$ time per operation.
Abboud and Dahlgaard~\cite{AbbDah} proved
a better lower bound near $N^{3/2}$ in the weighted case under the APSP Hypothesis, 
or near $N^{4/3}$ in the unweighted case under the OMv Hypothesis~\cite{HenzingerKNS15}, but the latter
bound under the OMv Hypothesis holds only for \emph{online} queries and updates, when considering general noncombinatorial algorithms.
We have obtained the first conditional lower bounds for the unweighted case that hold in the offline setting.

Gawrychowski and Janczewski~\cite{GawJan} have adapted Abboud and Dahlgaard's technique to prove conditional lower bounds for certain dynamic data structure versions of the \emph{longest increasing subsequence (LIS)} problem.  In the unweighted case, their reduction was again based on the OMv Hypothesis and applicable only for the online setting.  Our approach should similarly yield new conditional lower bounds in the offline setting for their problem.

The preceding applications are not meant to be exhaustive, but
the applications to \BatchMode{} and dynamic planar shortest paths should suffice to
illustrate the potential usefulness of our technique to (unweighted) data structure problems in general.
A common way to obtain conditional lower bounds for such data structure problems is via reduction from Boolean matrix multiplication, which is useful only for combinatorial algorithms, or
via the OMv Hypothesis, which is only for online settings.
Our technique provides a new avenue, allowing us to obtain (weaker, but still nontrivial) lower bounds for general noncombinatorial algorithms in offline or batched settings: namely, it
suffices to reduce from rectangular Min-Plus products when the inner dimension 
and the integer range are both small.

\subsection{Min-Witness Equality Product}

Lastly, we revisit the \MinWitnessEq{} problem.
Let $\Tminwiteq(n)$ be the time complexity of \MinWitnessEq{}
for $n\times n$ matrices.  Chan, Vassilevska W., and Xu~\cite{CVXicalp21} showed
that $\TunwtdirAPSP(n)=\OO(\Tminwiteq(n))$, so we immediately obtain a near-$n^{7/3}$ lower bound
for \MinWitnessEq{} under the \StrongAPSP{}.  However, the following corollary gives an alternative lower bound, which is worse if $\omega=2$, but is better if the current bound of
$\omega$ turns out to be close to tight.  Note that $(2\omega+5)/4$ is strictly larger
than $\omega$ for all $\omega\in [2,2.373)$.

\begin{corollary}
If $\Tminwiteq(n) = O(n^{(2\omega+5)/4-\eps})$,
then $M^*(n,n,n\mid n^{3-\omega}) = O(n^{3-\Omega(\eps)})$.
\end{corollary}
\begin{proof}
Let $\Tgeneq(n_1,n_2,n_3\mid\ell)$ be the time complexity of
the generalized equality product problem in Lemma~\ref{lem:geneqprod}.
Let $\Teq(n_1,n_2,n_3)$ be the time complexity of \ExistEquality{}, which is a variant of \Equality{} where we only need to determine if each of the outputs of the standard \Equality{} is nonzero or not.
It is easy to see that $\Tgeneq(n_1,n_2,n_3\mid\ell) \le O(\ell^2\cdot \Teq(n_1,n_2,n_3))$.

The proof of Theorem~\ref{thm:main} shows that for any $s\le n_2$
and $t\le\ell$,
\[ M^*(n_1,n_2,n_3\mid\ell)\ =\ 
\OO\big( (n_2/s) M^*(n_1,s,n_3\mid t)
\,+\, s\cdot\Tgeneq(n_1,n_2,n_3\mid \ell/t) \big). \]
Consequently, for any $t\le n^{3-\omega}$,
\[ M^*(n,n,n\mid n^{3-\omega})\ =\ 
\OO\big( (n/s) M^*(n,s,n\mid t)
\,+\, s(n^{3-\omega}/t)^2\cdot \Teq(n,n,n) \big). \]
Set $t=n/s$.
Chan, Vassilevska W. and Xu~\cite{CVXicalp21} gave a reduction showing
that $M^*(n,s,n\mid n/s)=O(\Tminwiteq(n))$.
Trivially, $\Teq(n,n,n)=O(\Tminwiteq(n))$.
It follows that $M^*(n,n,n\mid n^{3-\omega})\ =\ O((n/s + s^3/n^{2\omega-4})\cdot \Tminwiteq(n))$.  The result follows by setting $s=n^{(2\omega-3)/4}$.
\end{proof}

\section{Lower Bounds under the \uAPSPH{}}\label{sec:uapsp-lower-bound}

We can similarly apply our key reduction in Theorem~\ref{thm:main} to
obtain (better) lower bounds under the \uAPSPH{},
using the following corollary:

\begin{corollary}\label{cor:unwtdirapsp}
Let $\rho$ be such that $\omega(1,\rho,1)=1+2\rho$.
Fix any constant $\sigma>\rho$, and
let $\kappa=\frac{\omega(1,\sigma,1)-1-2\rho}{\sigma-\rho}$.
For any constant $\beta$,
if $M^*(n,n^\beta,n\mid n^{(1+\kappa)\beta})=O(n^{2+\beta-\eps})$,
then $\TunwtdirAPSP(n)=O(n^{2+\rho-\Omega(\eps)})$.
\end{corollary}
\begin{proof}
Chan, Vassilevska W. and Xu~\cite{CVXicalp21} have shown that $\TunwtdirAPSP(n)=O(n^{2+\rho-\Omega(\eps)})$ is equivalent to
$M^*(n,n^\rho,n\mid n^{1-\rho})=O(n^{2+\rho-\Omega(\eps)})$.

By Theorem~\ref{thm:main},
\[ M^*(n,n^\rho,n\mid n^{1-\rho})\ =\ \OO\left((n^{\rho}/s)M^*(n,s,n\mid t) + sn^{2+\rho}/r + (sn^{1-\rho}/t) M(n,rn^{\rho},n)\right).
\]
Setting $s=n^{\beta-\eps'}$, $t=n^{(1+\kappa)\beta}$, and $r=n^{\beta}$ with $\eps'=\eps/2$ yields
\[ M^*(n,n^\rho,n\mid n^{1-\rho})\ =\ \OO\left(n^{\rho-\beta+\eps'}M^*(n,n^\beta,n\mid n^{(1+\kappa)\beta}) + n^{2+\rho-\eps'} + n^{1-\rho-\kappa\beta+\omega(1,\rho+\beta,1)-\eps'}\right),
\]
which is $\OO(n^{2+\rho-\Omega(\eps)})$,
since $\omega(1,\rho+\beta,1) \le \omega(1,\rho,1) + 
\frac{\omega(1,\sigma,1) - \omega(1,\rho,1)}{\sigma-\rho}\beta
= 1+2\rho + \kappa\beta$, by convexity of $\omega(1,\cdot,1)$.

The above assumes $(1+\kappa)\beta\le 1-\rho$ and $\beta\le\sigma-\rho$, but 
this assumption may be removed, since Lemma~\ref{lem:down} allows us to replace $\beta$ with any sufficiently small positive constant $\gamma$.
\end{proof}

We pick $\sigma=0.85$.  By known bounds~\cite{legallurr}, $\omega(1,0.85,1)<2.258317$.
Since $\rho\ge 0.5$, we have $\kappa\le \frac{1.258317-2\rho}{0.85-\rho}<0.7381$.
If $\omega=2$, then $\kappa=0$.

\subsection{Min-Witness Product}

Combining Corollary~\ref{cor:unwtdirapsp} and Lemma~\ref{lem:reduce:from:minplus:minwit} immediately gives the following:

\begin{corollary}\label{cor:lower-bounds-under-uAPSP:minwit}
Let $\rho$ and $\kappa$ be as in Corollary~\ref{cor:unwtdirapsp}.
For any constant $\beta$,
if $\Tminwit(n,n^{(3+2\kappa)\beta},n)=O(n^{2+\beta-\eps})$, then
$\TunwtdirAPSP(n)=O(n^{2+\rho-\Omega(\eps)})$.
\end{corollary}

By setting $\beta=1/(3+2\kappa)$ (which is $1/3$ if $\omega=2$, or $<0.223$ regardless), we have thus proved that \MinWitness{} of two $n\times n$ Boolean matrices
cannot be computed in $O(n^{7/3-\eps})$ time if $\omega=2$, or
$O(n^{2.223})$ time regardless of the value of $\omega$, under the \uAPSPH{}. 
(This is better than the near-$n^{11/5}$ lower bound we obtained from the \StrongAPSP{}.)
The question of proving lower bounds for \MinWitness{} from the \uAPSPH{} was left open in the paper by
Chan, Vassilevska W. and Xu~\cite{CVXicalp21} (they were only able to do so for \MinWitnessEq{}).

Furthermore, by setting $\beta=\gamma/(3+2\kappa)$, 
\MinWitness{} of an $n\times n^{\gamma}$ and an $n^{\gamma}\times n$ Boolean matrix cannot be computed in $O(n^{2+0.223\gamma-\eps})$ time for any $\gamma\le 1$ under the same hypothesis.

\subsection{All-Pairs Shortest Lightest Paths}

Combining Corollary~\ref{cor:unwtdirapsp} and Lemma~\ref{lem:reduce:from:minplus:apslp} immediately gives the following:

\begin{corollary}\label{cor:lower-bounds-under-uAPSP:apslp}
Let $\rho$ and $\kappa$ be as in Corollary~\ref{cor:unwtdirapsp}.
For any constant $\beta\le \frac1{3+2\kappa}$
and $\TundirAPSLP(n,n^{1+\beta})=O(n^{2+\beta-\eps})$, then
$\TunwtdirAPSP(n)=O(n^{2+\rho-\Omega(\eps)})$.
\end{corollary}

By setting $\beta=1/(3+2\kappa)$, we have thus proved that \APSLP{} cannot be solved
in $O(n^{7/3-\eps})$ time if $\omega=2$, or
$O(n^{2.223})$ time regardless of the exact value of $\omega$, under the \uAPSPH{}. 
(This is better than the near-$n^{11/5}$ lower bound we obtained from the \StrongAPSP{}.)
The same result holds for \APLSP{}.  Previously, Chan, Vassilevska W. and Xu~\cite{CVXicalp21} proved a still better near-$n^{\rho}$ lower bound for $\{0,1\}$-weighted  \APLSPIntro{}  from the same hypothesis, but their proof crucially relied on zero-weight edges and also
did not work for \APSLPIntro{} (leaving open the question of finding nontrivial conditional lower bounds for both \APLSP{} and \APSLP{}, which we answer here).

\subsection{Batched Range Mode}

By combining  Chan, Vassilevska W. and Xu's observation that 
$\TunwtdirAPSP(n)=\OO(\max_\ell M^*(n,n/\ell,n\mid \ell))$~\cite{CVXicalp21,zwickbridge} with
Lemma~\ref{lem:reduce:from:minplus:mode}, and
setting $n=\sqrt{N}$, we see that
\BatchMode{} for $N$ queries on $N$ elements cannot be
solved in $O(N^{5/4-\eps})$ time if $\omega=2$, or
in $O(N^{1+\rho/2-\eps})$ time regardless, under the \uAPSPH{}.

To obtain a lower bound for general $Q$, we
combine Corollary~\ref{cor:unwtdirapsp} and Lemma~\ref{lem:reduce:from:minplus:mode}:

\begin{corollary}\label{cor:lower-bounds-under-uAPSP:mode}
Let $\rho$ and $\kappa$ be as in Corollary~\ref{cor:unwtdirapsp}.
For any constant $\beta$,
if $\Trangemode(n^{1+(2+\kappa)\beta},n^2\mid n^\beta)=O(n^{2+\beta-\eps})$, 
then $\TunwtdirAPSP(n)=O(n^{2+\rho-\Omega(\eps)})$.
\end{corollary}

Thus, by setting $n=\sqrt{Q}$ and $n^\beta=(N/\sqrt{Q})^{1/(2+\kappa)}$, \BatchMode{} of $Q$ queries on $N$ elements cannot be answered in
 $O(Q^{3/4}N^{1/2-\eps})$ time if $\omega=2$,
or $O(Q^{1-0.365/2}N^{0.365-\eps})$ time 
regardless, for any $Q\le N^2$, under the \uAPSPH{}.  
For example, for $Q=N^{1.6}$, the lower bound is near $N^{1.7}$
if $\omega=2$, or near $N^{1.673}$ regardless.
(This is slightly better than the lower bound we obtained in previous section from the \StrongAPSP{}.)

\subsection{Dynamic Shortest Paths in Planar Graphs}

Combining Corollary~\ref{cor:unwtdirapsp} and Lemma~\ref{lem:dynplanar} with $\alpha=\beta$
immediately gives the following:

\begin{corollary}\label{cor:lower-bounds-under-uAPSP:planar}
Let $\rho$ and $\kappa$ be as in Corollary~\ref{cor:unwtdirapsp}.
For any constant $\beta$,
if
$\Tdynplanar(n^{(4+\kappa)\beta},n^{1+\beta},n^{1+\beta})=O(n^{1+2\beta-\eps})$,
then
$\TunwtdirAPSP(n)=O(n^{2+\rho-\Omega(\eps)})$.
\end{corollary}

By setting $\beta=1/(3+\kappa)$
and $N=n^{1+\beta}$, we have thus proved that 
an offline sequence
of $N$ shortest path queries and $N$ updates on an unweighted, undirected $N$-node planar
graph cannot be processed in $O(N^{5/4-\eps})$ time if $\omega=2$, 
of $O(N^{1.211})$
time regardless, under the
\uAPSPH{}.  
(This is slightly better than the lower bound we obtained in previous section from the \StrongAPSP{}.)

\subsection{An Equivalence Result}

We also obtain an interesting equivalence result:

\begin{corollary}
Let $\alpha$ be such that $\omega(1,\alpha,1)=2$.
For any constants $\beta,\gamma\in (0,\alpha)$,
there exists $\eps>0$ such that $M^*(n,n^\beta,n\mid n^\beta)=O(n^{2+\beta-\eps})$ if and only if there exists $\eps'>0$ such that
$M^*(n,n^\gamma,n\mid n^\gamma) = O(n^{2+\gamma-\eps'})$.
\end{corollary}
\begin{proof}
W.l.o.g., assume $\gamma<\beta$.
The ``only if'' direction is shown in Lemma~\ref{lem:down}.
For the ``if'' direction,
suppose $M^*(n,n^\gamma,n\mid n^\gamma) = O(n^{2+\gamma-\eps'})$.
By Theorem~\ref{thm:main},
\[ M^*(n,n^\beta,n\mid n^\beta)\ =\ \OO\left((n^{\beta}/s)M^*(n,s,n\mid t) + sn^{2+\beta}/r + (sn^\beta/t) M(n,rn^\beta,n)\right).
\]
Setting $s=n^{\gamma-\eps}$, $t=n^{\gamma}$, and
$r=n^\gamma$, with $\eps=\eps'/2$
yields $M^*(n,n^\beta,n\mid n^\beta)=O(n^{2+\beta-\Omega(\eps')})$,
assuming that $\gamma\le \alpha-\beta$.

This assumption may be removed, since Lemma~\ref{lem:down} allows
us to replace $\gamma$ with a sufficiently small positive constant $\gamma'$.
\end{proof}

\begin{corollary}
If $\omega=2$,
then for any constant $\beta\in (0,1)$,
there exists $\eps>0$ such that
$\TunwtdirAPSP(n)=O(n^{2.5-\eps})$ iff
there exists $\eps'>0$ such that $M^*(n,n^\beta,n\mid n^\beta)=\OO(n^{2+\beta-\eps'})$.
\end{corollary}
\begin{proof}
If $\omega=2$, then $\rho=1/2$ and then
Chan, Vassilevska W. and Xu's result~\cite{CVXicalp21} showed that $\TunwtdirAPSP(n)=O(n^{2.5-\eps})$ for some $\eps>0$
is equivalent to $M^*(n,\sqrt{n},n\mid \sqrt{n})=O(n^{2.5-\eps'})$
for some $\eps'>0$.
Since $\omega=2$ implies $\alpha=1$, we can apply
the preceding corollary for any $\beta\in(0,1)$ and $\gamma=1/2$.
\hspace*{\fill}
\end{proof}

Let $\Phi(\beta)$ be the claim that $M^*(n,n^\beta,n\mid n^\beta)$
is not in $O(n^{2+\beta-\eps})$ for any $\eps>0$.
If $\omega=2$, $\Phi(1)$ is just
the \StrongAPSP{}, but intriguingly, by the above corollary, $\Phi(0.99)$
is equivalent to the \uAPSPH{}, which has given us strictly better conditional lower bound results.

%% file: counting.tex
Continuing the approach in Section~\ref{sec:counting:preview} for proving equivalence between \AEExactTriCount{} and \AEExactTri{},
we now derive more equivalence results between other counting and detection problems.

\subsection{Min-Plus Product}

In this section, we use $A$ and $B$ to denote the inputs to  a \MinPlus{} or  \MinPlusCount{} instance, and we use $W_{ij}$ to denote the set of $k$ where $A_{ik}+B_{kj} = (A\star B)_{ij}$, i.e., the set of witnesses for $(i, j)$. 

\begin{lemma}
\label{lem:min-plus}
Given two $n \times n$ matrices $A, B$ and a subset $S \subseteq [n]$, we can compute a matrix $D$ in $\OO(|S| \cdot n^{(3+\omega)/2})$ time such that $D_{ij} = |W_{ij}|$ for pairs of $(i, j)$ where $S \cap W_{ij} \ne \emptyset$. 
\end{lemma}
\begin{proof}
For every $s \in S$, we do the following. Let $A^{(s)}$ be a matrix where $A^{(s)}_{ik} = A_{ik} - A_{is}$ and $B^{(s)}_{kj} = B_{sj} - B_{kj}$. Then we  compute the equality product $C^{(s)}$ of $A^{(s)}$ and $B^{(s)}$ in $\OO(n^{(3+\omega)/2})$ time for each $s$ using Matou{\v{s}}ek's algorithm~\cite{MatIPL}. Finally, let $D_{ij}$ be $C^{(s)}_{ij}$ where  $A_{is} + B_{sj}$ is the smallest over all $s \in S$ (breaking ties arbitrarily). The running time for computing $D$ is clearly $\OO(|S| \cdot n^{(3+\omega)/2})$. 

Suppose $S \cap W_{ij} \ne \emptyset$ for some $(i, j)$. Then $D_{ij}$ equals $C^{(s)}_{ij}$ where $A_{is} + B_{sj} = (A \star B)_{ij}$. For any $k$,  $A^{(s)}_{ik} = B^{(s)}_{kj}$ if and only if $A_{ik} + B_{kj} = A_{is} + B_{sj} = (A \star B)_{ij}$ by Fredman's trick. Therefore, $D_{ij} = C^{(s)}_{ij} = |W_{ij}|$.
\end{proof}

\begin{theorem}
\label{thm:minplus-count}
If \MinPlus{} for $n \times n$ matrices has an $O(n^{3-\eps})$ time algorithm for some $\eps > 0$, then \MinPlusCount{} for $n \times n$ matrices  has an $O(n^{3-\eps'})$ time algorithm for some $\eps' > 0$.
\end{theorem}
\begin{proof}
Given a \MinPlusCount{} instance on $n \times n$ matrices $A, B$, we first list up to $n^{0.99}$ elements in $W_{ij}$ for every $i, j$. 
By well-known techniques (e.g.\  \cite{focsyj}), an $O(n^{3-\eps})$ time \MinPlus{} algorithm implies an $O(n^{3-\eps''})$ for $\eps'' > 0$ time algorithm for listing up to $n^{0.99}$ witnesses for each $(i, j)$ in a \MinPlus{} instance. 

If we list less than $n^{0.99}$ elements for some $(i, j)$, then these elements are all the elements in $W_{ij}$. Thus we can output the number of elements we list as the exact witness count for $(i, j)$. For each of the remaining pairs of $(i, j)$, we have found $n^{0.99}$ witnesses. By the standard greedy algorithm for hitting set, in $\OO(n^{2.99})$ time, we can find a set $S$ of size $\OO(n^{0.01})$ that intersects with $W_{ij}$ for each of these remaining $(i, j)$ pairs. Therefore, we can apply Lemma~\ref{lem:min-plus} to compute the number of witnesses for these remaining $(i, j)$ pairs in $\OO(|S| \cdot n^{(3+\omega)/2}) \le O(n^{2.70})$ time. 

The total running time for the \MinPlusCount{} instance is thus $\OO(n^{3-\eps''} + n^{2.99}+n^{2.70})$, which is truly subcubic. 
\end{proof}

We then show the reduction in the other direction. 
The proof is similar to the reduction from a certain version of Min-Plus product to certain versions of  APSP counting in unweighted directed graphs~\cite{CVXicalp21}.
\begin{theorem}
\label{thm:minplus-count-rev}
If \MinPlusCount{} for $n \times n$ matrices has an $O(n^{3-\eps})$ time algorithm for some $\eps > 0$, then \MinPlus{} for $n \times n$ matrices  has an $O(n^{3-\eps'})$ time algorithm for some $\eps' > 0$.
\end{theorem}
\begin{proof}

Let $A, B$ be two $n \times n$ matrices of a \MinPlus{} instance. Let $A'$ be another $n \times n$ matrix where $A'_{ik} = M \cdot A_{ik} + k$ for some large enough integer $M$ (say $M > n$). Similarly, we create an $n \times n$ matrix $B'$ where $B'_{kj} = M \cdot B_{kj}$. This way, for every $i, j$, there exists exactly one $k_{ij}$ such that $A'_{ik_{ij}} + B'_{k_{ij}j} = (A'\star B')_{ij}$. Furthermore, we clearly also have $A_{i,k_{ij}} + B_{k_{ij},j} = (A \star B)_{ij}$. 

Then for each integer $p \in \left[\lceil \log(n) \rceil\right]$, we do the following. Let $A'^{(p)}$ be a copy of $A'$, but we duplicate all columns $k$ where the $p$-th bit in $k$'s binary representation is $1$. We similarly create $B'^{(p)}$ which is a copy of $B'$ but we duplicate all rows $k$ where the $p$-th bit in $k$'s binary representation is $1$. Then we run the $O(n^{3-\eps})$ time  \MinPlusConvCount{} algorithm for $A'^{(p)}$ and $B'^{(p)}$. Suppose for some $i, j$, the number of witnesses is $2$, then we know that the $p$-th bit of  $k_{ij}$ is $1$; otherwise, the $p$-th bit of $k_{ij}$ is $0$. 

After all $\left[\lceil \log(n) \rceil\right]$ rounds, we can compute $k_{ij}$ for every $i, j$. Since $A_{i, k_{ij}} + B_{k_{ij}, j} = (A \star B)_{ij}$, we can then compute the Min-Plus product between $A$ and $B$ in $\OO(n^2)$ time. 
\end{proof}

\subsection{\texorpdfstring{$3$}{3}SUM Convolution and Min-Plus Convolution}
Let $A, B$ and $C$ be the inputs of an \AllThreeSUMConv{} or  \AllThreeSUMConvCount{} instance, and let $W_{k}$ be the set of $i$ where $A_{i}+B_{k-i} =C_{k}$, i.e., the set of witnesses for $k$.

\begin{lemma}
\label{lem:3sum-conv}
Given a \AllThreeSUMConvCount{} instance for length $n$ arrays $A, B, C$, we can compute $|W_k|$ for every $k$ such that $|W_k| \ge L$ in $\OO(n^{(9+\omega)/4}/L)$ randomized time. 
\end{lemma}
\begin{proof}
First, we find an arbitrary prime $p$ between $2n$ and $4n$, which can be done in $\OO(n)$ time. Also, let $x$ be a uniformly random number sampled from $\mathbb{F}_p \setminus \{0\}$ and let $y$ be a uniformly random number sampled from  $\mathbb{F}_p$. 
Then we create three arrays $A', B'$ and $C'$, indexed by $\mathbb{F}_p$ as follows:
\begin{equation*}
    \begin{split}
        A'_{i} &= \left\{
  \begin{array}{ll}
    A_{x^{-1} (i-y) \bmod{p}} & : x^{-1} (i-y) \bmod{p} \in [n]\\
    M & : \text{otherwise}
  \end{array}
\right.,\\
B'_{i} &= \left\{
  \begin{array}{ll}
    B_{x^{-1} (i+y) \bmod{p}} & :  x^{-1} (i+y) \bmod{p} \in [n]\\
    M & : \text{otherwise}
  \end{array}
\right.,\\
C'_{i} &= \left\{
  \begin{array}{ll}
    C_{x^{-1} i \bmod{p}\phantom{(+y)}}  & :  x^{-1} i \bmod{p} \in [n]\\
    3M & : \text{otherwise}
  \end{array}
\right.,
    \end{split}
\end{equation*}
where $M$ is a large enough number (say $M$ is larger than $10$ times the largest absolute value of the input numbers). 

If we use $W'_k$ to denote the set of $i$ such that $A'_{i} + B'_{(k - i) \bmod{p}} = C'_k$, then it is not difficult to verify that $|W'_{xk \bmod{p}}| = |W_k|$. Thus, from now on, we aim to compute $|W'_k|$ for indices $k$ where $|W'_k| \ge L$. 

We start with the following claim. 
\begin{claim}
\label{cl:3sum-conv-algo}
Let $\mathcal{I} \subseteq \mathbb{F}_p$ be any fixed interval of length $\Theta(\sqrt{n})$ and let $1 \le L' \le \sqrt{n}$ be a fixed value. Then there exists an $\OO(n^{(5+\omega)/4}/L')$ time algorithm that computes $|W_k' \cap \mathcal{I}|$ for every $k$ such that $|W_k' \cap \mathcal{I}| \ge L'$, with high probability. Furthermore, for other values of $k$, we either also compute $|W_k' \cap \mathcal{I}|$ correctly, or declare that we don't know the value of $|W_k' \cap \mathcal{I}|$. 
\end{claim}
\begin{proof}
We first reduce the problem of computing $|W_k' \cap \mathcal{I}|$ to an instance of \AEExactTriCount{}. Similar reductions from convolution problems to matrix-product type problems were known before \cite{bremner2006necklaces, VWfindingcountingj}. Without loss of generality, we assume $\mathcal{I} = \{0, 1, \ldots, \ell - 2, \ell - 1\}$ for some $\ell = \Theta(\sqrt{n})$, by subtracting $\min \mathcal{I}$ from all indices of $A'$ and adding $\min \mathcal{I}$ to all indices of $B'$. 

We then create the following tripartite weighted graph $G$ with three parts $I, J, T$, where $|I| = \ell, |J| = [\lceil p/\ell \rceil]$ and $|T| = 2\ell - 1$. We use $I_i$ to denote the $i$-th node in $I$, $J_j$ to denote the $j$-th node in $J$ and $T_t$ to denote the $t$-th node in $T$. We then add the following edges to the graph: 
\begin{itemize}
    \item For every $i \in [|I|], t \in [|T|]$ such that $i - t + \ell - 1 \in \mathcal{I}$, we add an edge between $I_i$ and $T_t$ with weight $w(I_i, T_t) = A'_{i - t + \ell - 1}$.
    \item For every $t \in [|T|]$ and $j \in [|J|]$, we add an edge between $T_t$ and $J_j$ with weight $w(T_t, J_j) = B'_{((j-1)\ell + t - \ell) \bmod{p}}$. 
    \item For every $i \in [|I|], j \in [|J|]$ such that $(j-1)\ell + i - 1 < p$, we add an edge between $I_i$ and $J_j$ with weight $w(I_i, J_j) = -C'_{(j-1)\ell + i - 1}$. 
\end{itemize}
Consider any $(i, j) \in [\ell] \times [\lceil p/\ell \rceil]$ such that $(j-1)\ell + i - 1 < p$. The nodes $T_t$ such that $i - t + \ell - 1 \in \mathcal{I}$ form triangles with edge $(I_i, J_j)$. The multiset of the weights of these triangles is 
\begin{equation*}
    \begin{split}
        \left\{w(I_i, T_t) + w(J_j, T_t) + w(I_i, J_j) \right\}_{t=i}^{i+\ell-1} &= \left\{A'_{i - t + \ell - 1} + B'_{((j-1)\ell + t - \ell) \bmod{p}} - C'_{(j-1)\ell + i - 1} \right\}_{t=i}^{i+\ell-1}\\
        &= \left\{A'_r + B'_{\left( (j-1)\ell + (i-1) - r\right) \bmod{p}} -C'_{(j-1)\ell + i - 1}\right\}_{r=0}^{\ell - 1}.
    \end{split}
\end{equation*}
Thus, the number of triangles with weight $0$ containing edge $(I_i, J_j)$ in $G$ is exactly  $|W'_k \cap \mathcal{I}|$ for $k = (j-1)\ell + i - 1$.
In particular, if $|W'_k \cap \mathcal{I}| \ge L'$, then the number of witnesses for $(I_i, J_j)$ in the \AEExactTriCount{} instance on graph $G$ and target value $0$ is also at least $L'$. Now let $S$ be a random subset of $V(G)$ of size $C n^{0.5} \log n/L'$ for a sufficiently large constant $C$. Then with high probability, $S$ intersects with the set of witnesses for every edge $(I_i, J_j)$ for which  $k=(j-1)\ell + (i-1)$ has at least $L'$ witnesses in $\mathcal{I}$. Now we can apply Lemma~\ref{lem:exact-tri} on graph $G$, target value $0$ and set $S$ to compute the number of witnesses for these edges $(I_i, J_j)$ in $\OO(|S| \sqrt{n}^{(3+\omega)/2}) = \OO(n^{(5+\omega)/4}/L')$ time. 

If $S$ does not intersect with the witnesses for some edge $(I_i, J_j)$ (which is easy to check in $O(|S|n)$ total time), we declare that we don't know the value of $|W_k' \cap \mathcal{I}|$ for  $k=(j-1)\ell + (i-1)$.
\end{proof}

\begin{claim}
\label{cl:3sum-conv-chebyshev}
Let $\mathcal{I} \subseteq \mathbb{F}_p$ be any fixed interval and $k \in [n]$ be any fixed index. If $|W_k| \ge L$, then 
$$\Pr_{\substack{x \sim \mathbb{F}_p \setminus \{0\} \\ y \sim \mathbb{F}_p}}\left[ \left| W'_{xk \pmod{p}} \cap \mathcal{I} \right| \le \frac{L|\mathcal{I}|}{2p}\right] = O(\frac{n}{L|\mathcal{I}|}).$$
\end{claim}
\begin{proof}
First, note that $W'_{xk \pmod{p}} = \{xw + y \pmod{p}: w \in W_k\}$. Let $X$ be the random variable denoting $\left| W'_{xk \pmod{p}} \cap \mathcal{I}\right|$. First, since for any $w \in W_k$, $xw + y \pmod{p}$ is uniformly at random, $\Pr[xw + y \pmod{p} \in \mathcal{I}] = \frac{|\mathcal{I}|}{p}$, and consequently $\mathbb{E}[X] = \frac{|W_k||\mathcal{I}|}{p}$. 

For any $w, w' \in W_k$ where $w \ne w'$, the probability that both $xw + y \pmod{p}$ and $xw' + y \pmod{p}$ fall in $\mathcal{I}$ can be expressed as
$$\sum_{i_1, i_2 \in \mathcal{I}} \Pr[xw+y \equiv i_1 \pmod{p} \wedge xw'+y \equiv i_2 \pmod{p}]. $$
If $i_1 = i_2$, then $xw+y \equiv i_1 \pmod{p}$ and $xw'+y \equiv i_2 \pmod{p}$ cannot both happen; otherwise, there exists at most one pair $(x, y) \in \mathbb{F}_p \times \mathbb{F}_p$ for which $xw+y \equiv i_1 \pmod{p}$ and $xw'+y \equiv i_2 \pmod{p}$ are both true. Thus, $\Pr[xw + y \pmod{p} \in \mathcal{I} \wedge xw' + y \pmod{p} \in \mathcal{I}] \le \frac{|\mathcal{I}|(|\mathcal{I}|-1)}{p(p-1)} \le \frac{|\mathcal{I}|^2}{p^2}$. Thus, $$\Var[X] = \mathbb{E}[X^2] - \mathbb{E}[X]^2 \le \left(|W_k|(|W_k|-1)\frac{|\mathcal{I}|^2}{p^2} + \mathbb{E}[X] \right)- \mathbb{E}[X]^2 \le \mathbb{E}[X].$$
Also, 
\begin{equation*}
    \Pr\left[X \le \frac{L|\mathcal{I}|}{2p}\right] \le \Pr\left[X \le \frac{|W_k||\mathcal{I}|}{2p}\right] \le \Pr\left[\left|X - \mathbb{E}[X] \right| \le \frac{1}{2} \mathbb{E}[X]\right].
\end{equation*}
By Chebyshev's inequality, this probability can be upper bounded by $\frac{\Var[X]}{(\frac{1}{2} \mathbb{E}[X])^2} = O(\frac{n}{L|\mathcal{I}|})$. 
\end{proof}

We now describe our algorithm for computing $|W_k|$. First, we split $\mathbb{F}_p$ into $\ell = \Theta(\sqrt{n})$ intervals $\mathcal{I}_1, \mathcal{I}_2, \ldots, \mathcal{I}_\ell$, each of size $\Theta(\sqrt{n})$. Then it suffices to compute $|W'_{xk \bmod{p}} \cap \mathcal{I}_i|$ for each $i \in [\ell]$, since $|W_k| = |W'_{xk \bmod{p}}| = \sum_{i=1}^\ell |W'_{xk \bmod{p}} \cap \mathcal{I}_i|$. 

We first run the algorithm in Claim~\ref{cl:3sum-conv-algo} for each $i$ with $L' = \frac{L |\mathcal{I}_i|}{2p}$, which takes $\OO(n^{(5+\omega)/4} / L') = \OO(n^{(7+\omega)/4}/L)$ time. Claim~\ref{cl:3sum-conv-algo}  computes $|W'_{xk \bmod{p}} \cap \mathcal{I}_i|$ as long as $|W'_{xk \bmod{p}} \cap \mathcal{I}_i| \ge L'$. For each fixed $k$, it fails to compute $|W'_{xk \bmod{p}} \cap \mathcal{I}_i|$ with probability $O(\sqrt{n}/L)$ by Claim~\ref{cl:3sum-conv-chebyshev}. For these $k$, we enumerate over $j \in \mathcal{I}_i$, check if $j \in W'_{xk \bmod{p}}$, and then compute  $|W'_{xk \bmod{p}} \cap \mathcal{I}_i|$. In expectation, the cost of these $k$ is $O(\frac{n \cdot \sqrt{n}}{L}  \cdot |\mathcal{I}_i|) = O(n^2 / L)$. 

Summing over all $i \in [\ell]$, the total expected running time of the algorithm is $\OO(n^{(9+\omega)/4}/L)$. 
\end{proof}

\begin{theorem}
\label{thm:all-3sum-conv-count}
If \AllThreeSUMConv{} for length $n$ arrays has an $O(n^{2-\eps})$ time algorithm for some $\eps > 0$, then \AllThreeSUMConvCount{} for length $n$ arrays  has an $O(n^{2-\eps'})$ time randomized algorithm for some $\eps' > 0$
\end{theorem}
\begin{proof}

Similar as before, given an \AllThreeSUMConvCount{} instance on length $n$ arrays $A, B, C$, we can count the number of witnesses for $C_k$ that have at most $n^{0.99}$ witnesses in $O(n^{2-\eps''})$ time by well-known techniques~\cite{focsyj} when \AllThreeSUMConv{} has a truly subquadratic algorithm.

For the rest values of $k$, we run the algorithm in Lemma~\ref{lem:3sum-conv} which runs in $\OO(n^{(9+\omega)/4} / n^{0.99}) = O(n^{1.86})$ time. 

Overall, the algorithm for \AllThreeSUMConvCount{} runs in $O(n^{2-\eps''} + n^{1.86})$ time, which is truly subquadratic. 
\end{proof}

As in Remark~\ref{rem:exact-tri-count}, Theorem~\ref{thm:all-3sum-conv-count} implies that \ThreeSUMConv{} is subquadratically equivalent to \ThreeSUMConvCount{}.

\begin{theorem}
\label{thm:minplus-conv-count}
If \MinPlusConv{} for length $n$ arrays has an $O(n^{2-\eps})$ time algorithm for some $\eps > 0$, then \MinPlusConvCount{} for length $n$ arrays  has an $O(n^{2-\eps'})$ time randomized algorithm for some $\eps' > 0$.
\end{theorem}
\begin{proof}

Given an \MinPlusConvCount{} instance for length $n$ arrays $A, B$, we first run the assumed \MinPlusConv{} algorithm to compute the Min-Plus convolution $C$ of $A$ and $B$ in $O(n^{2-\eps})$ time. 

Similar as before, given the $O(n^{2-\eps})$ time algorithm for \MinPlusConv{}, we can count the number of witnesses for $C_k$ that have at most $n^{0.99}$ witnesses in $O(n^{2-\eps''})$ time by well-known techniques~\cite{focsyj}.

For the rest values of $k$, we run the algorithm in Lemma~\ref{lem:3sum-conv} with arrays $A, B, C$ and $L = n^{0.99}$ which runs in $\OO(n^{(9+\omega)/4} / n^{0.99}) = O(n^{1.86})$ time. 

Overall, the algorithm for \MinPlusConvCount{} runs in $O(n^{2-\eps}+ n^{2-\eps''} + n^{1.86})$ time, which is truly subquadratic. 
\end{proof}

We then show a reduction from \MinPlusConv{} to \MinPlusConvCount{}. 
\begin{theorem}
\label{thm:minplus-conv-count-rev}
If \MinPlusConvCount{} for length $n$ arrays has an $O(n^{2-\eps})$ time algorithm for some $\eps > 0$, then \MinPlusConv{} for length $n$ arrays  has an $O(n^{2-\eps'})$ time randomized algorithm for some $\eps' > 0$.
\end{theorem}
\begin{proof}
Let $A$ and $B$ be two length $n$ arrays for a \MinPlusConv{} instance. Let $C$ be their Min-Plus convolution. As in proof of Theorem~\ref{thm:minplus-count-rev},
we can assume $|W_k|=1$ for every $k$, i.e., there exists a unique $i_k$ such that $A_{i_k} +B_{k - i_k} = C_k$. 

For each $p \in \left[ \lceil \log(n)\rceil\right]$, we perform the following round. Let $A'$ be a length $2n$ array such that $A'_{2i-1} = A_i$ for every $i \in [n]$, $A'_{2i} = A_i$ for every $i \in [n]$ whose $p$-th bit in its binary representation is $1$, and $A'_{2i} = \infty$ for the rest of $i$. Also, let $B'$ be a length $2n$ array such that $B'_{2i-1} = B'_{2i} = B_i$ for every $i \in [n]$. Now we use the \MinPlusConvCount{} algorithm for arrays $A'$ and $B'$. Suppose $C'$ is the Min-Plus convolution between $A'$ and $B'$. Clearly, for any $k \in [n]$, $C'_{2k-1} = C_k$. Also, suppose $C'_{2k-1}$ has $2$ witnesses, then we know that $A'_{2i_k} = A_{i_k}$ and thus the $p$-th bit in $i_k$-th binary representation is $1$; otherwise the $p$-th bit in $i_k$-th binary representation is $0$. 

Thus, after the $\lceil \log(n)\rceil$ rounds, we can compute $i_k$ for each $k \in [n]$, which can then be used to compute the Min-Plus convolution $C$ between $A$ and $B$ in $O(n)$ time. 
\end{proof}

\subsection{All-Numbers 3SUM}

\begin{theorem}
\label{thm:all-3sum-count}
If \AllThreeSUM{} for  sets of $n$ numbers has an $O(n^{2-\eps})$ time algorithm for some $\eps > 0$, then \AllThreeSUMCount{} for sets of $n$ numbers  has an $O(n^{2-\eps'})$ time randomized algorithm for some $\eps' > 0$.
\end{theorem}
\begin{proof}
If \AllThreeSUM{} for sets of $n$ numbers has an $O(n^{2-\eps})$ time algorithm for $\eps > 0$, then so does \AllThreeSUMConv{} for length $n$ arrays, since \AllThreeSUMConv{} is not harder than \AllThreeSUM{}. Then by Theorem~\ref{thm:all-3sum-conv-count}, \AllThreeSUMConvCount{} for length $n$ arrays has an $O(n^{2-\eps''})$ time algorithm for some $\eps'' > 0$. Therefore, it suffices to reduce \AllThreeSUMCount{} to \AllThreeSUMConvCount{}. Some previous reductions from \ThreeSUM{} to \ThreeSUMConv{} actually work for the counting variants as well~\cite{patrascu2010towards, ChanHe}. Arguably the simplest such reduction is given in~\cite[Section 3]{ChanHe}. Applying their reduction finishes the proof.
\end{proof}

As in Remark~\ref{rem:exact-tri-count}, Theorem~\ref{thm:all-3sum-count} implies that \ThreeSUM{} is subquadratically equivalent to \ThreeSUMCount{}.

Still more equivalence results for other counting and detection problems are given in Appendix~\ref{sec:more_counting}.

\subsection{Discussion}\label{sec:counting:discuss}

Abboud, Feller and Weimann \cite{AbboudFW20} showed that counting the number of Negative Triangles in a graph (even mod 2) can solve \ExactTri, thus presenting a barrier to showing that the Negative Triangle  problem ({\sf Neg-Tri}) is equivalent to its counting variant: Vassilevska W. and Williams \cite{focsyj} showed that {\sf Neg-Tri} is equivalent to \APSP{} under subcubic fine-grained reductions; then if {\sf \#Neg-Tri} can be reduced to {\sf Neg-Tri}, one can also reduce it to \APSP, and since there are fine-grained reductions from \ThreeSUM{}  to \ExactTri{} \cite{VWfindingcountingj}, and from \ExactTri{} to {\sf \#Neg-Tri} \cite{AbboudFW20}, one would get a very surprising reduction from \ThreeSUM{} to \APSP. There is some evidence that such a reduction would be difficult to obtain: for instance, while \APSP{} has a superlogarithmic improvement over its simple cubic algorithm \cite{Williams18}, the best improvement over the simple quadratic algorithm of \ThreeSUM{} only shaves two logarithmic factors (e.g. \cite{baran2005subquadratic})!

Our equivalences between \MinPlus{} (and thus Minimum Weight Triangle) and \ExactTri{} respectively with their counting variants exhibit a strange phenomenon: {\sf Neg-Tri} seems different from these problems! Or, perhaps, if we believe that {\sf Neg-Tri} is like these problems and is equivalent to {\sf \#Neg-Tri}, then we should be more optimistic about the existence of a fine-grained reduction from \ThreeSUM{} to \APSP.

Another line of work in which counting variants of fine-grained problems have been considered is in worst-case to average-case reductions and {\em fine-grained cryptography} \cite{BallRSV17,BallRSV18,Boix-AdseraBB19,GO2020,DalirrooyfardLW20,LaVigneLW19,merkle}: building cryptographic primitives from worst-case fine-grained assumptions that might still hold even if $\text{P}=\text{NP}$. The known techniques 
 for worst-case to average-case reductions for fine-grained problems only work for counting problems, whereas the design of fine-grained public key protocols \cite{LaVigneLW19,merkle} seem to require that the decision variants are hard on average.

Suppose that one can use the known toolbox for worst-case to average-case reductions for counting problems to show that \ExactTriCount{} or \ThreeSUMCount{} is hard on average.
Then via our reductions back to \ExactTri{} and \ThreeSUM{}, one would get some distributions for which these decision problems are actually hard. This could pave the way to new public-key protocols.

%% file: co-nondet.tex
\newcommand{\NcoNTIME}{(N\cap\textit{coN})\textit{TIME}}

\subsection{Co-Nondeterministic Algorithms for Counting Problems with Integer Inputs}
\label{sec:nondet-int}

First, we show how to modify the co-nondeterministic algorithms in \cite{carmosino2016nondeterministic} to work for the counting versions of \AEExactTri{}, \AENegTri{} and \AllThreeSUM{}. 

Before we show these algorithms, we start with the following simple observation.
\begin{claim}
\label{cl:sumof3neg}
For any integers $a, b, c$, $a+b+c \le 0$ if and only if one of the followings is true: 
\begin{itemize}
    \item $\lceil \frac{a}{2} \rceil + \lceil \frac{b}{2} \rceil  + \lceil \frac{c}{2} \rceil \le 0$;  
    \item At least $2$ of $a,b,c$ are odd and $\lceil \frac{a}{2} \rceil + \lceil \frac{b}{2} \rceil  + \lceil \frac{c}{2} \rceil = 1$. 
\end{itemize}
\end{claim}

It has the following immediate consequence:
\begin{lemma}
\label{lem:negtri_to_exacttri}
Given a \AENegTriCount{} instance on an $n$-node tripartite graph $G$ with node partitions $A$, $B$, and $C$ and with edge weights in $[\pm n^{O(1)}]$, we can reduce it to $O(\log n)$ instances of \AEExactTriCount{} on graphs with the same vertex set and with edge weights in  $[\pm n^{O(1)}]$ deterministically in $\OO(n^2)$ time. Furthermore, the count for an edge in the original \AENegTriCount{}  instance can be written as a sum of the count for the corresponding edge among the \AEExactTriCount{} instances (if this edge does not exist in a certain instance, then the count is simply $0$). 
\end{lemma}
\begin{proof}
Without loss of generality, we assume that the instance $G$ has weight function $w$ and  we would like to compute for every $a\in A,c\in C$, the number of $b\in B$ such that $w(a,b)+w(b,c)+w(a,c)<0$. Then we add $1$ to the weight of every edge $(a, c) \in A \times C$. Now we need to count the number of number of $b\in B$ such that $w(a,b)+w(b,c)+w(a,c)\le 0$.

Consider recursing with respect to the bound $[\pm W]$ on the edge weights, initially $W = n^{O(1)}$. 
\begin{enumerate}
    \item If $W = O(1)$. Then we enumerate all possible combinations of $3$ weights $w_1, w_2, w_3 \in [\pm W]$ such that $w_1 + w_2 + w_3 \le 0$, and for each combination, we create an \AEExactTriCount{} instance with all edges $(a, b) \in A \times B$ s.t. $w(a, b) = w_1$, all edges $(b, c) \in B \times C$ s.t. $w(b, c) = w_2$ and all edges $(c, a) \in C \times A$ s.t. $w(c, a) = w_3$. In these instances, we make all edge weights $0$, so that all triangles have weights $0$. Thus, the total number of exact triangles through edge $(a, c)$ in these $O(1)$ instances is exactly the number of non-positive-weight triangles through edge $(a, c)$ in the original instance. 
    \item Otherwise, consider Claim~\ref{cl:sumof3neg}. A triangle $(a, b, c)$ has $w(a, b) + w(b, c) + w(c, a) \le 0$ if and only if $\lceil \frac{w(a, b)}{2} \rceil + \lceil \frac{w(b, c)}{2} \rceil  + \lceil \frac{w(c, a)}{2} \rceil \le 0$ or at least $2$ of $w(a, b),w(b, c), w(c, a)$ are odd and $\lceil \frac{w(a, b)}{2} \rceil + \lceil \frac{w(b, c)}{2} \rceil  + \lceil \frac{w(c, a)}{2} \rceil = 1$. Note that these two cases are disjoint, so we can separately consider them and sum up the counts. For the first case, we can create a graph $G'$ by replacing the  weight $w(u, v)$ of each edge with $\lceil \frac{w(u, v)}{2}\rceil$ and recursively solve the problem on $G'$. For the second case, we enumerate which subset of  $w(a, b), w(b, c)$ and $w(c, a)$ are odd (there should be at least two of them) and keep the corresponding edges in $G'$ only if they meet the parity condition. This way, we will create $4$ sub-problems, where each sub-problem is an exact triangle instance with target value $1$. 
\end{enumerate}
Overall, we will create $O(\log n)$ \AEExactTriCount{} instances as initially $W=n^{O(1)}$. 
\end{proof}
\begin{remark}\rm
Lemma~\ref{lem:negtri_to_exacttri} implies that \AENegTriCount{} subcubically reduces to \AEExactTriCount{}. As shown in~\cite{AbboudFW20}, \ExactTriCount{} reduces to \NegTriCount{}, and the same reduction works from  \AEExactTriCount{} to \AENegTriCount{}. Thus, \AENegTriCount{} and \AEExactTriCount{} are subcubically equivalent. By Theorem~\ref{thm:exact-tri-count}, they are also subcubically equivalent to \AEExactTri{}. 
\end{remark}

\begin{theorem}
\label{thm:nondet-int-exacttri-count}
\AEExactTriCount{} for graphs with integer weights in $[\pm n^\alpha]$ for any constant $\alpha$ is in $\NcoNTIME[\OO(n^{(3+\omega)/2})]$. The same bound holds for \AENegTriCount{}. 
\end{theorem}
\begin{proof}

Suppose we are given an instance of \AEExactTriCount{}. Without loss of generality, we assume the instance is a weighted graph $G$ with node partitions $A$, $B$, and $C$ and weight function $w$, and  we would like to compute for every $a\in A,c\in C$, the number of $b\in B$ such that $w(a,b)+w(b,c)+w(a,c)=0$. We will actually more generally count the number of negative-weight, zero-weight, and positive-weight triangles through each edge in $A \times C$. 

Both the prover and the verifier run the reduction in Lemma~\ref{lem:negtri_to_exacttri} on $G$ to get \AEExactTriCount{} instances $G_{<}^{(1)}, \ldots, G_{<}^{(O(\log n))}$. They also similarly run the reduction in Lemma~\ref{lem:negtri_to_exacttri} on $G$ but with all edge weights negated to  get \AEExactTriCount{} instances $G_{>}^{(1)}, \ldots, G_{>}^{(O(\log n))}$. Without loss of generality, we can assume these \AEExactTriCount{} instances all have target value  $0$. By construction in Lemma~\ref{lem:negtri_to_exacttri}, all these graphs have the same vertex set $A \cup B \cup C$.  Also, let $q$ be a constant to be fixed later. 

The prover provides the following for each of the graphs $G, G_{<}^{(1)}, \ldots, G_{<}^{(O(\log n))}, G_{>}^{(1)}, \ldots, G_{>}^{(O(\log n))}$ (we refer to it by $G'$):
\begin{itemize}
    \item A prime $p$ in the interval $[n^{1-q},Cn^{1-q}\log n]$ for a  large enough constant $C$.
    
    The prime $p$ is supposed to be such that the number of triangles in $G'$ that are zero mod $p$ but are nonzero otherwise is at most $O(n^{2+q})$. Such a prime is guaranteed to exist since each triangle that has nonzero weight in $[-n^\alpha,n^\alpha]$ is zero mod at most $\log(3n^\alpha)/\log(n^{1-q})$ primes in the interval $[n^{1-q},Cn^{1-q}\log n]$. The interval contains at least $n^{1-q}$ primes, so some prime must give rise to at most $\frac{n^3\log(3n^\alpha)}{\log(n^{1-q}) \cdot n^{1-q}} =O(n^{2+q})$ fake zero triangles.

    \item A set $R$ of $O(n^{2+q})$ triangles. These are supposed to be the triangles in $G'$ that are nonzero but are zero mod $p$.

\end{itemize}

The verifier first checks that $R$ contains only triangles whose weights are nonzero but are zero mod $p$. This takes $O(|R|) = O(n^{2+q})$ time.  Then it counts in $\tilde{O}(n^{1-q+\omega})$ time (using \cite{ALONGM1997}) the number of triangles that are zero mod $p$ through each edge $(a, c)$. Call it $t_{G'}(a, c)$. For each edge $(a, c)$, the  verifier subtracts the number of triangles through it in $R$ from $t_{G'}(a, c)$. Now, notice that $t_{G'}(a, c)$ must be an upper bound on the number of zero triangles through $(a, c)$; also, if $R$ contains all the triangles whose weights are nonzero but are zero mod $p$ as it is supposed to, $t_{G'}(a, c)$ will be equal to the number of zero triangles through $(a, c)$. 

Now, the verifier applies the second part of Lemma~\ref{lem:negtri_to_exacttri} to sum up the corresponding counts. For every edge $(a, c) \in A \times C$, it will get $T_{<}(a, c) = \sum_{i} t_{G_{<}^{(i)}}(a, c)$, which is supposed to be the number of negative-weight triangles through $(a, c)$; $T_{>}(a, c)=\sum_{i} t_{G_{>}^{(i)}}(a, c)$, which is supposed to be the number of positive-weight triangles through $(a, c)$; and $t_{G}(a, c)$, which is supposed to be the number of zero-weight triangles through $(a, c)$. 

Then the algorithm counts the number of triangles $s(a, c)$ through each edge $(a, c)$ in $\tO(n^\omega)$ time. 

Finally, the algorithm verifies $s(a, c) = T_{<}(a, c) + T_{>}(a, c) + t_{G}(a, c)$ and then outputs $t_G(a, c)$ for every edge $(a, c)$. 

The correctness of the algorithm relies on the following: Suppose $s_{<}(a, c), s_{=}(a, c), s_{>}(a, c)$ are the actual number of negative-weight, zero-weight, and positive-weight triangles through each edge $(a, c)$. Then by previous discussion, the algorithm is sure that $s_{<}(a, c) \le T_{<}(a, c), s_{=}(a, c) \le t_G(a, c)$ and $ s_{>}(a, c) \le T_{>}(a, c)$, and trivially $s(a, c) = s_{<}(a, c) + s_{=}(a, c) +s_{>}(a, c)$. These two combined with $s(a, c) = T_{<}(a, c) + T_{>}(a, c) + t_{G}(a, c)$ implies $s_{<}(a, c) = T_{<}(a, c), s_{=}(a, c) = t_G(a, c)$ and $s_{>}(a, c) = T_{>}(a, c)$. 

The running time is $\OO(n^{2+q}+n^{1-q+\omega})$ and is minimized for $q=(\omega-1)/2$.
The verifier's running time is thus $\OO(n^{(3+\omega)/2})$.

The above algorithm clearly also works for \AENegTriCount{}. 
\end{proof}

Similarly, we can design a sub-quadratic time co-nondeterministic algorithm for \AllThreeSUMCount{}. 

\begin{theorem}
\label{thm:nondet-int-3sum-count}
    \AllThreeSUMCount{} for size-$n$ sets $A, B, C$ of integers  in $[\pm n^\alpha]$ for any constant $\alpha$ is in $\NcoNTIME[\OO(n^{1.5})]$.
\end{theorem}
\begin{proof}[Proof sketch]
    The proof is essentially the same as the proof of Theorem~\ref{thm:nondet-int-exacttri-count}, with the following modifications. 

    Without loss of generality, we assume \AllThreeSUMCount{} asks to count the number of $(a, b) \in A \times B$ where $a + b + c = 0$ for every fixed $c \in C$. 
    
    Then, similar to Lemma~\ref{lem:negtri_to_exacttri}, we can reduce counting the number of $(a, b) \in A \times B$ where $a + b + c < 0$ and counting the number of $(a, b) \in A \times B$ where $a + b + c > 0$ to $O(\log n)$ instances of \AllThreeSUMCount{}. 

    For each of these instances on $(A', B', C')$, the prover provides a prime $p$ in the interval $[n^{1+q}, C n^{1+q} \log n]$ for a large enough constant $C$ and a constant $q$ to be fixed. It also provides a set of $R$ of $O(n^3 / n^{1+q}) = O(n^{2-q})$ triples $(a, b, c)$, which are supposed to be the triples in $A' \times B' \times C'$ whose sums are nonzero but are zero mod $p$. 

    For each of these instances, the verifier checks that $R$ contains only triples whose sum is nonzero but is zero mod $p$. Then it counts in $\OO(p) = \OO(n^{1+q})$ time (using FFT) the number of triples $(a, b, c)$ whose sum is zero mod $p$ for every $c \in C'$, and then subtracts the corresponding count in $R$ from this count. 

    Similar to Theorem~\ref{thm:nondet-int-exacttri-count}, by summing up the counts in all these instances for every $c \in C$, and checking whether that equals $|A||B|$, the verifier will be sure to get the correct count of zero-sum triples involving each $c \in C$. 

    The running time of the verifier is $\OO(n^{1+q}+n^{2-q})$, which is $\OO(n^{1.5})$ by setting $q=0.5$. 
    
\end{proof}

\subsection{Co-Nondeterministic Algorithms for Counting Problems with Real Inputs}
\label{sec:nondet-real}

The co-nondeterministic algorithms in Section~\ref{sec:nondet-int} heavily rely on the idea of modulo a prime $p$, which will no longer be possible if the input numbers are reals. In this section, we study algorithms for \AEExactTriCountReal{}, \AENegTriCountReal{} and \AllThreeSUMCountReal{}. 

For problems with real inputs in the nondeterministic model, we assume the verifier has access to a ``Reasonable'' Real RAM model, which was  discussed in \cite{CVXstoc22}. In such models, only a restricted subset of operations are allowed on the real-valued inputs, while any operations in the Word RAM model with $O(\log n)$-bit words are allowed for the integer parts of the computation. In particular, our algorithms work in the Real RAM with low-degree predicates model. Also, if we slightly change the definition of \ExactTri{} to finding a triangle $(a, b, c)$ such that $w(a, b) + w(b, c) = w(a, c)$ (and similarly for \NegTri{}), our algorithms will work in the Real RAM with 4-linear comparisons model. See \cite{CVXstoc22} for more details about ``Reasonable'' Real RAM models.

Our ideas for proving equivalence between counting and detection problems can be used to obtain new nondeterministic algorithms for counting problems with real inputs, as we show in this subsection.

We start with the following co-nondeterministic algorithm for \Equality{} and \Dominance{}. 

\begin{lemma}
\label{lem:nondet-equality}
    \Equality{} between an $n_1 \times n_2$ matrix $A$ and an $n_2 \times n_3$ matrix $B$, where we only need to determine the results on a given subset $X \subseteq [n_1] \times [n_3]$, is in $$\NcoNTIME\left[\OO\left( |X|n_2^s +  M(n_1, n_2^{2-s}, n_3) \right)\right]$$ time for any $s \in [0, 1]$. The same bound holds for \Dominance{}. 
\end{lemma}
\begin{proof}
    Without loss of generality, we can assume all entries of $A$ and $B$ are integers in $O(n_1n_2 + n_2n_3)$, by replacing each entry by its rank. 

    For every $(i, j) \in [n_1] \times [n_3]$, we use $c_{=}(i, j)$ to denote the number of $k$ where $A_{i k} = B_{k j}$, $c_{<}(i, j)$ to denote the number of $k$ where $A_{i k} < B_{k j}$ and $c_{>}(i, j)$ to denote the number of $k$ where $A_{i k} > B_{k j}$. Instead of only computing $c_{=}(i, j)$ for $(i, j) \in X$, we will more generally compute $c_{<}(i, j)$ and $c_{>}(i, j)$ as well. 

    By known reductions from \Dominance{} to \Equality{} \cite{labib2019hamming,vnotes}, we can create $O(\log n)$ instances of \Equality{} on matrices of the same dimensions, and use the sum of resulting values on entry $(i, j)$ over these $O(\log n)$ instances to compute $c_{<}(i, j)$. It similarly holds for $c_{>}(i, j)$. 

    For each of the $O(\log n)$ instances of \Equality{} (the $O(\log n)$ instances generated above, and the original instance), the prover provides the following (say the instance is on matrices $A', B'$):
    \begin{itemize}
    \item A prime $p$ in the interval $[n_2^{1-s},Cn_2^{1-s}\log (n_2)]$ for a  large enough constant $C$.
    
    The prime $p$ is supposed to be such that the number of triples $(i, k, j)$ where $(i, j) \in X$, $k \in [n_2]$ and $A'_{i k} \ne B'_{k j}$ while $A'_{i k} \equiv B'_{k j} \pmod{p}$ is at most $\OO(|X|n_2^{s})$. Such a prime is guaranteed to exist since for every triple $(i, k, j)$ with $A'_{i k} \ne B'_{k j}$, $A'_{i k}$ and $B'_{k j}$ are congruent mod at most  $\OO(1)$ primes in the interval $[n_2^{1-s},Cn_2^{1-s}\log (n_2)]$. The interval contains at least $n_2^{1-s}$ primes, so some prime must give rise to at most $\OO\left(\frac{|X|n_2}{n_2^{1-s}}\right) = \OO(|X|n_2^s)$ fake zero triangles.

    \item A set $R$ of $\OO(|X|n_2^{s})$ triples $(i, k, j)$ where $(i, j) \in X$ and $k \in [n_2]$. These are supposed to be the triples where $A'_{i k} \ne B'_{k j}$ while $A'_{i k} \equiv B'_{k j} \pmod{p}$. 

\end{itemize}

The verifier is similar to the verifier in the proof of Theorem~\ref{thm:nondet-int-exacttri-count}. For each  \Equality{} instance, after checking that the set $R$ only contains triples $(i, k, j)$ where $(i, j) \in X$, $k \in [n_2]$ and $A'_{i k} \ne B'_{k j}$ while $A'_{i k} \equiv B'_{k j} \pmod{p}$ in $\OO(|R|)$ time, it computes $C'_{i j}$, the number of $k \in [n_2]$ such that   $A'_{i k} \equiv B'_{k j} \pmod{p}$, for each pair of $(i, j) \in [n_1] \times [n_3]$ in $\OO(M(n_1, p n_2, n_3))$ time, by packing $p$ instances of matrix multiplications of dimensions $n_1 \times n_2 \times n_3$ together. Similar as before, by subtracting the number of triples involving $(i, j)$ in $R$ from $C'_{i j}$, $C'_{i j}$ becomes an upper bound on the number of $k$ such that $A'_{i k} = B'_{k j}$. 

Finally, by checking that the sum of $C'_{i j}$ over all the \Equality{} instances equals $n_2$ for every $(i, j) \in X$, the verifier will be sure that $C'_{i j}$ equals exactly the number of $k$ such that $A'_{i k} = B'_{k j}$ in each of the instances. In particular,  it can compute $c_{=}(i, j), c_{<}(i, j)$ and $c_{>}(i, j)$ for every $(i, j) \in X$. 

The running time of the verifier is $\OO(|R|+M(n_1, p n_2, n_3)) = \OO(|X|n_2^s +  M(n_1, n_2^{2-s}, n_3))$, as desired. 

\end{proof}

Next, we are ready to present our co-nondeterministic algorithm for  \AEExactTriCountReal{} and \AENegTriCountReal{}. 

\begin{theorem}
\label{thm:nondet-real-exacttri-count}
\AEExactTriCountReal{} for $n$-node graphs is in $\NcoNTIME[\OO(n^{\frac{6+\omega}{3}})]$. The same bound holds for \AENegTriCountReal{}. 
\end{theorem}
\begin{proof}
    Without loss of generality, we assume the instance is a weighted graph $G$ with node partitions $A$, $B$, and $C$ and weight function $w$, and  we would like to compute for every $a\in A,c\in C$, the number of $b\in B$ such that $w(a,b)+w(b,c)+w(a,c)=0$.  Also, let $q$ be a constant to be fixed later. 

    The prover provides the following:

    \begin{itemize}
        \item A subset $R \subseteq [n]$ of size $\OO(n^q)$. 

        For every $(a, c) \in A\times C$, let $L_{a c} = \{w(a, b) + w(b, c) : b \in R\}$. Let $p_{a c}$ be the index of the predecessor of $-w(a, c)$ (including $-w(a, c)$) in $L_{a c}$, i.e., it is $\argmax_{b \in R, w(a, b) + w(b, c) \le -w(a, c)} (w(a, b) + w(b, c))$, and let $s_{a c}$ be the index of the successor of $-w(a, c)$ (excluding $-w(a, c)$) in $L_{a c}$, i.e., it is  $\argmin_{b \in R, w(a, b) + w(b, c) > -w(a, c)} (w(a, b) + w(b, c))$. The set $R$ is supposed to be that, if $-w(a, c) \not \in L_{a c}$, then the number of $b \in B$ where $w(a, p_{a c}) + w(p_{a c}, c) < w(a, b) + w(b, c) < w(a, s_{a c}) + w(s_{a c}, c)$ is $O(n^{1-q})$. Such $R$ exists because a random $R$ satisfies these properties with high probability. 

        \item For every $b_0 \in R$, the prover  uses the protocol in Lemma~\ref{lem:nondet-equality} to create outputs for the purpose of counting $\left| \left\{w(a, b) - w(a, b_0) = w(b_0, c) - w(b,c) : b \in B \right\}\right|$ and $\left| \left\{w(a, b) - w(a, b_0) \le w(b_0, c) - w(b,c) : b \in B \right\}\right|$, where $X$ is the set of $(a, c)$ such that $b_0 = p_{a c}$ or $b_0 = s_{a c }$. 

        \item Finally, for every $(a, c)$ with $-w(a, c) \not \in L_{a c}$, it provides a list $\ell_{a c} \subseteq [n]$, which is supposed to contain all indices $b$ with $w(a, p_{a c}) + w(p_{a c}, c) < w(a, b) + w(b, c) < w(a, s_{a c}) + w(s_{a c}, c)$. Note that $|\ell_{a c}|=O(n^{1-q})$ by the choice of $R$. 
    \end{itemize}

    The verifier does the following. First, it computes $L_{a c}, p_{a c}, s_{a c}$ in $\OO(n^2 |R|) = \OO(n^{2+q})$ time. 
    Next, it uses the algorithm in Lemma~\ref{lem:nondet-equality} to correctly count, for every $b_0 \in R$, the values of $$\left| \left\{w(a, b) - w(a, b_0) = w(b_0, c) - w(b,c) : b \in B \right\}\right|$$ and $$\left| \left\{w(a, b) - w(a, b_0) \le w(b_0, c) - w(b,c) : b \in B \right\}\right|$$ for  $(a, c)$ such that $b_0 = p_{a c}$ or $b_0 = s_{a c}$. Let $x_i$ be the size of $X$ for the $i$-th call of Lemma~\ref{lem:nondet-equality}. The running time of the $i$-th call of Lemma~\ref{lem:nondet-equality} can be bounded by
    $$\OO\left(|X|n^s +  M(n, n^{2-s}, n) \right) = \OO\left(|X| n^s + n^{1-s} n^\omega\right) = \OO(n^{\frac{\omega+1}{2}} \sqrt{x_i} + n^\omega), $$
    by setting $n^s = \min\left(n, \sqrt{\frac{n^{1+\omega}}{|X|}}\right)$. 
    Then we notice that $\sum_i x_i = O(n^2)$. Therefore, by convexity, the running time can be bounded as 
    \begin{align*}
        \sum_i \OO\left(n^{\frac{\omega+1}{2}}\sqrt{x_i} + n^\omega\right) 
        &= \OO \left( |R| \cdot n^{\frac{\omega+1}{2}}\sqrt{\frac{n^2}{|R|}} + |R| \cdot n^\omega\right)
        = \OO(n^{\frac{\omega+3+q}{2}} + n^{\omega+q}).
    \end{align*}

     For some $(a, c)$, if $w(a, p_{a c}) + w(p_{a c}, c) = -w(a, c)$, then by Fredman's trick, 
        \begin{align*}
            &\left| \left\{w(a, b) - w(a, b_0) = w(b_0, c) - w(b,c) : b \in B \right\}\right|\\
            =& \left| \left\{w(a, b) + w(b,c)  = w(a, b_0) + w(b_0, c) : b \in B \right\}\right|\\
            =& \left| \left\{w(a, b) + w(b,c)  + w(a, c) = 0 : b \in B \right\}\right|
        \end{align*}
    for $b_0 = p_{a c}$, which is exactly the count we seek. 
    Otherwise, the algorithm computes the number of $b$ where $w(a, p_{a c}) + w(p_{a c}, c) < w(a, b) + w(b, c) < w(a, s_{a c}) + w(s_{a c}, c)$ via
    \begin{align*}
        &\left| \left\{w(a, b) - w(a, s_{a c}) \le w(s_{a c}, c) - w(b,c) : b \in B \right\}\right|\\ 
        -& \left| \left\{w(a, b) - w(a, s_{a c}) = w(s_{a c}, c) - w(b,c) : b \in B \right\}\right|\\
        - & \left| \left\{w(a, b) - w(a, p_{a c}) \le w(p_{a c}, c) - w(b,c) : b \in B \right\}\right|, 
    \end{align*}  
    where all three counts are computed earlier. Next, the verifier checks that the length of $\ell_{a c}$ equals this count, and for every $b \in \ell_{a c}$, the verifier checks  $w(a, p_{a c}) + w(p_{a c}, c) < w(a, b) + w(b, c) < w(a, s_{a c}) + w(s_{a c}, c)$. If these checks pass, then $\ell_{a c}$ contains exactly the set of $b$ where $w(a, p_{a c}) + w(p_{a c}, c) < w(a, b) + w(b, c) < w(a, s_{a c}) + w(s_{a c}, c)$. By the definition of $p_{a c}$ and $s_{a c}$, it must be the case that $w(a, p_{a c}) + w(p_{a c}, c) < -w(a, c) < w(a, s_{a c}) + w(s_{a c}, c)$. Therefore, by reading the list $\ell_{a c}$ and count how many $b \in \ell_{a c}$ has $w(a, b) + w(b, c) = -w(a, c)$, the verifier can correctly compute the number of exact triangles through edge $(a, c)$. Overall, the cost of this step is the total length of $\ell_{a c}$, which is $O(n^{3-q})$. 
        
    The running time of the verifier is $\OO(n^{2+q} + n^{\frac{\omega+3+q}{2}} + n^{\omega+q}+n^{3-q})$. By setting $q = \frac{3-\omega}{3}$, it becomes $\OO(n^{\frac{6+\omega}{3}})$. 

    To adapt the algorithm to \AENegTriCountReal{}, for every $(a, c)$, we also count the number of $b$ such that $w(a, b) + w(b, c) < w(a, p_{ac}) + w(p_{ac}, c)$ via 
    \begin{align*}
       & \left| \left\{w(a, b) - w(a, p_{a c}) \le w(p_{a c}, c) - w(b,c) : b \in B \right\}\right| \\ 
       -& \left| \left\{w(a, b) - w(a, p_{a c}) = w(p_{a c}, c) - w(b,c) : b \in B \right\}\right|. 
    \end{align*} 
    If $-w(a, c) \in L_{ac}$, then $w(a, p_{ac}) + w(p_{ac}, c) = w(a, c)$, so the above count is exactly the number of negative triangles through $(a, c)$. Otherwise, every $b \in B$ with $w(a, b) + w(b, c) < w(a, p_{ac}) + w(p_{ac}, c) < -w(a, c)$ forms a negative triangle with $(a, c)$. All other negative triangles can be found via searching $\ell_{ac}$. 
\end{proof}

Before we present our co-nondeterministic algorithm for \AllThreeSUMCountReal{}, we first introduce the following lemma. 

\begin{lemma}
\label{lem:nondet-3sum-helper}
    Given an $n_1 \times n_2$ matrix $A$, an $n_3 \times n_2$ matrix $B$, and a  subset $X \subseteq [n_1] \times [n_3]$, counting the number of $(k, \ell)$ where $A_{i k} = B_{j \ell}$ for every $(i, j) \in X$ is in $$\NcoNTIME\left[\OO\left( |X|n_2^s +  M(n_1, n_2^{2-s}, n_3) \right)\right]$$ time for any $s \in [0, 1]$. Consequently, counting the number of $(k, \ell)$ where $A_{i k} \le B_{j \ell}$ for every $(i, j) \in X$ is also in $$\NcoNTIME\left[\OO\left( |X|n_2^s +  M(n_1, n_2^{2-s}, n_3) \right)\right]$$ time for any $s \in [0, 1]$.
\end{lemma}
\begin{proof}
    First, we can assume all entries of $A$ and $B$ are integers bounded by $O(n_1n_2 + n_2n_3)$, by replacing each entry by its rank. 
    
    We enumerate $b_1, b_2 \in  \{0\}\cup [\lfloor \log(n_2)\rfloor]$. If the occurrence of a number on the $i$-th row of $A$ has a $1$ in its binary representation on the bit corresponding to $2^{b_1}$, we keep one copy of this number on the $i$-th row of $A$; otherwise, we drop this number. Similarly, if the occurrence of a number on the $j$-th row of $B$ has a $1$ in its binary representation on the bit corresponding to $2^{b_2}$, we keep one copy of this number on the $j$-th row of $B$; otherwise, we drop this number. We then solve the original problem on these two modified matrices. Finally, we can sum up the results obtained on these modified matrices, weighted by $2^{b_1 + b_2}$. This way, we can assume all numbers on the $i$-th row of $A$ are distinct for any $i$, and all numbers on the $j$-th row of $B$ are distinct for any $j$. 

    The prover provides the following:
    \begin{itemize}
        \item A function $h: [O(n_1n_2 + n_2n_3)] \rightarrow [n_2]$. 

        The function is supposed to be that, for every $i \in [n_1]$, each value in $[n_2]$ matches at most $\tO(1)$ numbers in the multiset $\{h(A_{i k}): k \in [n_2]\}$. Similarly, for every $j \in [n_3]$, each value in $[n_2]$ matches at most $\tO(1)$ numbers in the multiset $\{h(B_{j k}): k \in [n_2]\}$. Such a function exists because a random function satisfies these constraints with high probability. 

        \item Two matrices $A', B'$. $A'$ is an $n_1 \times \tO(n_2)$ matrix, where the columns of it is indexed by $[n_2] \times [\tO(1)] \times [\tO(1)]$. For every $i \in [n_1], k \in [n_2], s_2 \in [\tO(1)]$, we set $A'_{i, (h(A_{i k}), s_1, s_2)}$ to $A_{i k}$, where $A_{i k}$ is the $s_1$-th number on row $i$ that get mapped to $h(A_{i k})$. Similarly, $B'$ is an $n_3 \times \tO(n_2)$ matrix, where the columns of it is indexed by $[n_2] \times [\tO(1)] \times [\tO(1)]$. For every $j \in [n_3], k \in [n_2], s_1 \in [\tO(1)]$, we set $B'_{j, (h(B_{j k}), s_1, s_2)}$ to $B_{j k}$, where $B_{jk}$ is the $s_2$-th number on row $j$ that get mapped to $h(B_{jk})$. All other entries of $A'$ and $B'$ are set to values distinct from any other values. 

        \item Use the protocol in Lemma~\ref{lem:nondet-equality} to create outputs for the purpose of computing the equality product between $A'$ and $B'^T$ with the output set $X$. 
    \end{itemize}

    The verifier does the following. First, it checks that $h, A', B'$ are valid. Then it runs the algorithm in Lemma~\ref{lem:nondet-equality} to compute, for every $(i, j) \in X$, the number of $(k, s_1, s_2)$ where $A'_{i, (k, s_1, s_2)} = B'_{j, (k, s_1, s_2)}$. Note that this is exactly the number of $(k, \ell)$ where $A_{i k} = B_{j \ell}$. The running time of the algorithm is same as that of Lemma~\ref{lem:nondet-equality}.

    If we instead want to count the number of $(k, \ell)$ where $A_{i k} \le B_{j \ell}$ for every $(i, j) \in X$, we can use the idea that reduces \Dominance{} to \Equality{} \cite{labib2019hamming,vnotes} to create $\OO(1)$ instances of the previous problem, so it only incurs an additional $\OO(1)$ factor. 
\end{proof}

Given an \AllThreeSUM{} instance on size-$n$ sets $A, B, C$, we can first sort $A$ and $B$ and then divide  them to consecutive sub-lists  $A_1, \ldots, A_{n/d}$ and $B_1, \ldots, B_{n/d}$ of size $d$. It is a well-known observation that, in order to determine whether some $c \in C$ is in a 3SUM solution, it suffices to search for $c$ in $A_i + B_j$ for $O(n/d)$ pairs of $(i, j) \in [n/d]^2$. For \AllThreeSUMCountReal{}, the observation is still true: For every $c \in C$, it suffices to count the number of $c$ in $A_i + B_j$ for $O(n/d)$ pairs of $(i, j) \in [n/d]^2$. We can thus determinitically reduce \AllThreeSUMCountReal{} to the following problem (see its non-counting version in, e.g., \cite{chan3sum}). 

\begin{problem}
\label{prob:3sum-variant}
We are given two real $(n/d) \times d$ matrices $A$ and $B$, and a set $C_{ij}$ of real numbers for every $(i, j)$. For every $c \in C_{ij}$, we are asked to count the number of $(k, \ell)$ such that $A_{i,k} + B_{j, \ell} = c$. Additionally $\sum_{i, j} |C_{ij}| = O(n^2/d)$. 
\end{problem}
    
\begin{theorem}
\label{thm:nondet-real-3sum-count}
    \AllThreeSUMCountReal{} for size-$n$ sets is in $\NcoNTIME[\OO(n^{\frac{3\omega+3}{\omega + 3}})]$.
\end{theorem}
\begin{proof}
    We first deterministically reduce \AllThreeSUMCountReal{} to Problem~\ref{prob:3sum-variant}. The proof then proceeds similarly to the proof of Theorem~\ref{thm:nondet-real-exacttri-count}. 

    The prover provides the following:
    \begin{itemize}
        \item A subset $R \subseteq [d] \times [d]$ of size $\OO(r)$. 

        For every $i, j$, let $L_{i j} = \{A_{ik} + B_{j \ell}: (k, \ell) \in R\}$. Let $(pk_{ijc}, p\ell_{ijc})$ be the index of the predecessor of $c$ in $L_{ij}$ (including $c$) and let $(sk_{ijc}, s\ell_{ijc})$ be the index of the successor of $c$ in $L_{ij}$ (excluding $c$). The set  $R$ is supposed to be that, for some $c \in C_{ij}$ and $c \not \in L_{ij}$, the number of $(k, \ell)$ with $A_{i,pk_{ijc}} + B_{j, p\ell_{ijc}} < A_{i k} + B_{j \ell} < A_{i, sk_{ijc}} + B_{j, s\ell_{ijc}}$ is $O(d^2/r)$. Such a set exists because a random $R$ satisfies these properties with high probability. 

        \item For every $(k_0, \ell_0) \in R$, it uses the protocol in Lemma~\ref{lem:nondet-3sum-helper} to create outputs for the purpose of counting $\left| \left\{ A_{ik} - A_{ik_0} = B_{j\ell_0} - B_{j\ell}\right\}\right|$ and $\left| \left\{ A_{ik} - A_{ik_0} \le B_{j\ell_0} - B_{j\ell}\right\}\right|$, where $X$ is the set of $(i, j)$ such that there exists $c \in C_{ij}$ with $(pk_{ijc}, p\ell_{ijc}) = (k_0, \ell_0)$ or $(sk_{ijc}, s\ell_{ijc}) = (k_0, \ell_0)$. 

        \item Finally, for every $(i,j)$ and $c \in C_{ij}$ with $c \not \in L_{ij}$, it provides a list $\ell\ell_{ijc} \subseteq [d] \times [d]$, which is supposed to contain all pairs $(k, \ell)$ with $A_{i,pk_{ijc}} + B_{j, p\ell_{ijc}} < A_{i, k} + B_{j, \ell} < A_{i, sk_{ijc}} + B_{j, s\ell_{ijc}}$. Note that $|\ell\ell_{ijc}| = O(d^2/r)$ by the choice of $R$. 
        
    \end{itemize}
\end{proof}

The verifier is almost identical to the verifier in Theorem~\ref{thm:nondet-real-exacttri-count}, so we omit its details for conciseness. The running time of the verifier is 
\begin{align*}
    \OO\left(\frac{d^2}{r} \cdot \sum_{i, j} |C_{ij}| + \sum_i \left(x_i d^s + M\left(\frac{n}{d}, d^{2-s}, \frac{n}{d}\right) \right)\right), 
\end{align*}
where $x_i$ is the size of $X$ in the $i$-th call of  Lemma~\ref{lem:nondet-3sum-helper}. By picking $d = n^{\frac{1}{3-s}}$, the running time can be simplified to 
\begin{align*}
    \OO\left(\frac{n^{\frac{2}{3-s}}}{r} \cdot \sum_{i, j} |C_{ij}| + \sum_i \left(x_i n^{\frac{s}{3-s}} + n^{\frac{2-s}{3-s}\cdot \omega} \right)\right). 
\end{align*}
We know $\sum_i x_i = O(\sum_{i, j} |C_{ij}|) = O(n^2/d) = O(n^{\frac{5-2s}{3-s}})$, and we call  Lemma~\ref{lem:nondet-3sum-helper} a total of $O(|R|) = O(r)$ times. Thus, we can further upper bound the running time by 
\begin{align*}
    \OO\left(\frac{n^{\frac{7-2s}{3-s}}}{r} + n^{\frac{5-s}{3-s}} + r \cdot n^{\frac{2-s}{3-s}\cdot \omega}\right).
\end{align*}
By setting $s = \frac{2\omega-3}{\omega}$ and $r = n^{\frac{3}{\omega+3}}$, 
the running time becomes $\OO(n^{\frac{3\omega+3}{\omega + 3}})$.

\begin{remark}\rm 
The above approach also leads to new consequences in a certain
unrealistic model of computation: an \emph{unrestricted Real RAM}, 
supporting standard arithmetic operation on real numbers with
unbounded precision, but without the floor function.
With the floor function, it is known that the model 
enables PSPACE-hard problems to be solved in polynomial time~\cite{Schonhage79}.
In a recent paper \cite{CVXstoc22}, it was noted that even without
the floor function, the model may 
still be unreasonably powerful; for example, there is a truly subcubic time 
algorithm for APSP for \emph{integer} input under this model. But the question
of whether there are similarly subcubic algorithms for APSP 
and other related problems for \emph{real} input was not answered.

    Our proof of Theorem~\ref{thm:nondet-real-exacttri-count} implies a truly subcubic randomized time algorithm for \AEExactTriCountReal{} in this unrestricted Real RAM without the floor function. This is because we use nondeterminism mainly to generate all the witnesses.  But using large 
numbers, we can do standard matrix product and have all the witnesses 
represented as a long bit vector.  When witnesses are needed for an 
output entry, we can generate them one by one using a 
most-significant-bit operation, which can be simulated by binary search. Small adaptation of the proof also shows a truly subcubic randomized time algorithm for the real-valued version of \MinPlusCount{}. Similarly, Theorem~\ref{thm:nondet-real-3sum-count} implies a truly subquadratic randomized time algorithm for \AllThreeSUMCountReal{} in the unrestricted Real RAM without the floor function.

\end{remark}

\subsection{Quantum Algorithms for Counting Problems}
\label{sec:quantum}
Our equivalences between counting and detection problems immediately imply faster quantum algorithms for counting problems. For instance, we obtain the following

\begin{corollary}
\label{cor:3sum-quantum}
There exists an $\OO(n^{2-\eps})$ time quantum algorithm for \ThreeSUMCount{}{} for some $\eps>0$.
\end{corollary}
\begin{proof}
It is known that \ThreeSUM{} can be solved in  $\OO(n)$ quantum time (see e.g. \cite{AmbainisL20}). Applying Theorem~\ref{thm:all-3sum-count} finishes the proof. 
\end{proof}

\subsection{Discussion}\label{sec:nondet:discuss}

{\sf \#CNF-SAT} asks to count the number of satisfying assignments to a CNF formula, and it is considered harder than {\sf CNF-SAT}: the counting version \#SETH of the Strong Exponential Time Hypothesis (SETH) \cite{ipz1,cip10} is considered even more believable than SETH (see~\cite{CurticapeanM16}). Williams' \cite{Williams05} reduction from {\sf CNF-SAT} to \OV{} preserves the counts, and thus is also a fine-grained reduction from {\sf \#CNF-SAT} to \OVCount{}. Similar to the situation for {\sf CNF-SAT}, there's no known fine-grained reduction from \OVCount{} to \OV{}.

Our techniques do not yet give equivalences between the decision and counting variants of {\sf CNF-SAT} and \OV{}. If such equivalences do not exist, this would indicate that \OV{} and {\sf CNF-SAT} are different from the other core problems in FGC. 

Such an indication was already observed by Carmosino et al.~\cite{carmosino2016nondeterministic} who studied the nondeterministic and co-nondeterministic complexity of fine-grained problems. They formulated NSETH that asserts that there is no $O((2-\eps)^n)$ time nondeterministic algorithm which can verify that a given CNF formula has no satisfying assignment.
They also exhibited nondeterministic algorithms for verifying the YES and NO solutions of \ExactTri, \APSP{} in truly subcubic time  and \ThreeSUM{} in truly subquadratic time, and concluded that if NSETH holds, then there can be no deterministic fine-grained reduction from {\sf CNF-SAT} or \OV{} to any of \ExactTri{}, \APSP{} or \ThreeSUM{}.

Because of our efficient nondeterministic algorithms for \ExactTriCount{}, \NegTriCount{} and \ThreeSUMCount,
we then get that under NSETH, there can be no deterministic fine-grained reductions from {\sf CNF-SAT} or \OV{} to \ExactTriCount{}, \APSPCount{} or \ThreeSUMCount. 

Recently, Akmal, Chen, Jin, Raj and Williams~\cite{Akmal0JR022} showed, among other things, that \ExactTriCount{} has an $\OO(n^2)$ time Merlin-Arthur protocol. Their protocol crucially uses polynomial identity testing, and hence it is not known how to derandomize it and make it nondeterministic. 
Our results yield
a truly subcubic nondeterministic protocol.

%% file: bsg.tex
In Section~\ref{sec:decomposition-zero-tri}, we have seen how our techniques may
replace some of the previous uses of the BSG Theorem in algorithmic applications.
In this section, we show how our techniques can actually prove
variants of the BSG Theorem itself.

We begin with a quick review of the BSG Theorem.  Many different 
versions of the theorem can be found in the literature, and
the following is one version that is easy to state:

\begin{theorem}\label{thm:BSG0}
{\bf (BSG Theorem)}\ \ 
Given subsets $A$, $B$, and $C$ of size $n$ of an abelian group,
and a parameter $s$,
if $|\{(a,b)\in A\times B:\ a+b\in C\}|\ge n^2/s$, 
then there exist subsets $A'\subseteq A$ and $B'\subseteq B$
both of size $\Omega(n/s)$, such that
\[ |A'+B'|\:=\: O(s^5n).
\]
\end{theorem}

The earliest version of the theorem, with super-exponential
factors in $s$, was obtained by
Balog and Szemer\'edi~\cite{BalogSze}, via the regularity lemma.
Gowers~\cite{Gowers01} was the first to obtain a version with polynomial
dependency on $s$.  The version stated above was proved by Balog~\cite{Balog07} and Sudakov, Szemer\'edi and Vu~\cite{SudakovSV94}.  Although the proof is not long and does not need advanced tools, it is clever and
not easy to think of; see~\cite{TaoVu06,Lovett17,Viola11} for various
different expositions.

Chan and Lewenstein~\cite{ChanLewenstein} gave algorithmic applications
using  the following variant which we will call the ``BSG Covering
Theorem'' (it was called the ``BSG Corollary'' in their paper).  Instead of extracting a single
pair of large subsets $(A',B')$, the goal is to construct a cover
by multiple pairs of subsets $(A^{(i)},B^{(i)})$: 

\begin{theorem}\label{thm:BSG0:cover}
{\bf (BSG Covering)}\ \ 
Given subsets $A$, $B$, and $C$ of size $n$ of an abelian group,
and a parameter $s$,
there exist a collection of $\ell=O(s)$ subsets $A^{(1)},\ldots,A^{(\ell)}\subseteq A$ and $B^{(1)},\ldots,B^{(\ell)}\subseteq B$, and a set $R$ of $O(n^2/s)$ pairs in $A\times B$, such that
\begin{enumerate}
\item[\rm(i)] $\{(a,b)\in A\times B: a+b\in C\}\ \subseteq\
R\,\cup\, \bigcup_\lam (A^{(\lam)}\times B^{(\lam)})$, and
\item[\rm(ii)] $|A^{(\lam)} + B^{(\lam)}| = O(s^5 n)$ for each $\lam$
(and so $\sum_\lam (|A^{(\lam)} + B^{(\lam)}|) = O(s^6 n)$).
\end{enumerate}
\end{theorem}

The BSG Covering Theorem is not implied by the BSG Theorem as stated,
but the known proofs by Balog~\cite{Balog07} and Sudakov et al.~\cite{SudakovSV94}
established an extension of the BSG Theorem  that involves an input graph, 
and repeated applications of this theorem indeed provide multiple pairs of subsets 
satisfying the stated properties.

\newcommand{\pop}{\mbox{\rm pop}}

\subsection{A New Simpler Version}

We will now show how our techniques, combined with some new extra ideas, can
derive a version of the BSG Covering Theorem where  the
$O(s^6n)$ bound is weakened to $\OO(s^2n^{3/2})$.  Although the new bound is superlinear in $n$, the lower polynomial dependency on $s$ actually compensates to yield
improved results in some algorithmic applications of the Covering Theorem.
In particular, $\OO(s^2n^{3/2})$ is better when $s\gg n^{1/8}$.
A key advantage of the new proof is its simplicity:  it constructs a cover directly, instead of repeatedly extracting subsets one at a time (thus, it avoids
the need to extend the BSG Theorem with an input graph, and thereby simplifies
the algorithm considerably).

We focus on the setting where each input set $A$ is of the
form $\{(i,a_i): i\in [n]\}$ for a sequence of integers or reals $a_1,\ldots,a_n$.
We call such a set an \emph{indexed set}.
This case turns out to be sufficient for applications involving
integer input (because known hashing-based techniques used in reductions
from \ThreeSUM{} to \ThreeSUMConv{}~(e.g. \cite{kopelowitz2016higher}) can map integer sets to indexed sets---\ThreeSUMConv{} is just \ThreeSUM{} for indexed sets).
Focusing on indexed sets makes the proof more intuitive, and also makes
the construction more efficient.  (In Appendix~\ref{app:bsg}, we present a variant of
the proof for general sets in an abelian group.)

Note that in the theorem stated below, we work with ``monochromatic'' difference
sets of the form $A-A$, instead of bichromatic sum sets $A+B$.  This
form is actually
more general, since given $A$ and $B$, we can reset $A$ to $A\cup (-B)$.
The reduction does not go the other way, since knowing that $|A^{(\lam)}+B^{(\lam)}|$
is small does not mean $|(A^{(\lam)}\cup (-B^{(\lam)})) - (A^{(\lam)}\cup (-B^{(\lam)}))|$ is small.  (The proofs for the known $O(s^6n)$ bound work only for the
bichromatic sum sets but not for monochromatic difference sets, whereas Gower's earlier
proof works for monochromatic difference sets.)

\begin{theorem}\label{thm:bsg:simple}
{\bf (Simpler BSG Covering)}\ \ 
Given indexed sets $A$ and $C$ of size $n$ and a parameter $s$,
there exist a collection of $\ell=\OO(s^3)$ subsets $A^{(1)},\ldots,A^{(\ell)}\subseteq A$, and a set $R$ of $\OO(n^2/s)$ pairs in $A\times A$, such that
\begin{enumerate}
\item[\rm(i)] $\{(a,b)\in A\times A: a-b\in C\}\ \subseteq\
R\,\cup\, \bigcup_\lam (A^{(\lam)}\times A^{(\lam)})$, and
\item[\rm(ii)] $\sum_\lam |A^{(\lam)} - A^{(\lam)}| = \OO(s^2 n^{3/2})$.
\end{enumerate}
The $A^{(\lam)}$'s and $R$ can be constructed in $\OO(n^2)$
Las Vegas randomized time.
\end{theorem}
\begin{proof}
Let $A=\{(i,a_i): i\in [n]\}$ and $C=\{(k,c_k): k\in [n]\}$.
As a preprocessing step, we
sort the multiset $\{a_{i+k}-a_i: i\in [n]\}$ for each $i$, 
in $\OO(n^2)$ total time.
Let $W_k=\{i\in [n]: a_{i+k} - a_i = c_k\}$.
\begin{itemize}
\item {\bf Few-witnesses case.}
For each $k$ with $|W_k|\le n/s$, add $\{((i+k,a_{i+k}),(i,a_i)): i\in W_k\}$ to $R$.
The number of pairs added to $R$ is $O(n\cdot n/s)$.
The running time of this step is also $\OO(n\cdot n/s)$.

\item {\bf Many-witnesses case.}
Pick a random subset $H\subseteq [\pm n]$ of size $c_0s\log n$ for a sufficiently large constant $c_0$.  
Let $L^{(h)}$ be the multiset $\{a_{i+h}-a_i: i\in [n]\}$.
Let $F^{(h)}$ be the elements of frequency more than $n/r$ in $L^{(h)}$.
Note that $|F^{(h)}|\le r$.
These can be computed in $\OO(sn)$ time.

\begin{itemize}
\item {\bf Low-frequency case.}
For each $h\in H$ and $i\in [n]$,
if $a_{i+h}-a_i\not\in F^{(h)}$,
we examine each of the at most $n/r$ indices $j$ with
$a_{i+h}-a_i=a_{j+h}-a_j$ and add $((j,a_j),(i,a_i))$ to $R$.
The number of pairs added to $R$ is $\OO(sn\cdot n/r)=\OO(n^2/s)$
by choosing $r:=s^2$.
The running time of this step is also bounded by $\OO(n^2/s)$.
\item {\bf High-frequency case.}
For each $h\in H$ and $f\in F^{(h)}$,
add the following subset to the collection:
\[ A^{(h,f)} = \{(i,a_i)\in A: a_{i+h}-a_i = f\}.
\]
The number of subsets is $\OO(s\cdot r)=\OO(s^3)$.
The running time of this step is \[O\left(\sum_{h\in H,\ f\in F^{(h)}}|A^{(h,f)}|\right)=\OO(sn).\]
\end{itemize}
\end{itemize}

\emph{Correctness.}
To verify (i), consider a pair $((i+k,a_{i+k}),(i,a_i))\in A\times A$ with 
$a_{i+k}-a_i=c_k$.
If $|W_k|\le n/s$, then $((i+k,a_{i+k}),(i,a_i))\in R$ due to
the ``few-witnesses'' case.
So assume $|W_k|>n/s$.
Then $H$ hits $W_k-i$ w.h.p., so
there exists $h\in H$ with $i+h\in W_k$, i.e.,
$a_{i+h+k}-a_{i+h}=c_k$.  Since $a_{i+k}-a_i=c_k$,
we have $a_{i+h+k}-a_{i+k}=a_{i+h}-a_i$ (by Fredman's trick).
Let $f=a_{i+h}-a_i$.
If $f\not\in F^{(h)}$, then $((i+k,a_{i+k}),(i,a_i))\in R$  due to
the ``low-frequency'' case.
If $f\in F^{(h)}$, then $((i+k,a_{i+k}),(i,a_i))\in A^{(h,f)}\times A^{(h,f)}$ due to the ``high-frequency'' case.

Thus far, the proof ideas are similar to what we have seen before.  However, to verify (ii), we will propose a new  probabilistic argument.
Consider a fixed $h$.
We want to bound the sum $S^{(h)}:=\sum_{f\in F^{(h)}} |A^{(h,f)}-A^{(h,f)}|$.
This is equivalent to bounding the number of triples $(k,c,f)$ such that
$f\in F^{(h)}$ and $\exists i$ with $a_{i+k}-a_i=c$ and
$a_{i+h}-a_i=a_{i+k+h}-a_{i+k}=f$.  Note that we also have
$a_{i+k+h}-a_{i+h}=c$ (by Fredman's trick).
Let $Y_{k,c} = \{i: a_{i+k}-a_i=c\}$
and $q$ be a parameter.
Then $S^{(h)}$ is upper-bounded by the number of triples $(k,c,f)$ satisfying
\begin{enumerate}
\item $|Y_{k,c}|>q$ and $f\in F^{(h)}$, or
\item $|Y_{k,c}|\le q$ and $\exists i$ with $i\in Y_{k,c}$ and $i+h\in Y_{k,c}$ and $f=a_{i+h}-a_i$.
\end{enumerate}
Since $\sum_{k,c} |Y_{k,c}|=O(n^2)$,
the number of triples of type~1 is $O((n^2/q)\cdot r)=\OO(s^2n^2/q)$.
On the other hand, the expected number of triples of type~2 for a random
$h\in[\pm n]$ is 
\[
O\left(\sum_{k,c:\: |Y_{k,c}|\le q} |Y_{k,c}|\cdot |Y_{k,c}|/n\right)
\ =\ O(n^2\cdot q/n)\ =\ O(nq).
\]
Hence, $\Ex[S^{(h)}] = \OO(s^2n^2/q + nq) = \OO(sn^{3/2})$
by choosing $q=s\sqrt{n}$.

The expected total sum $\sum_{h\in H}\sum_{f\in F^{(h)}} |A^{(h,f)}-A^{(h,f)}|$ is $\OO(s^2n^{3/2})$.
Thus, the probability that the desired bounds are not met is less than an arbitrarily small constant (by Markov's inequality and a union bound).

Note that the algorithm can be converted to Las Vegas, because we can
verify correctness of the construction by examining all $O(n^2)$ pairs
and computing all difference sets $A^{(h,f)}-A^{(h,f)}$ using
known output-sensitive algorithms~\cite{ColeHariharanSTOC02,ChanLewenstein,BringmannFN} in total time $\OO(n^2 + s^2n^{3/2})=\OO(n^2)$ (we may assume $s<n^{1/4}$, for otherwise the theorem is trivial).
\end{proof}

\subsection{Application: Improved 3SUM in Preprocessed Universes}
\label{sec:3sum:prep:rand}

As one immediate application, we can solve  \ThreeSUM{} with preprocessed
universe, improving Chan and Lewenstein's previous solution which required
$\OO(n^{13/7})$ query time~\cite{ChanLewenstein}, and also improving Corollary~\ref{cor:3sum:prep0} regardless of the value of $\omega$: the query algorithm does not use
fast matrix multiplication but uses FFT instead, though randomization is
now needed in the  preprocessing algorithm.

\begin{corollary}
We can preprocess sets $A$, $B$, and $C$ of $n$ integers
in $\OO(n^2)$ Las Vegas randomized time, so that given any subsets $A'\subseteq A$,
$B'\subseteq B$, and $C'\subseteq C$, we can solve \AllThreeSUM{} on $(A',B',C')$ in $\OO(n^{11/6})$ time.
\end{corollary}
\begin{proof}
Kopelowitz, Pettie and Porat~\cite{kopelowitz2016higher} gave a simple randomized reduction from \ThreeSUM{} to $O(\log n)$ instances of \ThreeSUMConv{} via hashing.
The same approach works in the preprocessed universe setting, and
transform the input into $O(\log n)$ instances where $A$, $B$, and $C$
are indexed sets.

During preprocessing, we apply Theorem~\ref{thm:bsg:simple}
to $A\cup (-B)$, producing subsets $A^{(\lam)}$ and a set $R$ of pairs.

During a query with given subsets $A'\subseteq A$,
$B'\subseteq B$, and $C'\subseteq C$,
we first examine each pair $(a,-b)\in R$ and check whether $a\in A'$,
$b\in B'$, and $a+b\in C'$.  This takes $\OO(n^2/s)$ time.

Next, for each $\lam$, we compute $(A^{(\lam)}\cap A') + ((-A^{(\lam)})\cap B')$
by known FFT-based algorithms~\cite{ColeHariharanSTOC02,ChanLewenstein,BringmannFN}; the running time is
near-linear in the output size, which is bounded by $|A^{(\lam)}-A^{(\lam)}|$.
For each output value $c$, we check whether $c\in C'$.

The total query time is $\OO(n^2/s + s^2n^{3/2})$.
Choosing $s=n^{1/6}$ yields the theorem.
\end{proof}

We remark that the same $\OO(n^{11/6})$ bound holds for a slightly more general case when
$C'$ is an arbitrary set of $n$ integers, i.e., the preprocessing does not need the set $C$.
This is because in the proof of Theorem~\ref{thm:bsg:simple},
the $\OO(n^2)$-time preprocessing step is independent of $C$, and
the rest of the construction takes $\OO(n^2/s +sn)$ time.
In contrast, Chan and Lewenstein's paper obtained a weaker 
$\OO(n^{19/10})$ bound for the same case without $C$.

\subsection{Reinterpreting Gower's Version}

In this section, we show that Gower's proof~\cite{Gowers01} 
(see also a related recent proof by Schoen~\cite{Schoen}) can also be modified to
construct a cover directly.\footnote{
There are multiple different exposition of Gower's and subsequent proofs of the BSG Theorem in the literature. For example, 
some presentations \cite{Balog07,SudakovSV94,TaoVu06,Viola11} cleanly separate the algebraic from the combinatorial components, by
reducing the problem to some combinatorial lemma about graphs (counting paths of length 2 or 3 or 4).
But these versions of the proof do not achieve our goal of computing a cover directly and efficiently.  Our reinterpretation is nontrivial and requires examining Gower's proof from the right perspective.
}  
Gower's proof requires more clever arguments; our new presentation highlights the similarities
and differences with the proof of Theorem~\ref{thm:bsg:simple}.
This variant of the proof will be needed in a later
application  in Section~\ref{sec:min-equal-conv} to conditional lower bounds---so ideas from additive combinatorics
will be useful after all!

In the following, we let $\pop_A(x)$ (the \emph{popularity of $x$}) denote the number of pairs $(a,b)\in A\times A$ with
$x=a-b$; in other words, $\pop_A(x)=|\{a\in A: a-x\in A\}|$.

\newcommand{\AAA}{\widetilde{A}}

\begin{theorem}
\label{thm:bsg:gower}
Given indexed sets $A$ and $C$ of size $n$ and a parameter $s$,
there exist a collection of $\ell=\OO(s^3)$ subsets $A^{(1)},\ldots,A^{(\ell)}\subseteq A$, and a set $R$ of $\OO(n^2/s)$ pairs in $A\times A$, such that
\begin{enumerate}
\item[\rm(i)] $\{(a,b)\in A\times A: a-b\in C\}\ \subseteq\
R\,\cup\, \bigcup_\lam (A^{(\lam)}\times A^{(\lam)})$, and
\item[\rm(ii)] $|A^{(\lam)} - A^{(\lam)}| = \OO(s^6 n)$ for each $\lam$
(and so $\sum_\lam |A^{(\lam)} - A^{(\lam)}| = \OO(s^9 n)$).
\end{enumerate}
The $A^{(\lam)}$'s and $R$ can be constructed in $\OO(n^2)$
Las Vegas randomized time.
\end{theorem}
\begin{proof}
We follow the proof in Theorem~\ref{thm:bsg:simple} but
modify the handling of the ``high-frequency'' case.
Let $t$ be a parameter.  Define 
\begin{eqnarray*}
G^{(h,f)}&=&\{ (a,b)\in A^{(h,f)}\times A^{(h,f)}: \pop_A(a-b)\le n/t\}\\
Z^{(h,f)}&=&\{a\in A^{(h,f)}: \deg_{G^{(h,f)}}(a) > |A^{(h,f)}|/4\}.
\end{eqnarray*}
Here, $\deg_{G^{(h,f)}}(a)$ refers to the degree of $a$ in $G^{(h,f)}$ when
viewed as a graph.
For each $h\in H$ and $f\in F^{(h)}$,
add $Z^{(h,f)}\times A^{(h,f)}$ and $A^{(h,f)}\times Z^{(h,f)}$ to $R$.
For each $h\in H$ and $f\in F^{(h)}$,
instead of adding the subset $A^{(h,f)}$, we add the subset
$\AAA^{(h,f)}:=A^{(h,f)}\setminus Z^{(h,f)}$ to the collection.

\smallskip
\emph{Correctness.}
To analyze this modified construction, first observe that every pair
previously covered by $A^{(h,f)}\times A^{(h,f)}$ is now covered by
$\AAA^{(h,f)}\times \AAA^{(h,f)}$ or by the extra pairs added to $R$ (i.e., $Z^{(h,f)}\times A^{(h,f)}$ or $A^{(h,f)}\times Z^{(h,f)}$).

We bound
the expected number of extra pairs added to $R$.
Consider a fixed $h$.
Note that $|Z^{(h,f)}|\le O\left(\frac{|G^{(h,f)}|}{|A^{(h,f)}|/4}\right)$, and so $\sum_f |Z^{(h,f)}||A^{(h,f)}| = O\left(\sum_f |G^{(h,f)}|\right)$.
The sum $\sum_f |G^{(h,f)}|$ is bounded by the number of triples $(i,j,f)$ such that
$a_{i+h}-a_i=a_{j+h}-a_j=f$ and $\pop_A((i-j,a_i-a_j))\le n/t$.
This is bounded by the number of pairs $(i,j)$ such that 
$\pop_A((i-j,a_i-a_j))\le n/t$ and
$a_i-a_j=a_{i+h}-a_{j+h}$ (by Fredman's trick).
For a fixed $(i,j)$ with $\pop_A((i-j,a_i-a_j))\le n/t$,
the number of $h$'s with $a_i-a_j=a_{i+h}-a_{j+h}$ is at most $n/t$,
and so the probability that $a_i-a_j=a_{i+h}-a_{j+h}$ for a random $h\in [\pm n]$ is $O(1/t)$.
It follows that 
\[ \Ex_h\left[ \sum_f |Z^{(h,f)}||A^{(h,f)}| \right] = O(n^2\cdot 1/t).
\]
Consequently, the expected number of extra pairs added to $R$
is $\OO(s\cdot n^2/t)=\OO(n^2/s)$ by setting $t := s^2$.

Finally, we consider a fixed $h$ and fixed $f\in F^{(h)}$ and provide an upper bound on
$|\AAA^{(h,f)}-\AAA^{(h,f)}|$.  For each $c\in \AAA^{(h,f)}-\AAA^{(h,f)}$,
pick a lexicographically smallest $(a,b)\in \AAA^{(h,f)}\times\AAA^{(h,f)}$
with $c=a-b$.
Consider all $y\in A^{(h,f)}$ with $(a,y),(b,y)\not\in G^{(h,f)}$;
the number of such $y$'s is at least $|A^{(h,f)}|-|A^{(h,f)}|/4-|A^{(h,f)}|/4=\Omega(|A^{(h,f)}|)=\Omega(n/r)$ (since $|A^{(h,f)}|>n/r$ for $f\in F^{(h)}$).  For each such $y$, examine
each $(a',a'')\in A\times A$ with $a-y=a'-a''$
and each $(b',b'')\in A\times A$ with $b-y=b'-b''$,
and mark the quadruple $(a',a'',b',b'')$.
Since $(a,y),(b,y)\not\in G^{(h,f)}$,
there are at least $n/t$ choices of $(a',a'')$ and
at least $n/t$ choices of $(b',b'')$ for each such $y$.
Letting $Q$ be the number of quadruples marked, we obtain 
\[ Q\ =\ \Omega\left(|\AAA^{(h,f)}-\AAA^{(h,f)}|\cdot (n/r)\cdot (n/t)^2\right).
\]
On the other hand, each quadruple $(a',a'',b',b'')$, is marked once,
since it uniquely determines the element $c=(a'-a'')-(b'-b'')$, from
which $(a,b)$ is uniquely determined and $y=a-(a'-a'')$ is uniquely
determined.  Thus, $Q=O(n^4)$.  We conclude that
\[ |\AAA^{(h,f)}-\AAA^{(h,f)}|\ =\ O\left(\frac{n^4}{(n/r)\cdot (n/t)^2}\right)\ =\ O(rt^2n)\ =\ \OO(s^6n).
\]

\smallskip
\emph{Construction time.}
One could naively construct the sets $Z^{(h,f)}$ and
$\AAA^{(h,f)}$ in $O(\sum_{h,f} |A^{(h,f)}|^2)$ time,
but a faster way is to use random sampling.
Given any value $c$, we can approximate $\pop_A(c)$
with additive error $\delta n/t$ w.h.p.\ by
taking a random subset $A'\subseteq A$
of $O((1/\delta^2)t\log n)=\OO(t)$ elements and computing
$|\{a\in A': a-c\in A\}|\cdot |A|/|A'|$ (by a standard Chernoff bound).
We will not construct $G^{(h,f)}$ explicitly.
Instead, given any $(a,b)$, we can test for membership in $G^{(h,f)}$
in $\OO(t)$ time.
Furthermore, given $a$, we can approximate $\deg_{G^{(h,f)}}(a)$
with additive error $\delta |A^{(h,f)}|$ w.h.p.\ by
taking a random subset $A''\subseteq A^{(h,f)}$
of $O((1/\delta^2)\log n)=\OO(1)$ elements and computing
$|\{y\in A'': (a,y)\in G^{(h,f)}\}|\cdot |A^{(h,f)}|/|A''|$.
This way, $Z^{(h,f)}$ (and thus $\AAA^{(h,f)}$) can be generated in $\OO(t|A^{(h,f)}|)$ time for each $h$ and $f$.  The total time bound is
$\OO(t\sum_{h,f}|A^{(h,f)}|)=\OO(tsn)=\OO(s^3n)$, which is dominated by
other costs (we may assume that $s<n^{1/6}$, for otherwise the theorem is trivial).
Due to these approximations, our earlier analysis needs small adjustments
in the constant factors, but is otherwise the same.

As before, the algorithm can be converted to Las Vegas in $\OO(n^2)$ additional time.
\end{proof}

\newcommand{\sss}{\hat{s}}

One important advantage of the above proof is that
the running time is actually \emph{subquadratic}, excluding the $\OO(n^2)$-time
preprocessing step, which is needed only in the ``few-witnesses'' case (ignoring the conversion to Las Vegas).
In particular, we immediately obtain subquadratic running time for the following variant of the theorem, which requires only the ``many-witnesses''
case (where we reset $r$ and $t$ to $s\sss$ instead of $s^2$).
This variant will be useful later.

\begin{theorem}
\label{thm:bsg:gower:fast}
Given an indexed set $A$ of size $n$ and parameters $s$ and $\sss$,
there exist a collection of $\ell=\OO(s^2\sss)$ subsets $A^{(1)},\ldots,A^{(\ell)}\subseteq A$, and a set $R$ of $\OO(n^2/\sss)$ pairs in $A\times A$, such that
\begin{enumerate}
\item[\rm(i)] $\{(a,b)\in A\times A: \pop_A(a-b) > n/s\}\ \subseteq\ 
R\,\cup\, \bigcup_\lam (A^{(\lam)}\times A^{(\lam)})$, and
\item[\rm(ii)] $|A^{(\lam)} - A^{(\lam)}| = \OO(s^3\sss^3 n)$ for each $\lam$
(and so $\sum_\lam |A^{(\lam)} - A^{(\lam)}| = \OO(s^5\sss^4 n)$).
\end{enumerate}
The $A^{(\lam)}$'s and $R$ can be constructed in $\OO(n^2/\sss + s^2\sss n)$
Monte Carlo randomized time.
\end{theorem}

Although we are able to reinterpret Gower's proof, we are unable to modify the proof by
Balog~\cite{Balog07} or Sudakov et al.~\cite{SudakovSV94} to achieve similar subquadratic
construction time.

%% file: min_eq_conv.tex
In this section, we prove conditional lower bounds for the \MinEqualityConv{}
problem under \StrongAPSP{}, the \uAPSPH{}, and the
\StrongConv{}.  
The lower bound under the \StrongAPSP{} or \uAPSPH{}
follows just by combining Corollary~\ref{cor:strong-intapsp-imply} %
with the known reduction from \uAPSP{} to \MinWitnessEq{} \cite{CVXicalp21} (and noticing that \MinWitnessEq{} is easier than \MinEqualityProd{}), and then using known ideas for reducing matrix product problems to convolution problems (more specifically, the unpublished reduction from \BMM{} to pattern-to-text Hamming distances, attributed to Indyk -- see e.g.\  \cite{GawrychowskiU18}). 
The lower bound under the \StrongConv{} is more delicate: %
interestingly, we will combine ideas that we have developed
for conditional lower bounds for intermediate matrix product problems, with
one of our new versions of the BSG Theorem from the previous section.

\begin{theorem}
\label{thm:min-equal-conv-under-product}
Under the \StrongAPSP{}, \MinEqualityConv{} for length $n$ arrays  requires $n^{1+1/6-o(1)}$ time.
\end{theorem}
\begin{proof}
First, by Corollary~\ref{cor:strong-intapsp-imply}, \uAPSP{} requires $n^{7/3-o(1)}$ time under the Strong Integer-APSP Hypothesis. Zwick's algorithm \cite{zwickbridge} can be seen as a reduction from \uAPSP{} to $\OO(1)$ instances of \MinPlus{} between $n \times n^\alpha$ matrices and $n^\alpha \times n$ matrices with weights in $[n^{1-\alpha}]$ for various values of $\alpha \in [0, 1]$. As shown in \cite{CVXicalp21}, each  instance of \MinPlus{} in this form can be reduced to an instance of \MinWitnessEq{} for $O(n) \times O(n)$ matrices (they only stated the reduction for $\alpha = \rho$ for a particular value of $\rho$,  but their proof works for any $\alpha \in [0, 1]$). Therefore, \MinWitnessEq{} requires $n^{7/3-o(1)}$ time under the Strong Integer-APSP Hypothesis. 

We can easily reduce a \MinWitnessEq{} instance to a \MinEqualityProd{} instance. Suppose the input of a \MinWitnessEq{} instance is $A, B$. Assume all entries are in $[2n]$ without loss of generality. We can create two $n \times n$ matrices $A'$ and $B'$, where $A'_{ik} = A_{ik}+2nk$ and $B'_{kj} = B_{kj} + 2nk$, then the Min-Witness Equality product between $A$ and $B$ can be computed in $\OO(n^2)$ time given the Min-Equality product between $A'$ and $B'$. 

Finally, we reduce \MinEqualityProd{} to \MinEqualityConv{}, following the strategy of the unpublished reduction from Boolean matrix multiplication to  pattern-to-text Hamming distances, attributed to Indyk, see e.g. \cite{GawrychowskiU18}. 

Let $A$ and $B$ be the inputs of \MinEqualityProd{}. W.l.o.g., we can assume all entries of $A$ and $B$ are integers in $[2n^2]$. 
We first create two length $2n^2$ arrays $a$ and $b$, where initially all entries of $a$ and $b$ are $\infty$. For every $(i, k) \in [n] \times [n]$, we set $a_{(n+1)(i-1) + k}$ to $n A_{ik} + k - 1$; for every $(k, j) \in [n] \times [n]$, we set $b_{jn-k}$ to $n B_{kj} + k - 1$. 

Suppose the Min-Equality product between $A$ and $B$ is $C$ and the Min-Equality convolution between $a$ and $b$ is $c$, we will show that $C_{ij} = \lfloor c_{(n+1)(i-1)+jn} / n \rfloor$ (and  $C_{ij} = \infty$ if $c_{(n+1)(i-1)+jn}$ is  $\infty$), which will complete the reduction. 

To show the equality, first notice that 
\begin{equation}
\label{eqn:min-equal-conv}
\begin{split}
    &\min \left\{a_{(n+1)(i-1) + k} : k \in [n] \wedge a_{(n+1)(i-1) + k} = b_{jn-k} \right\} \\
    =& \min \left\{nA_{ik} + k -1 : k \in [n] \wedge A_{ik}=B_{kj} \right\}
\end{split}
\end{equation}
contributes to the minimization of $c_{(n+1)(i-1)+jn}$. Also, no other terms less than $\infty$ can contribute: suppose there exists some $x$ such that $a_x = b_{y} < \infty$ and $x+y = (n+1)(i-1)+jn$, then $a_x \bmod{n}$ must match $b_y \bmod{n}$. Thus, there must exist $i', j', k' \in [n]$ such that $x = (n+1)(i'-1)+k'$ and $y = j'n - k'$, so $(n+1)(i'-1) + j'n = (n+1)(i-1) + jn$. Then it must be the case that $i=i'$ and $j=j'$, so $a_x$ corresponds to one of the terms in Equation~(\ref{eqn:min-equal-conv}).
\end{proof}
The reduction in the proof of Theorem~\ref{thm:min-equal-conv-under-product} from \uAPSP{} to \MinEqualityConv{} also easily imply the following:

\begin{theorem}
\label{thm:min-equal-conv-under-uAPSP}
Under the \uAPSPH{}, \MinEqualityConv{} for length $n$ arrays requires $n^{1+\rho/2-o(1)}$ time, where $\rho$ is the constant satisfying $\omega(1, \rho, 1) = 1 + 2\rho$, or $n^{1.25-o(1)}$ time if $\omega = 2$. 
\end{theorem}

We finally show the lower bound of \MinEqualityConv{} under the \StrongConv{}. 

\begin{theorem}
\label{thm:min-equal-conv-under-conv}
Under the \StrongConv{}, \MinEqualityConv{} for length $n$ arrays  requires $n^{1+1/11-o(1)}$ time.
\end{theorem}
\begin{proof}

We will show that if \MinEqualityConv{} for length $n$ arrays has an $\OO(n^{1+\delta})$ time algorithm for some $\delta \ge 0$, then \MinPlusConv{} for length $n$ arrays of entries that are bounded by $O(n)$  has an $\OO(n^{2 - \frac{1-11\delta}{21}})$ time randomized algorithm.

Let $A$ and $B$ be the input arrays of a \MinPlusConv{} instance. Let $t$ be a parameter to be fixed later, and let $g = \lceil n/t\rceil$. Similar to the proof of Theorem~\ref{thm:main}, we can assume $(A_i \bmod{g}) < g/2$ and $(B_i \bmod{g}) < g/2$
for each $i \in [n]$. Also, for each $i$, we can write $A_i$ as $A'_i g + A''_i$, for $0 \le A'_i \le t$ and $0 \le A''_i < g/2$. Similarly, we can write $B_i$ as $B'_i g + B''_i$. 

We first compute the Min-Plus convolution $C'$ of $A'$ and $B'$ in $\OO(tn)$ time. Let $W_k = \{i \in [n]: k - i \in [n] \wedge C'_k = A'_i + B'_{k - i}\}$. Suppose we can compute $C''$, which is defined as $C''_k = \min_{i \in W_k} (A''_i + B''_{k-i})$, then we can compute the Min-Plus convolution $C$ of $A$ and $B$ as $C_k = C'_k g + C''_k$. 

We then compute $C''_k$ by two methods depending on whether $|W_k|$ is greater than $n/s$ or not, for some parameter $s$ to be determined.

\begin{claim}
\label{cl:min-equal-conv-heavy}
For any parameter $\sss$, 
we can compute $C''_k$ for every $k$ where $|W_k| > n / s$ in $\OO(n^2 / \sss + s^5 \sss^4 n^2 / t)$ time. 
\end{claim}
\begin{proof}
We create the following indexed set $\mathcal{A}$ of size $O(n)$:
$$\left\{(i, A'_i): i \in [n]\right\} \cup \left\{(i+n, \infty) : i \in [n]\right\} \cup \left\{(3n + 1 - i, -B'_i): i \in [n] \right\}.$$
Note that for any $k \in [n]$, $|W_k| = \pop_{\mathcal{A}}((k - 3n - 1, C'_k))$. Then we apply Theorem~\ref{thm:bsg:gower:fast} with the index set $\mathcal{A}$ and parameters $s, \sss$ to find a collection of $\ell = \OO(s^2 \sss)$ subsets $\mathcal{A}^{(1)}, \ldots, \mathcal{A}^{(\ell)}$, and a set $R$ of $\OO(n^2/\sss)$ pairs in $\mathcal{A} \times \mathcal{A}$ in $\OO(n^2/\sss + s^2 \sss n )$ randomized time. Furthermore, Theorem~\ref{thm:bsg:gower:fast} guarantees that
\begin{enumerate}
    \item[\rm(i)] $\{(a,b)\in \mathcal{A}\times \mathcal{A}: \pop_\mathcal{A}(a-b) > n/s\}\ \subseteq\ 
R\,\cup\, \bigcup_\lambda (\mathcal{A}^{(\lambda)}\times \mathcal{A}^{(\lambda)})$. This further means that, for every $k$ where $|W_k| > n/s$, $$\left\{ \left((i, A'_i), (3n+1-(k-i), -B'_{k-i}) \right): i \in W_k\right\} \subseteq\ 
R\,\cup\, \bigcup_\lambda (\mathcal{A}^{(\lambda)}\times \mathcal{A}^{(\lambda)}).$$
\item[\rm(ii)] $\sum_\lambda |\mathcal{A}^{(\lambda)} - \mathcal{A}^{(\lambda)}| = \OO(s^5 \sss^4 n)$. 
\end{enumerate}
Then we first enumerate $((i_1, v_1), (i_2, v_2)) \in R$. If this pair corresponds to some $A'_i$ and $B'_j$ (i.e., this pair has $i_1 = i, i_2 = 3n+1-j$), we use $A''_i + B''_j$ to update $C''_{i+j}$ if $A'_i + B'_j = C'_{i+j}$. This takes $\OO(|R|) = \OO(n^2 / \sss)$ time. 

For each $\lambda \in [\ell]$, we consider the possible witnesses in $\mathcal{A}^{(\lambda)} \times \mathcal{A}^{(\lambda)}$. We prepare a map $f$ from $\mathcal{A}^{(\lambda)}$ to $[g] \cup \{\infty\}$ as follows: if $a \in \mathcal{A}^{(\lambda)}$ corresponds to some $A'_i$ (i.e., $a = (i, A'_i)$), we set $f(a) = A''_i$; if $a$ corresponds to some $B'_j$ (i.e., $a = (3n+1-j, -B'_j)$), we set $f(a) = B''_j$; otherwise, we set $f(a) = \infty$. Then we compute the following Min-Plus ``convolution'' $\mathcal{C}^{(\lambda)}$:
$$\mathcal{C}^{(\lambda)}_c = \min_{\substack{(a, b) \in \mathcal{A}^{(\lambda)} \times \mathcal{A}^{(\lambda)} \\ a - b = c}} (f(a) + f(b)).$$
By known techniques for solving Min-Plus convolution with small integer weights \cite{ALONGM1997}, we can instead solve a normal convolution with weights bounded by $2^{\OO(g)}$. More specifically, let $h(a) = M^{f(a)}$ if $f(a) \ne \infty$ and $f(a) = 0$ otherwise, where $M = |\mathcal{A}^{(\lambda)}|+1$. Then it suffices to compute the following for every $c$:
$$\sum_{\substack{(a, b) \in \mathcal{A}^{(\lambda)} \times \mathcal{A}^{(\lambda)} \\ a - b = c}} h(a) \cdot h(b).$$
By known output-sensitive algorithms~\cite{ColeHariharanSTOC02,ChanLewenstein, BringmannFN}, it takes $\OO(|\mathcal{A}^{(\lambda)} - \mathcal{A}^{(\lambda)}|)$ arithmetic operations to compute the above convolution, and each  arithmetic operation takes $\OO(g)$ time. Thus, it takes $\OO(g|\mathcal{A}^{(\lambda)} - \mathcal{A}^{(\lambda)}|)$ time to compute $\mathcal{C}^{(\lambda)}$. 

After we compute $\mathcal{C}^{(\lambda)}$ for every $\lambda$, we use the value of $\mathcal{C}^{(\lambda)}_{(k-3n-1, C'_k)}$ to update $C''_k$ for every $k, \lambda$. 

Overall, the running time is $\OO(n^2/\sss + s^2 \sss n + g \sum_i |\mathcal{A}^{(\lambda)} - \mathcal{A}^{(\lambda)}|) = \OO(n^2/\sss + s^5 \sss^4 n g) = \OO(n^2/\sss + s^5 \sss^4 n^2 / t)$. 
\end{proof}

Next we show the following algorithm for the rest values of $k$ where $|W_k|$ is small. Recall that we assumed  \MinEqualityConv{} for length $n$ arrays has an $\OO(n^{1+\delta})$ time algorithm. 
\begin{claim}
\label{cl:min-equal-conv-light}
We can compute $C''_k$ for every $k$ where $|W_k| \le n / s$ in $\OO(\frac{n}{s} \cdot n^{1+\delta})$ time as long as $t\sqrt{s} = O(\sqrt{n})$.
\end{claim}
\begin{proof}
Let $I \subseteq [n]$ be a random subset of indices for which each index is kept in $I$ independently with probability $\frac{\sqrt{s}}{\sqrt{n}}$. Similarly let $J \subseteq [n]$ be such a random subset as well. With high probability, $|I|, |J| = O(\sqrt{sn})$.

In the sparse Min-Plus convolution between the two sparse arrays $A'_I$ and $B'_J$, we need to compute a length $n$ array $D$ where $D_k = \min_{i \in I, k - i \in J} (A'_i + B'_j)$. 

For every $k \in [n]$  and $i \in W_k$, the probability that $i \in I$ and $k - i \in J$ is $\frac{s}{n}$. Thus, for any particular $i \in W_k$, the probability that $i$ is the unique witness for $D_k$ in the sparse Min-Plus convolution between $A_I'$ and $B_J'$ is $\frac{s}{n} \cdot (1 - \frac{s}{n})^{|W_k| - 1}$, which is $\Theta(\frac{s}{n})$ if $|W_k| \le n/s$. Thus, if we keep sampling $I$ and $J$ for $\tO(\frac{n}{s})$ times, all indices in $W_k$ will be the unique witness for $D_k$ in at least one  time with high probability, as in  standard sampling techniques (see e.g. \cite{AlonGMN92, seidel1995}).

Suppose $i$ is indeed the unique witness for $D_k$, then we can find $i$ by repeatedly computing some instances of sparse Min-Plus convolutions. More specifically, in the $p$-th round, let $I^{(p)}$ be $I$ but only keeping the indices  whose $p$-th bit in the binary representation is $1$. Say the sparse Min-Plus convolution between $A'_{I^{(p)}}$ and $B'_J$ is $D^{(p)}$. Then if $D^{(p)}_k = D_k$, then we know the $p$-th bit of $i$ is $1$, and otherwise it is $0$. Thus, we can recover $i$ after $O(\log n)$ rounds. After we have $i$, we can use $A_i'' + B_{k-i}''$ to update $C''_k$. 

Therefore, it remains to show how to compute the sparse Min-Plus convolution between $A'_I$ and $B'_J$ for $|I|, |J| = O(\sqrt{sn})$.

Let $F: [t] \times [t] \rightarrow [n]$ be a random function (independent to the choice of $I, J$). We then create two arrays $X, Y$ each of length $O(n)$ as follows. Initially, all entries in $X$ and $Y$ are set to some distinct values out side of $[t] \times [t]$. 
For every $i \in I$ and $y \in [t]$, we set $X_{i + F(A'_i, y)}$ to $(A'_i, y)$; for every $j \in J$ and $x \in [t]$, we set $Y_{j + n - F(x, B'_j)}$ to $(x, B'_j)$. Suppose all entries are only set at most once. Then we compute the Min-Equality convolution $Z$ between $X$ and $Y$, where we compare two pairs $(x, y)$ and $(x', y')$ by comparing $x+y$ and $x' + y'$ and breaking ties arbitrarily. Then the sum of the two integers in the pair $Z_{k+n}$ equals $$\min_{\substack{(i, j, x, y) \in I \times J \times [t] \times [t] \\ i + F(A'_i, y) + j + n - F(x, B'_j) = k + n \\ (A'_i, y) = (x, B'_j)}} (A'_i + y) = \min_{\substack{(i, j) \in I \times J \\ i + j  = k}} (A'_i + B'_j).$$
Thus, computing a \MinEqualityConv{} instance gives the result of sparse Min-Plus convolution between $A'_I$ and $B'_J$. 

We then remove the assumption that each entry of $X$ and $Y$ is only set once by standard techniques. Consider a fixed entry $X_q$. For every $(x, y) \in t$, we set $X_q$ to $(x, y)$ if and only if $q-F(x, y) \in I$ and $A_{q - F(x, y)}' = x$. Since $F(x, y)$ is sampled from $[n]$ uniformly at random, the probability that we set $X_q$ to $(x, y)$ is at most $\frac{\left| \left\{i\in I: A'_i = x \right\}\right|}{n}$. Summing over all $x, y$, the expected number of times  that $X_q$ is set is $O(\frac{|I| t}{n}) = O(\frac{t\sqrt{s}}{\sqrt{n}}) = O(1)$ since $t \sqrt{s} = O(\sqrt{n})$. Since the values of $F(x, y)$ are independent for different $(x, y)$, by Chernoff bound, we conclude that $X_q$ is set only $O(\log n)$ times with high probability. Similarly, all indices in $Y$ are set only $O(\log n)$ times with high probability. Thus, we can create $O(\log n)$ arrays $X^{(a)}$, where $X^{(a)}_p$ equals the value of $X_p$ when we attempt to set it the $a$-th time. We can similarly create $O(\log n)$ arrays $Y^{(b)}$. Then it suffices to compute the \MinEqualityConv{} between $X^{(a)}$ and $Y^{(b)}$ for every $(a, b) \in [O(\log n)]^2$. 

Overall, the running time is $\tO(\frac{n}{s} \cdot n^{1+\delta})$ as long as $t\sqrt{s} = O(\sqrt{n})$, assuming   \MinEqualityConv{} for length $n$ arrays has an $\OO(n^{1+\delta})$ time algorithm.
\end{proof}

By Claim~\ref{cl:min-equal-conv-heavy} and Claim~\ref{cl:min-equal-conv-light}, we can compute $C''$ and thus $C$ in $\OO(n^2 / \sss + s^5 \sss^4 n^2 / t + \frac{n}{s} \cdot n^{1+\delta})$ time as long as $t\sqrt{s} = O(\sqrt{n})$. We can set $t = n^{\frac{10-5\delta}{21}}$, $s = n^{\frac{1+10\delta}{21}}$ and $\sss = n^{\frac{1-11\delta}{21}}$ to get the $\OO(n^{2 - \frac{1-11\delta}{21}})$ randomized running time. 
\end{proof}

Note that in order for the proof for Theorem~\ref{thm:min-equal-conv-under-conv} to work,  the more difficult BSG covering of Theorem~\ref{thm:bsg:gower:fast} that builds on Gower's proof~\cite{Gowers01} is necessary. If we instead use the simpler BSG covering of Theorem~\ref{thm:bsg:simple}, we will have an $s^{O(1)}n^{2.5} / t$ term from Claim~\ref{cl:min-equal-conv-heavy}, which cannot give subquadratic running time considering the $t\sqrt{s} = O(\sqrt{n})$ requirement in Claim~\ref{cl:min-equal-conv-light}. 

%% file: more_counting.tex
\subsection{Exact \texorpdfstring{$k$}{k}-Clique and Minimum \texorpdfstring{$k$}{k}-Clique}

Let $G$ be the input graph of an \ExactKClique{} or \ExactKCliqueCount{} instance, and let $w$ be the weight function of the graph. For every set $I \subseteq V(G)$ of size $k-1$, we use $W_I$ to denote the set of $j$ where $I \cup \{j\}$ forms a $k$-clique whose edge weights sum up to the required value $t$. 

\begin{theorem}
\label{thm:exact-k-clique-counting}
If \ExactKClique{} for $n$-node graphs has an $O(n^{k-\eps})$ time algorithm for some $\eps > 0$, then \ExactKCliqueCount{} for $n$-node graphs  has an $O(n^{k-\eps'})$ time algorithm for some $\eps' > 0$
\end{theorem}
\begin{proof}
Similar as before, by well-known techniques \cite{focsyj}, given a \ExactKCliqueCount{} instance on a graph $G$ with $n$ nodes and a target $t$, we can use the $O(n^{k-\eps})$ time algorithm for \ExactKClique{} to list up to $n^{0.99}$ witnesses for every set $I$ of $k-1$ nodes, in $O(n^{k-\eps''})$ time for some $\eps'' > 0$. 

Then we enumerate all possible subsets $J \subseteq V(G)$ of size $k-3$. For each $J$, we can reduce the problem of counting witnesses for sets $I \supseteq J$ of size $k-1$ to a \AEExactTriCount{} instance $(G', t')$ in a standard way \cite{Nesetril1985}. The set of nodes of $G'$ corresponds to $V(G) \setminus J$. Let the edge weight between $i_1$ and $i_2$ be $w'(i_1, i_2) = 2w(i_1, i_2) + \sum_{j \in J} (w(j, i_1) + w(j, i_2))$. Also, let $t'$ be $2t-2\sum_{\substack{j_1, j_2 \in J \\ j_1 < j_2}} w(j_1, j_2)$. It is not difficult to verify that $(i_1, i_2, i_3)$ forms an exact triangle in $G'$ if and only if $J \cup \{i_1, i_2, i_3\}$ forms an exact $k$-clique in $G$. Thus, the number of witnesses of $(i_1, i_2)$ in $G'$ equals the number of witnesses of $J \cup \{i_1, i_2\}$ in $G$. We then proceed similar to the proof of Theorem~\ref{thm:exact-tri-count}. If the number of witnesses of $(i_1, i_2)$ in $G'$ is less than $n^{0.99}$, then we have already listed all of their witnesses in the $O(n^{k-\eps''})$ time step. For $(i_1, i_2)$ that has at least $n^{0.99}$ witnesses, we can find a set $S$ of size $\OO(n^{0.01})$ that intersects with each of $W_{i_1, i_2}$ in $\OO(n^{2.99})$ time, and then apply Lemma~\ref{lem:exact-tri} to compute the witness count for these $(i_1, i_2)$ pairs in $\OO(|S| \cdot n^{(3+\omega)/2}) \le O(n^{2.70})$ time. 

Finally, summing up $W_I$ for every distinct set $I$ of $k-1$ nodes gives the total exact triangle count of $G$. The overall running time for the \ExactKCliqueCount{} instance is thus $\OO(n^{k-\eps''} + n^{k-3} \cdot (n^{2.99} + n^{2.70})) = \OO(n^{k-\min\{\eps'', 0.01\}})$. 
\end{proof}

We can similarly show that \MinKCliqueCount{} reduces to \MinKClique{}. Since the proof is essentially the same, we omit the proof of the following theorem for conciseness. 
\begin{theorem}
\label{thm:min-k-clique-counting}
If \MinKClique{} for $n$-node graphs has an $O(n^{k-\eps})$ time algorithm for some $\eps > 0$, then \MinKCliqueCount{} for $n$-node graphs  has an $O(n^{k-\eps'})$ time algorithm for some $\eps' > 0$
\end{theorem}

We then reduce \MinKClique{} to \MinKCliqueCount{}.
\begin{theorem}
\label{thm:min-k-clique-counting-rev}
If \MinKCliqueCount{} for $n$-node graphs has an $O(n^{k-\eps})$ time algorithm for some $\eps > 0$, then \MinKClique{} for $n$-node graphs  has an $O(n^{k-\eps'})$ time algorithm for some $\eps' > 0$
\end{theorem}
\begin{proof}
Let $G$ be the input graph for a \MinKClique{} instance. Without loss of generality, we can assume $G$ is a $k$-partite graph on node parts $V_1 \cup \cdots \cup V_k$. 

We first multiply all the edge weights of $G$ by a large enough number $M \ge 10 k \cdot n^k$. Then for each $i \in [k]$, $v_i \in V_i$, we add $v_i \cdot n^{i-1}$ to all edges adjacent to $v_i$. After these transformations, there will be a unique minimum weight $k$-clique in the graph. Furthermore, this $k$-clique must also be a minimum weight $k$-clique in the original graph. We denote this minimum weight $k$-clique by $(u_1, \ldots, u_k) \in V_1 \times \cdots \times V_k$. 

Then for each $i \in [k]$ and each $p \in \left[\lceil \log(n) \rceil \right]$, we do the following. Let $G^{(i, p)}$ be a copy of the graph (after the weight changes), and we duplicate all nodes $v \in V_i$ whose $p$-th bit in its binary representation is $1$. 
We use the assumed \MinKCliqueCount{} algorithm the count the number of minimum weight $k$-cliques in in $G^{(i, p)}$. 
If the number of minimum weight $k$-clique in $G^{(i, p)}$ is $2$, then we know the $p$-th bit of $u_i$ is $1$; otherwise, the $p$-th bit of $u_i$ is $0$. 

After all $k \lceil \log(n) \rceil$ rounds, we can recover $(u_1, \ldots, u_k)$, and thus compute the weight of the minimum weight $k$-clique in the original graph. 
\end{proof}

\subsection{Monochromatic Convolution}

\begin{theorem}
\label{thm:mono-conv-count}
If \MonoConv{} for  length $n$ arrays has an $O(n^{1.5-\eps})$ time algorithm for some $\eps > 0$, then \MonoConvCount{} for length $n$ arrays  has an $O(n^{1.5-\eps'})$ time randomized algorithm for some $\eps' > 0$
\end{theorem}
\begin{proof}
If \MonoConv{} has an $O(n^{1.5-\eps})$ time algorithm, then \ThreeSUM{} has a truly subquadratic time algorithm~\cite{lincoln2020monochromatic}, and consequently \AllThreeSUM{} has a truly subquadratic time algorithm~\cite{focsyj}. Furthermore, by Theorem~\ref{thm:all-3sum-count}, \AllThreeSUMCount{} has an $O(n^{2-\eps''})$ time algorithm for some $\eps'' > 0$. 

Thus, it suffices to reduce \MonoConvCount{} to \AllThreeSUMCount{}. Lincoln, Polak, and Vassilevska W.~\cite{lincoln2020monochromatic} showed that a truly subquadratic time algorithm for \AllThreeSUM{} implies a truly sub-$n^{1.5}$ time algorithm for \MonoConv{}. It is not difficult to check that their reduction preserves the number of solutions, and thus works for the counting versions as well. 
\end{proof}

\subsection{All-Pairs Shortest Paths}
\label{sec:apsp-count-mod}

\begin{theorem}
\label{thm:apsp-count-mod}
If \APSP{} for $n$-node  graphs with positive edge weights has an $O(n^{3-\eps})$ time algorithm for some $\eps > 0$, then \APSPCountMod{U} for $n$-node  graphs with positive edge weights has an $O(n^{3-\eps'})$ time algorithm for some $\eps' > 0$, for any $\OO(1)$-bit integer $U \ge 2$.
\end{theorem}
\begin{proof}

Let $(A, A')$ and $(B, B')$ be two pairs of $n \times n$ matrices where the entries of $A'$ and $B'$ are $\OO(1)$-bit integers. We define two ``funny'' matrix products (one of them was defined in the proof of Theorem~\ref{thm:apsp-count}). If $(C, C') = (A, A') \oplus (B, B')$, then $C_{ij} = \min(A_{ij}, B_{ij})$, and $C'_{ij} = [A_{ij} = C_{ij}]A'_{ij}+[B_{ij} = C_{ij}]B'_{ij}$, where $[\cdot]$ denotes the indicator function. Clearly, we can compute $(A, A') \oplus (B, B')$ in $\OO(n^2)$ time.

\begin{claim}
\label{cl:apsp-count-mod}
Suppose \APSP{} on $n$-node  graphs with positive edge weights has an $O(n^{3-\eps})$ time algorithm for some $\eps > 0$, then we can compute $(A, A') \otimes (B, B')$ in $O(n^{3-\eps''})$ time for some $\eps'' > 0$ where the entries of $A'$ and $B'$ are $\OO(1)$-bit integers. 
\end{claim}
\begin{proof}
If \APSP{} on $n$-node  graphs with positive edge weights has an $O(n^{3-\eps})$ time algorithm, then so does \MinPlus{} \cite{focsyj}. By Theorem~\ref{thm:minplus-count}, \MinPlusCount{} also has a truly subcubic time algorithm. It remains to reduce computing $(A, A') \otimes (B, B')$ to \MinPlusCount{}. 

For $p \in [O(\log n)]$, let $A^{(p)}$ be the matrix where $A^{(p)}_{ij} = A_{ij}$ if the $p$-th bit of $A'_{ij}$ is $1$, and $A^{(p)}_{ij} = \infty$ otherwise. We can similarly define $B^{(p)}$. Also, let $J$ be the $n \times n$ matrix whose  entries are all $1$. It is then not difficult to verify that
$$(A, A') \otimes (B, B') = \bigoplus_{p, q} (A^{(p)}, 2^p J) \otimes (B^{(q)}, 2^q J).$$
To compute each term in the above ``sum'', say $(C^{(p, q)}, C'^{(p, q)}) = (A^{(p)}, 2^p J) \otimes (B^{(q)}, 2^q J)$, we first use the assumed \MinPlus{} algorithm to compute $C^{(p, q)} = A^{(p)} \star B^{(q)}$ in $O(n^{3-\eps})$ time. Then $C'^{(p, q)}_{ij}$ is exactly the number of witnesses for $C^{(p, q)}$ in the previous Min-Plus product, multiplied by $2^{p+q}$, so we can use the truly subcubic time algorithm for \MinPlusCount{} to compute $C'^{(p, q)}$. 
\end{proof}

As showed in the proof of Theorem~\ref{thm:apsp-count}, \APSPCount{} reduces to $O(\log n)$ instances of the funny matrix product; also, now we can mod all entries in $A', B'$ by $U$ after each funny matrix product to keep them $\OO(1)$-bit integers. 
Thus, by Claim~\ref{cl:apsp-count-mod}, \APSPCountMod{U} has an $\OO(n^{3-\eps''})$ time algorithm assuming \APSP{} can be solved in truly subcubic time. 
\end{proof}

\begin{theorem}
\label{thm:apsp-count-mod-rev}
If \APSPCountMod{c} for $n$-node  graphs with positive edge weights has an $O(n^{3-\eps})$ time algorithm for some $\eps > 0$, where $c \ge 2$ is some $\OO(1)$-bit integer, 
then \APSP{} for $n$-node  graphs with positive edge weights has an $O(n^{3-\eps'})$ time algorithm for some $\eps' > 0$, 
\end{theorem}
\begin{proof}
Suppose there is an $O(n^{3-\eps})$ time algorithm for \APSPCountMod{c} for $n$-node graphs positive edge weights, then there is also an $O(n^{3-\eps})$ time algorithm for counting the number of witnesses modulo $c$ for \MinPlus{} for $n \times n$ matrices, by following the standard reduction from \MinPlus{} to \APSP{}~\cite{focsyj}. 

Then by essentially the same proof to the proof of Theorem~\ref{thm:minplus-count-rev}, there exists an $\OO(n^{3-\eps})$ time algorithm for \MinPlus{}. Finally, there is a truly subcubic time algorithm for \APSP{} since \APSP{} and \MinPlus{}  are subcubically equivalent \cite{focsyj}. 
\end{proof}

%% file: more_bsg.tex
\section{Bounded-Difference or Monotone Min-Plus Product from the Triangle Decomposition Theorem}
\label{app:bd:diff}

In this appendix, we note that our Triangle Decomposition Theorem (Theorem~\ref{thm:tri:decompose})
implies a truly subcubic algorithm for bounded-difference Min-Plus product. 
Existing algorithms \cite{BringmannGSW16, WilliamsX20,GuPWX21, ChiDXsoda22, ChiDXZstoc22} are faster, but our presentation is simpler, and so is interesting from the pedagogical perspective in our opinion.
In a way, the Triangle Decomposition Theorem clarifies conceptually why a subcubic algorithm is possible (if we don't care too much about optimizing the exponent in the running time).

\begin{theorem}
There is a truly subcubic algorithm for computing the Min-Plus product of two integer $n\times n$
matrices $A$ and $B$ satisfying the \emph{bounded difference property}, i.e., $|A_{i,k+1}-A_{ik}|\le c_0$ for all $i,k$
and $|B_{k+1,j}-B_{kj}|\le c_0$ for all $k,j$ for some constant $c_0$.
\end{theorem}
\begin{proof}
\newcommand{\CCC}{\tilde{C}}
Let $C=A\star B$ denote the matrix that we want to compute.
Let $\ell$ be a parameter, and let $D$ be the set of all indices in $[n]$ that are divisible by $\ell$.
Let $\ell'=2c_0\ell+1$.
We first compute $\CCC_{ij}=\min_{k\in D} (\up{A_{ik}/\ell'} + \up{B_{kj}/\ell'})$ for every $i,j\in [n]$.
By brute force, this takes $O(n^3/\ell)$ time.
Note that $\CCC_{ij}\ge C_{ij}/\ell'$.
For each $r\in [\ell]$, $i,j\in[n]$ and $k \in D$,
let $A_{ik}^{(r)} = A_{i,k+r} - \up{A_{ik}/\ell'}\ell'$ and
$B_{kj}^{(r)} = B_{k+r,j} - \up{B_{kj}/\ell'}\ell'$.  Note that $A_{ik}^{(r)},B_{kj}^{(r)}\in [\pm O(c_0\ell)]$.

To compute $C=A\star B$, we use the following formula:
\begin{eqnarray*}
 C_{ij} &=& \min_{k\in D,\, r\in[\ell]} (A_{i,k+r} + B_{k+r,j})\\
        &=& \min_{\Delta\in\{0,1,2\}} \left((\CCC_{ij} + \Delta)\ell'+ \min_{r\in[\ell],\, k\in D:\, \up{A_{ik}/\ell'} + \up{B_{kj}/\ell'} = \CCC_{ij}+\Delta}
(A_{ik}^{(r)} + B_{kj}^{(r)})  \right).
\end{eqnarray*}
The reason for restricting the range of $\Delta$ to $\{0,1,2\}$ is this:
if $\up{A_{ik}/\ell'} + \up{B_{kj}/\ell'} \ge \CCC_{ij}+3$,
then $A_{ik}/\ell' + B_{kj}/\ell' \ge \CCC_{ij}+1$, and so
$A_{i,k+r}+B_{k+r,j} \ge A_{ik}+B_{kj} - 2c_0\ell > \CCC_{ij}\ell'$; but
we know that the true value of $C_{ij}$ is at most $\CCC_{ij}\ell'$.

To evaluate the above formula, let's fix $\Delta\in \{0,1,2\}$.
We want to compute 
\[ C'_{ij}\ := \min_{r\in [\ell],\, k\in D:\, \up{A_{ik}/\ell'} + \up{B_{kj}/\ell'} = \CCC_{ij}+\Delta}
(A_{ik}^{(r)} + B_{kj}^{(r)}).
\]
Initialize $C'_{ij}$ to $\infty$.
Consider the tripartite graph with nodes $\{u_i:i\in[n]\}$, $\{x_k:k\in D\}$,
and $\{v_j:j\in[n]\}$, where $u_ix_k$ has weight $\up{A_{ik}/\ell'}$, and 
$x_kv_j$ has weight $\up{B_{kj}/\ell'}$, and $u_iv_j$ has weight $-\CCC_{ij}-\Delta$.
Apply Theorem~\ref{thm:tri:decompose} to get subgraphs $G^{(\lam)}$ and a set $R$
in $\OO(n^3/s+s^2n^2)$ time.

We first examine each triangle $u_ix_kv_j\in R$ and each $r\in[\ell]$ and reset $C'_{ij}$
to $A_{ik}^{(r)} + B_{kj}^{(r)}$ if it is smaller than the current value.
This takes $\OO(\ell \cdot (n^3/\ell)/s) = \OO(n^3/s)$ time.

Next, for each $\lam$ and each $r\in[\ell]$, we compute $\min_{k\in D:\, u_ix_k,x_kv_j\in G^{(\lam)}} (A'_{ik} +B'_{kj})$ and 
reset $C'_{ij}$ to this value if it is smaller, for every $u_iv_j\in G^{(\lam)}$.  
This reduces to a Min-Plus product instance on integers in $[\pm O(c_0\ell)]$
and takes $\OO(c_0\ell n^\omega)$ time for each $\lam$ and $r$.  The total time
is $\OO(s^3\cdot \ell \cdot c_0\ell n^\omega) = \OO(c_0\ell^2 s^3 n^\omega)$.

The overall running time is $\OO(n^3/\ell + n^3/s + c_0\ell^2 s^3 n^\omega)$.
Setting $\ell=s=n^{(3-\omega)/6}$ gives a bound of $\OO(n^{(15+\omega)/6})$ for constant $c_0$
(which is improvable by using rectangular matrix multiplication).
\end{proof}

The above algorithm easily extends to the case where we allow $O(n^{2-\delta})$ exceptional pairs $(i,k)$ to
violate the bounded difference property $|A_{i,k+1}-A_{ik}|\le c_0$, and similarly  $O(n^{2-\delta})$ exceptional pairs $(k,j)$ to
violate the bounded difference property $|B_{k+1,j}-B_{kj}|\le c_0$.  We just need to add an extra cost of
$O(\ell n^{2-\delta}\cdot n)$.  The total time bound $\OO(\ell n^{3-\delta} + n^3/\ell + n^3/s + 
c_0\ell^2 s^3 n^\omega)$ remains truly subcubic for an appropriate choice of $\ell$ and $s$.

The case of matrices with monotone rows/columns and integer entries in $[n]$ easily reduce to the bounded difference case with a nonconstant $c_0=n^\delta$ and $O(n^{2-\delta})$ exceptional pairs.  So, we also get a truly subcubic
algorithm (with an appropriate choice of $\delta$ and a slightly worse final exponent) for the monotone case.

\section{More on BSG}\label{app:bsg}

In this appendix, we show that the proof of
Theorem~\ref{thm:bsg:simple} can be modified to hold not just for
indexed sets, but also for arbitrary subsets $A$ and $C$ of an abelian group, if we ignore the construction time.
This requires a more clever argument in the ``low-frequency'' case.
(The proof of Theorem~\ref{thm:bsg:gower} can also be modified in a similar way.)

\begin{theorem} {\bf (Simpler BSG Covering)}\ \ \label{thm:bsg:cover:simple:app}
Given subsets $A$ and $C$ of size $n$ of an abelian group and a parameter $s$,
there exist a collection of $\ell=\OO(s^3)$ subsets $A^{(1)},\ldots,A^{(\ell)}\subseteq A$, and a set $R$ of $\OO(n^2/s)$ pairs in $A\times A$, such that
\begin{enumerate}
\item[\rm(i)] $\{(a,b)\in A\times A: a-b\in C\}\ \subseteq\
R\,\cup\, \bigcup_\lam (A^{(\lam)}\times A^{(\lam)})$, and
\item[\rm(ii)] $\sum_\lam |A^{(\lam)} - A^{(\lam)}| = \OO(s^2 n^{3/2})$.
\end{enumerate}
\end{theorem}
\begin{proof}\

\begin{itemize}
\item {\bf Few-witnesses case.}
First add $\{(a,b)\in A\times A:\ a-b\in C,\ \pop_A(a-b)\le n/s\}$ to $R$.
The number of pairs added to $R$ is $\OO(n\cdot n/s)$.

\item {\bf Many-witnesses case.}
Let $F =\{h: \pop_A(h) > n/r\}$.
Then $|F|=O(rn)$, since the total popularity is $n^2$.
By adding extra elements to $F$, we may assume that $|F|=\Theta(rn)$.
Pick a random subset $H\subseteq F$ of size $c_0sr\log n$ for a sufficiently large constant $c_0$.

\begin{itemize}
\item {\bf Low-frequency case.}
Add the following to $R$:
\[ \{(a,b)\in A\times A:\ |\{(a',b')\in A\times A: a-a'=b-b'\not\in F\}| > n/(2s)\}.
\]
Since for each $(a,a')$ with $a-a'\not\in F$, there are at most $n/r$ choices of $(b,b')$ satisfying $a-a'=b-b'$,
the number of pairs added to $R$ is $\OO(\frac{n^2\cdot n/r}{n/(2s)})=\OO(n^2/s)$
by choosing $r:=s^2$.
\item {\bf High-frequency case.}
For each $h\in H$,
add the following subset to the collection:
\[ A^{(h)} = \{a\in A:\ a-h\in A\}.
\]
The number of subsets is $\OO(sr)=\OO(s^3)$.
\end{itemize}
\end{itemize}

\emph{Correctness.}
To verify (i), consider a fixed pair $(a,b)\in A\times A$ with $a-b\in C$.
If $\pop_A(a-b)\le n/s$, then $(a,b)\in R$ due to
the ``few-witnesses'' case.
So assume $\pop_A(a-b)>n/s$.
Furthermore, assume that $|\{(a',b')\in A\times A: a-a'=b-b'\not\in F\}| \le n/(2s)$,
for otherwise $(a,b)\in R$ due to the ``low-frequency'' case.
There are at least $n/s$ pairs $(a',b')\in A\times A$ with $a'-b'=a-b$,
which also satisfy $a-a'=b-b'$ by Fredman's trick.  Among them,
there are at least $n/(2s)$ pairs $(a',b')\in A\times A$ with $a-a'=b-b'\in F$.
For $h=a-a'=b-b'$, we have $(a,b)\in A^{(h)}\times A^{(h)}$.
So, the probability that $(a,b)\in A^{(h)}\times A^{(h)}$ for a random $h\in F$
is $\Omega(\frac{n/(2s)}{rn})=\Omega(1/(sr))$.
Thus, $(a,b)\in \bigcup_{h\in H} (A^{(h)}\times A^{(h)})$ w.h.p.\ for
a random subset $H\subset F$ of size $c_0sr\log n$.

To verify (ii), consider a fixed $h$.
The set $A^{(h)}-A^{(h)}$ is equal to
$\{c:\ \exists a,b,a',b'\in A,\ c=a-b,\ a-h=a',\ b-h=b'\}$.
Let $Y_c = \{a\in A:\ a-c\in A\}$.
Then $A^{(h)}-A^{(h)}$ is contained in
$\{c: \exists a\in Y_c\ \mbox{and}\ a-h\in Y_c\}$.
The expected size of $A^{(h)}-A^{(h)}$ for a random $h\in F$ is thus bounded by
\[ \sum_c \min\{ |Y_c|\cdot |Y_c|/|F|,\ 1 \}\ \le\ O(n^2/\sqrt{|F|})
\ =\ O(n^{3/2}/\sqrt{r}),
\]
since $\sum_c |Y_c| = O(n^2)$.

The expected total sum $\sum_{h\in H}  |A^{(h)}-A^{(h)}|$ is bounded by $\OO(sr\cdot n^{3/2}/\sqrt{r}) = \OO(s^2n^{3/2})$.
\end{proof}

Chan and Lewenstein~\cite{ChanLewenstein} posed the following interesting combinatorial question:
\begin{quote}
    Given sets $A,B,C$ in an abelian group of size $n$,
    we want to cover $\{(a,b)\in A\times B: a+b\in C\}$ by bicliques $A^{(\lam)}\times B^{(\lam)}$, so as to minimize the cost function $\sum_\lam |A^{(\lam)}+B^{(\lam)}|$.
    Prove worst-case bounds on the minimum cost as a function of $n$.
\end{quote}
They observed that Theorem~\ref{thm:BSG0:cover} implies an $O(n^{13/7})$ upper bound for this problem.  Theorem~\ref{thm:bsg:cover:simple:app} implies an improved $\OO(n^{11/6})$ upper bound
(as we can use ``singleton'' bicliques to cover $R$
and choose $s$ to minimize $\OO(n^2/s+s^2n^{3/2})$).

If we just want an analog of the BSG Theorem that extracts a single subset rather than constructs a cover  (in the style of Theorem~\ref{thm:BSG0}), the proof of the above theorem  becomes even simpler (the ``few-witnesses'' and ``low-frequency'' cases may be skipped):

\begin{theorem}\label{thm:bsg:simpler} {\bf (Simpler BSG)}\ \ 
Given subsets $A$ and $C$ of size $n$ of an abelian group and a parameter $s$,
if $|\{(a,b)\in A\times A: a-b\in C\}|\ge n^2/s$,
then there exists a subset $A'\subseteq A$ of size $\Omega(n/s)$,
such that 
\[ |A'-A'| \:=\: O(s^{1/2}n^{3/2}).\]
\end{theorem}
\begin{proof}
Let $F =\{x: \pop_A(x) > n/(2s)\}$.  
Then $|F|\ge n/(2s)$, because otherwise,
$\sum_{c\in C}\pop_A(c) < |F|n + n^2/(2s) < n^2/s$, contradicting the
stated assumption.

Pick a random $h\in F$ and let $A^{(h)} = \{a\in A:\ a-h\in A\}$ as before.
Then $|A^{(h)}|=\pop_A(h) > n/(2s)$.

By the same argument as before,
the expected size of $A^{(h)}-A^{(h)}$ for a random $h\in F$ is 
bounded by $O(n^2/\sqrt{|F|})$, which is now $O(s^{1/2}n^{3/2})$.
\end{proof}

The above $O(s^{1/2}n^{3/2})$ bound is better than the
known $O(s^5n)$ bound (Theorem~\ref{thm:BSG0}) when
$s\gg n^{1/9}$, though the latter holds only for the bichromatic sum set setting.\footnote{
Bounds of the form $O(s^{O(1)}n)$ were also known in the setting of monochromatic difference sets, but direct comparisons require care: many previous work such as \cite{Gowers01,Schoen} started from a related but different assumption, namely, that
the \emph{energy} $|\{(a,b,a',b')\in A\times A\times A\times A: a+b=a'+b'\}|$
is at least $n^3/s'$ for some parameter $s'$.  This parameter $s'$ is not identical to the parameter $s$ in Theorem~\ref{thm:bsg:simpler}, though they are related polynomially (or quadratically).
}

Along the same lines, we can obtain the following combinatorial result from the Triangle Decomposition Theorem, which might be of independent interest:

\begin{theorem}
Given a real-weighted tripartite graph $G$ with $n$ nodes, and given a parameter $s$,
if $G$ contains $\Omega(n^3/s)$ zero-weight triangles,
then there exists a subgraph $G'$ with $\OOmega(n^3/s^4)$ triangles (and thus
$\OOmega(n^2/s^4)$ edges)
such that all triangles in $G'$ are zero-weight triangles.
\end{theorem}
\begin{proof}
Apply Theorem~\ref{thm:tri:decompose} with $s$ replaced by $s'$.
The size of $R$ is $\OO(n^3/s')$, which can be made less than $n^3/(2s)$ for some choice of  $s'=\widetilde{\Theta}(s)$.
Then there exists $G^{(\lam)}$
with $|\Triangles(G^{(\lam)})|=\OOmega((n^3/s)\cdot (1/s^3))$.
\hspace*{\fill}
\end{proof}